\documentclass{article}

\usepackage{graphicx} % Required for inserting images

\usepackage{hyperref}
\usepackage{euscript}
\usepackage{amsmath}
\usepackage{amsthm}
\usepackage{epsfig}
\usepackage{amssymb}\usepackage{amscd}
\usepackage{braket}
\usepackage{epic}
\usepackage{color}
\usepackage{enumitem}
\usepackage{nicematrix} % for matrix indexing %\usepackage{natbib}
\usepackage{biblatex}
\usepackage{tikz-cd} % comm diagram
\usepackage{tabularx} % tabular
\usepackage{authblk} % affiliations
\usepackage{subcaption}
\usepackage[normalem]{ulem}

\usepackage{comment}

\usepackage{dsfont} % mathds

\DefineBibliographyStrings{english}{
  page = {},
  pages = {}
} %Suppress p. in references

\numberwithin{equation}{section}
\numberwithin{equation}{subsection}

\theoremstyle{plain}

\newtheorem{lemma}[equation]{Lemma}

\newtheorem{thm}[equation]{Theorem}
\newtheorem{cor}[equation]{Corollary}
\newtheorem{lem}[equation]{Lemma}
\newtheorem{prop}[equation]{Proposition}

\theoremstyle{definition}
\newtheorem{example}[equation]{Example}
\newtheorem{remark}[equation]{Remark}

\newtheorem{ex}[equation]{Example}

\newtheorem{rem}[equation]{Remark}

\numberwithin{equation}{section}
\numberwithin{equation}{subsection}

\newcommand{\bnen}{\begin{equation}}
\newcommand{\eden}{\end{equation}}
\newcommand{\bean}{\begin{eqnarray}}
\newcommand{\eean}{\end{eqnarray}}
\newcommand{\bsen}{\begin{subequations}}
\newcommand{\esen}{\end{subequations}}
\newcommand{\bea}{\begin{eqnarray*}}
\newcommand{\eea}{\end{eqnarray*}}
\newcommand{\bne}{\begin{equation*}}
\newcommand{\ede}{\end{equation*}}

\def\C{\mathbb C}

\def\R{\mathbb R}
\def\Z{\mathbb Z}

\def\N{\mathbb N}

\def\P{\mathbb P}

\title{Multifold degeneracy points of quantum systems and singularities of matrix varieties}

\author[1]{Gy\"orgy Frank}
\author[1,2]{Andr\'as P\'alyi}
\author[2]{Gerg\H{o} Pint\'er}
\author[1,3]{D\'aniel Varjas}

\affil[1]{Department of Theoretical Physics, 
    Institute of Physics, 
    Budapest University of Technology and Economics, M\H{u}egyetem rkp.~3., H-111 Budapest, Hungary}
    \affil[2]{HUN-REN-BME-BCE Quantum Technology Research Group, Budapest University of Technology and Economics, M\H{u}egyetem rkp.~3., H-111 Budapest, Hungary}
\affil[3]{IFW Dresden and Würzburg-Dresden Cluster of Excellence ct.qmat, Helmholtzstrasse 20, 01069 Dresden, Germany}

\date{\today}

\begin{document}

\maketitle

\begin{abstract}

Parameter-dependent quantum systems often exhibit energy degeneracy points, whose comprehensive description naturally lead to the application of methods from singularity theory.
A prime
example is an electronic band structure where two energy levels coincide in a point of momentum
space. It may happen, and this case is the focus of our work, that three or more levels coincide at a parameter point, called multifold
degeneracy. Upon a generic perturbation, such a multifold degeneracy point is dissolved into a
set of Weyl points, that is, generic two-fold degeneracy points. In this work, we provide an upper
bound to the number of Weyl points born from the multifold degeneracy point. To compute this upper bound, we describe the geometric degeneracy variety in the space of complex matrices. We compute its multiplicity at certain singular points corresponding to a multifold degeneracy, and the multiplicity of holomorphic map germs with respect to this variety.
Our work covers physics and mathematics aspects in detail, and attempts to bridge the two disciplines and communities. For self-containedness, we
survey examples of multi-fold degeneracies in quantum systems and condensed-matter physics,
as well as the established tools of local algebraic geometry that we use to identify the upper
bound.    
\end{abstract}

\tableofcontents

\section{Introduction}

In this paper, we address a physically motivated question using the tools of local algebraic geometry. 
In short, the question is: given a multifold degeneracy point in the parameter space of a parameter-dependent Hermitian matrix, how many generic twofold degeneracy points are born from this multifold degeneracy point upon a perturbation?
Our work covers physics and mathematics aspects in detail, and attempts to bridge the two disciplines and communities.
To make this paper self-contained, a significant part is devoted to survey examples of multi-fold degeneracies in quantum systems and condensed-matter physics, as well as the established tools of local algebraic geometry that we use to address the above question.
%the established tools from algebraic geometry, which we use to address multifold degeneracies.
%\notegergo{Bridge, partially survey, available for both discipline.}

\subsection{The physical setup}  

In quantum mechanics and condensed matter physics, several systems are described by a parameter-dependent Hamiltonian matrix, which is a smooth map of a manifold $M$ to the space $\text{Herm}(n)$ of $n \times n$ Hermitian matrices. 
A point $p$ of $M$ is called degeneracy point if the corresponding Hermitian matrix $H(p)$ has coincident eigenvalues --  in other words, the energy levels cross each other.

An example in elementary quantum physics is a spin-$1/2$ object (e.g., a localized electron) in a magnetic field \cite{Griffiths2018}, described by the Hamiltonian $H(\mathbf{B}) = B_x \sigma_x + B_y \sigma_y + B_z \sigma_z$, expressed with the $2\times 2$ Pauli matrices $\sigma_x$, $\sigma_y$, $\sigma_z$, depending on the magnetic field vector  $\mathbf{B}=(B_x,B_y,B_z)$.
In this example, the manifold $M = \mathbb{R}^3$ is the space of magnetic field vectors, $\mathbf{B} $, and $n=2$. Since the two eigenvalues (energy levels) of $H(\mathbf{B})$ are $\lambda_{1,2}=\pm \| \mathbf{B} \|$, the only degeneracy point is $\mathbf{B}=0$.

An example in condensed matter physics is the electronic band structure~\cite{AshcroftMermin} of a three-dimensional Weyl semimetal \cite{ShuangJiaNatMat,Armitage2018}. 
Such a band structure is often calculated from an $n$-orbital tight-binding model, which is formulated as an $n\times n$ Hamiltonian matrix that depends on wave vector of the electron. 
In this context, the wave vector is taken from the first Brillouin zone, which is a three-dimensional torus, that is, $M = T^3$.
In Weyl semimetals, neighboring eigenvalues of the Hamiltonian become degenerate in a few isolated points of the Brillouin zone. 
%In particular, the Weyl semimetals are characterized by the presence of certain type of degeneracy points (namely, the topmost valence band and the lowermost conduction band touch, that is, become degenerate in a single point of the Brillouin zone $M$).  
These isolated degeneracy points contribute to exotic conduction properties such as the anomalous Hall effect and the chiral anomaly \cite{ShuangJiaNatMat,Armitage2018}, and are related to the appearance of peculiar surface states and surface Fermi arcs\cite{Huang2015,SuYangXu}.

%In Weyl semimetals, the topmost valence band and the lowermost conduction band touch, that is, become degenerate in a single point of the Brillouin zone. 
%These degeneracy points define physical linear response functions through their effect on the energy-dependent electronic density of states, contribute to electronic transport effects such as the anomalous Hall effect and the chiral anomaly \cite{Armitage2018}, and are related to the appearance of peculiar surface states and surface Fermi arcs\cite{Huang2015,SuYangXu}.

\begin{figure}[ht!]
    \centering
        \includegraphics[width=\linewidth]{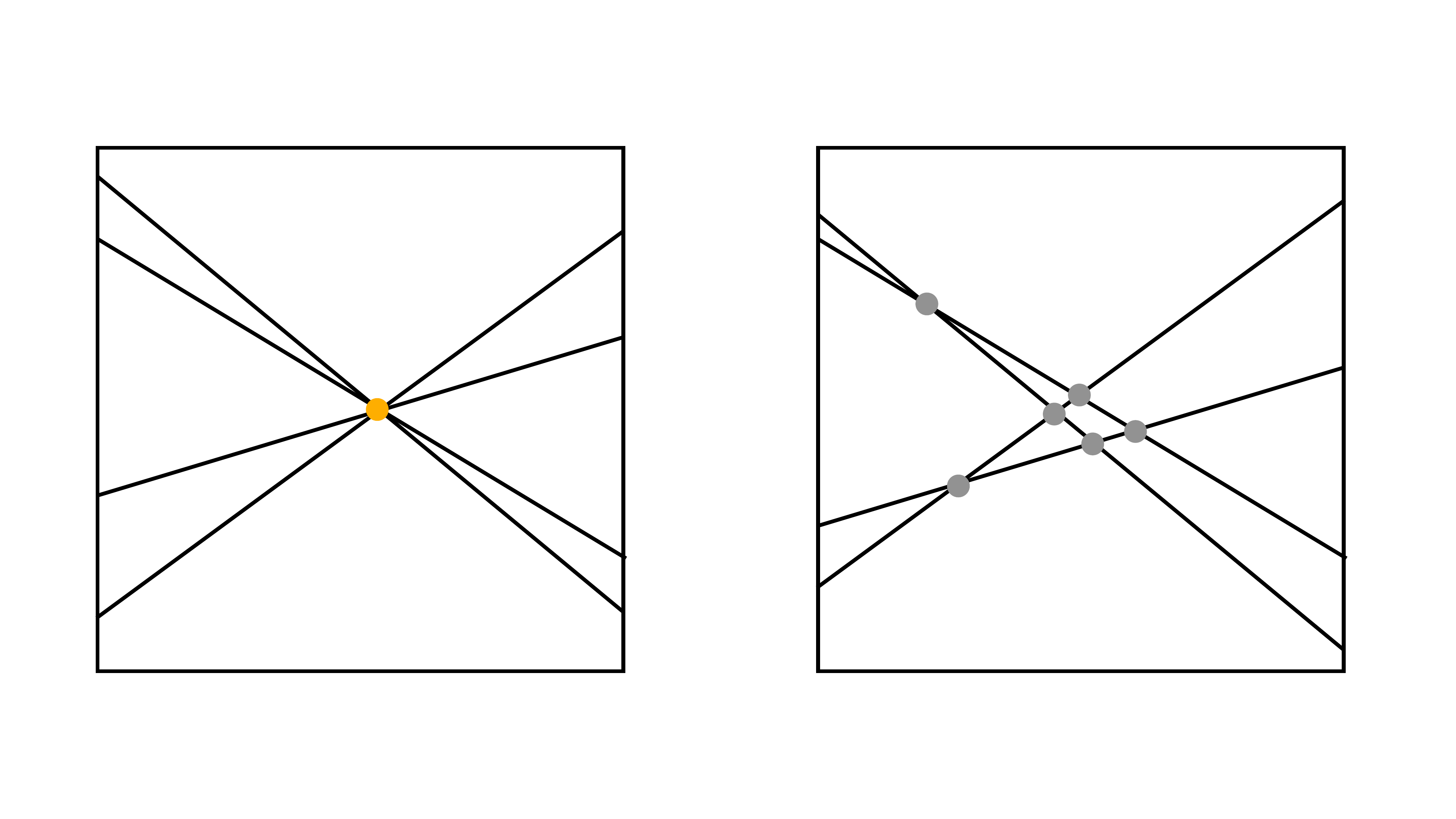}
        % \caption{First figure caption}
        % \label{fig:first}
    \caption{Fourfold degeneracy point of a 1-parameter linear family of $4 \times 4$ real diagonal matrices (left) and its generic perturbation (right). The lines are the graphs of the four eigenvalues, as functions of the parameter (horizontal line). For 1-parameter diagonal families $H(x)=\text{diag}(a_1x, \dots, a_k x) \in \text{Diag}_{\R}(k)$, $a_j$ are pairwise different, 0 is a $k$-fold degeneracy point with $H(0)=0$. A generic perturbation $H_t(x)=\text{diag}(a_jx +b_jt)$ has $\sharp \mathbf{WP}=\binom{k}{2}$ generic twofold degeneracy points for $t \neq 0$. See the diagonal case in Section~\ref{ss:proofdiag}.}
    \label{fig:both}
\end{figure}

\begin{figure}[ht!]
    \centering
    %\hspace*{5mm}
    \begin{subfigure}[t]{0.36\columnwidth}
        \centering
        \includegraphics[width=\linewidth]{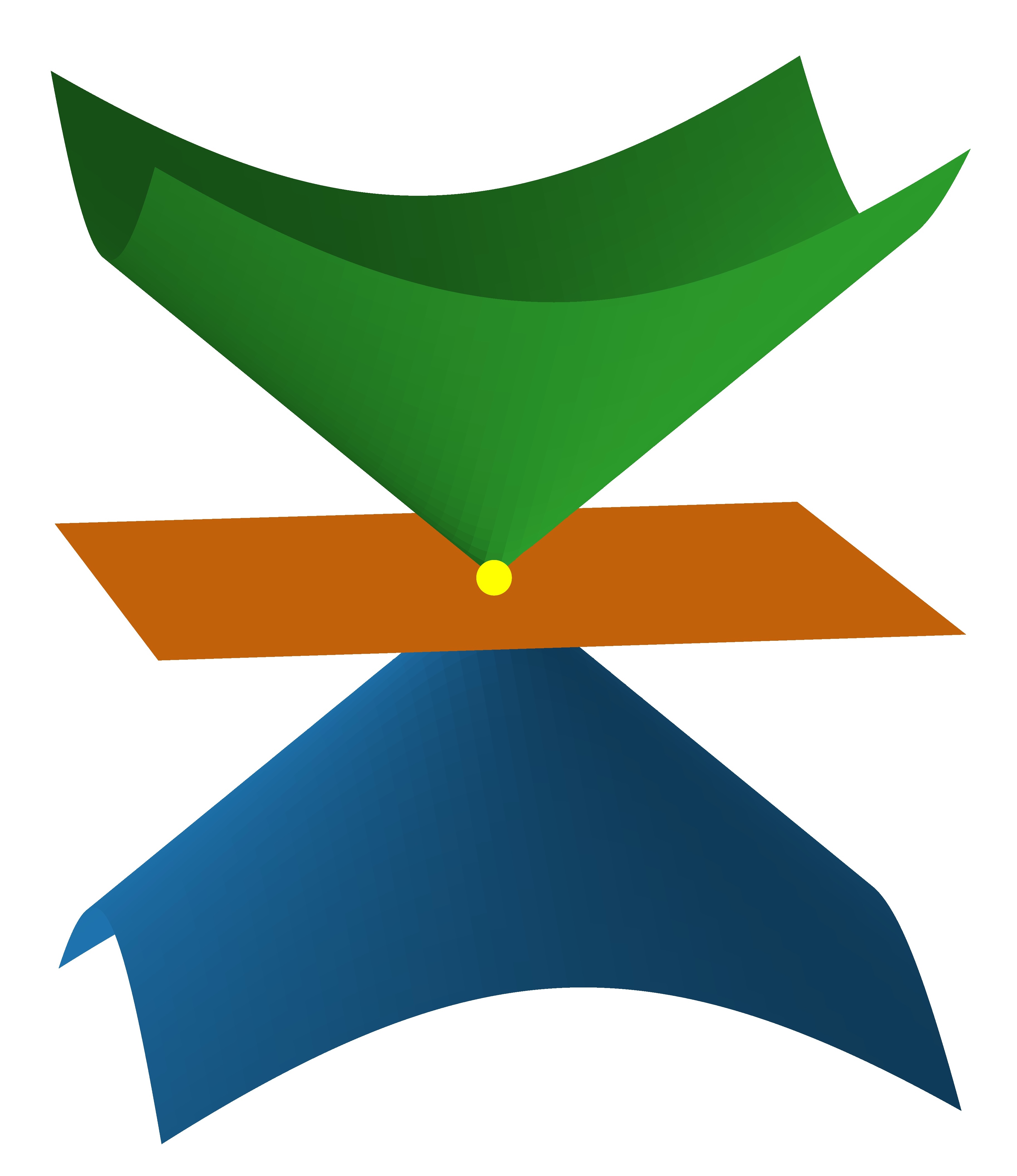}
        % \caption{First figure caption}
        % \label{fig:first}
    \end{subfigure}%
    %\hfill
    \hspace*{2cm}
    \begin{subfigure}[t]{0.36\columnwidth}
        \centering
        \includegraphics[width=\linewidth]{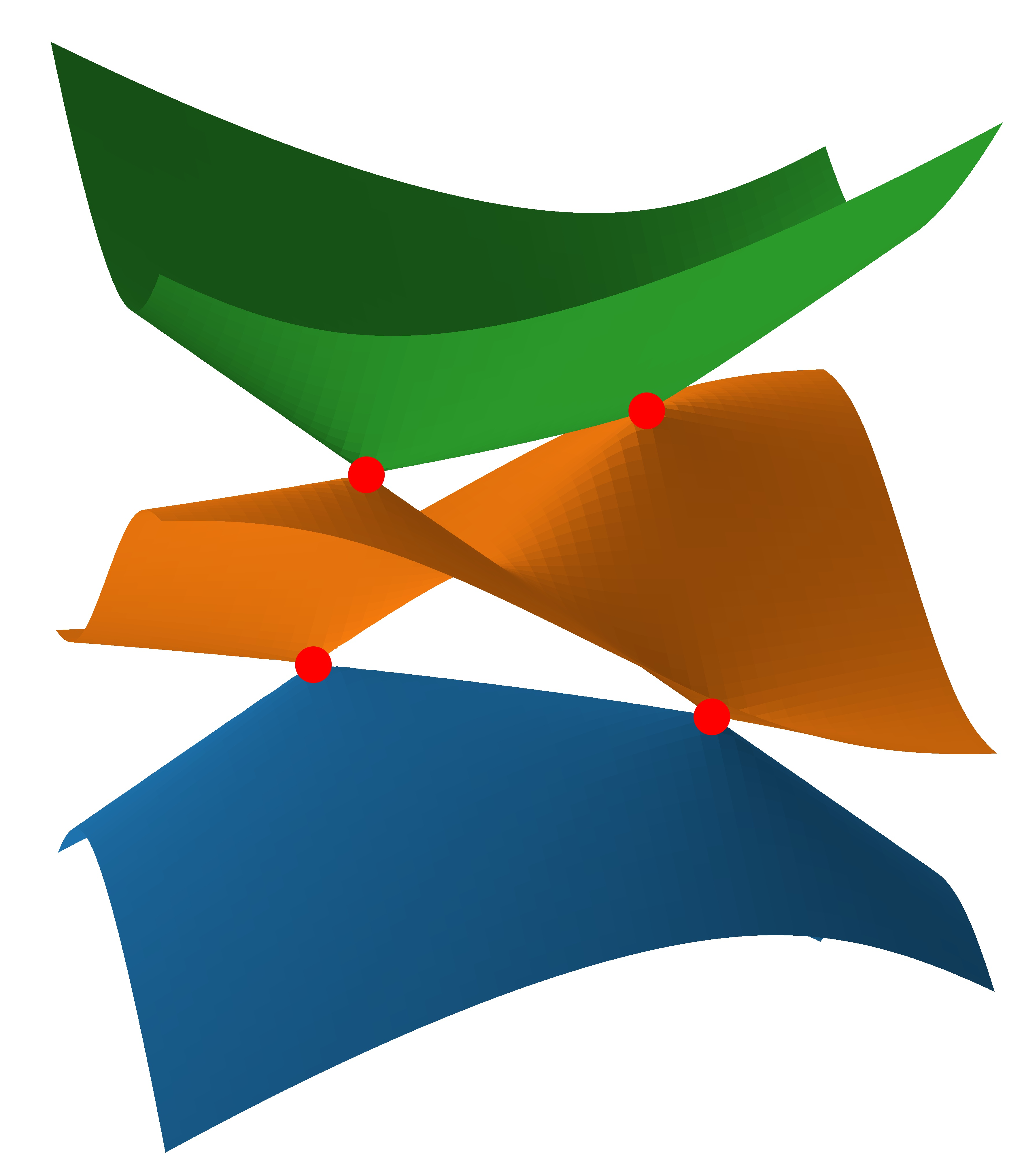}
        % \caption{Second figure caption}
        % \label{fig:second}
    \end{subfigure}
    \caption{Threefold degeneracy point of a 2-parameter linear family of $3 \times 3$ real symmetric matrices  (left) and a generic perturbation of it (right). The surfaces are the graphs of the three eigenvalues, as functions of the parameters (horizontal plane). 
    A 2-parameter symmetric family $H(x,y)=x A + y B \in \text{Symm}_{\R}(k)$ has a  $k$-fold degeneracy at $0$ with $H(0)=0$. For a generic choice of $A$ and $B$, 0 is an isolated degeneracy point, moreover, it is isolated also in the complex sense.  For a generic perturbation $H_t(x,y)=xA+yB+tC$, the inequality $\sharp \mathbf{WP} \leq \sharp \mathbf{cWP}=\frac{k(k^2-1)}{6}$ holds for the number of Weyl points and complex Weyl points, if $t \neq 0$. See the symmetric case in Section~\ref{ss:symmpart1}.}
    \label{fig:both2}
\end{figure}

In the above examples, the parameter space $M$ is three-dimensional.
Whenever the parameter space $M$ is three-dimensional, the generic degeneracy points are called Weyl points; they are isolated 2-fold degeneracy points (i.e., two eigenvalues coincide)\footnote{In the physics literature, this terminology is ambiguous: many articles refer to any isolated two-fold degeneracy point as a `Weyl point', not only to the generic degeneracy points.}. Multifold degeneracy points may appear due to the symmetries of the system. 
Isolated multifold degeneracy points appear for example in the case of a spinful particle with a spin higher than $1/2$, placed in a magnetic field~\cite{Scherubl2019,Frank2020,FrankTeleportation,Frank2022}, or in the electronic band structure of a crystalline material with certain lattice symmetries \cite{BradlynScience,Alpin,Robredo,encyclopedia}.
A generic symmetry breaking perturbation splits the multifold degeneracy point into Weyl points  ($ \mathbf{WP}$). 
Our main target quantity is the number of these Weyl points,  $\sharp \mathbf{WP}$. 
Of course, this number may depend on the choice of the perturbation.

A lower bound for $\sharp \mathbf{WP}$ can be given in terms of the first Chern numbers of the corresponding eigenvector bundles \cite{Bruno2006}, hence, the main goal of this paper is to give an upper bound for $\sharp \mathbf{WP}$.
We provide this upper bound using the toolkit of local algebraic geometry. 
This upper bound is interpreted as the number of complex Weyl points $\sharp \mathbf{cWP}$ of the perturbation, and we find that it does not depend on the choice of the perturbation. An important special case is a linear Hamiltonian $H: \R^3 \to \text{Herm}(k)$, which has a $k$-fold degeneracy at $0$, $H(0)=0$. If this degeneracy point is \emph{isolated also in complex sense} (defined in, e.g., Sec.~\ref{sec:summaryofresults}), then we show that $\sharp \mathbf{WP} \leq \sharp \mathbf{cWP} =k^2(k^2-1)/12$ holds for the number of real and complex Weyl points.

We have carried out a similar investigation for non-generic two-fold degeneracy points in \cite{BirthQ}; the present work is a generalization to multifold degeneracy points. We also studied degeneracy points in several different physical systems, such as interacting spins in external magnetic fields \cite{Scherubl2019,Frank2020,FrankTeleportation,Frank2022}, Weyl--Josephson superconducting circuits in quantum mechanics \cite{FrankTeleportation,frank2023singularity}, and ball and spring systems in classical mechanics \cite{GubaWeyl}. 

While most of these examples are described by a hermitian family depending on 3 parameters, this latter one \cite{GubaWeyl} is described by a 2-parameter family of real symmetric matrices. Motivated by this, here we also study real symmetric matrix families simultaneously with the hermitian case, which is our main focus. We also present our arguments for real diagonal matrix families depending on 1 parameter. This serves as a pedagogical illustration of the concepts, relations, and methods. 

Figure~\ref{fig:both} illustrates the dissolution of a four-fold degeneracy point (orange, left panel) into six generic twofold degeneracy points (gray, right panel) by a perturbation, for a 1-parameter diagonal matrix family. 
The degeneracy points correspond to the intersections of the graph of eigenvalues (dispersion relation, black lines).
Figure~\ref{fig:both2} is a similar illustration, showing the dissolution of a three-fold degeneracy point (yellow, left panel) into four generic two-fold degeneracy points (red, right panel) by a perturbation, for a 2-parameter real symmetric matrix family. 
Here, the degeneracy points correspond to the intersections of the graph of eigenvalues shown as the colored surfaces.

%\noteandras{ide jojjon hogy eddig hermitikusrol volt szo, meg egy pelda a Weyl josephson, hivatkozzuk magunkat. masreszt a valos szimmetrikus eset is relevans, GubaZoli. emlitsuk meg hogy elo fog jonni a diagonalis eset is (habar ahhoz lehet h nem tudunk direkt fizikai peldat tarsitani?)}

\subsection{Key mathematical insights}

For the characterization of degeneracy points first we describe the set of matrices with coincident eigenvalues inside the space $\text{Herm}(n)$ of $n \times n$ hermitian matrices and the space $\C^{n \times n}$ of all $n \times n$ complex matrices, then we characterize the intersection points of the Hamiltonian map $H$  (and its complexification) with these sets.

Our method is fundamentally based on the following observation: (1) $\lambda \in \R$ is a degenerate eigenvalue of an $n \times n$ hermitian matrix $A \in \text{Herm}(n)$ if and only if (2) the rank of $A-\lambda \mathds{1}_n$ (where $\mathds{1}_n$ is the $n \times n$ identity matrix) is at most $n-2$, and this happens if and only if (3) all the $(n-1) \times (n-1)$ minors $M_{ij}(A-\lambda \mathds{1}_n)$ (obtained by removing the $i$-th row and $j$-th column) are 0. For an arbitrary $n \times n$ complex matrix $A \in \C^{n \times n}$ and $\lambda \in \C$, the conditions (2) and (3) are still equivalent, and they are equivalent with (1') $\lambda$ is a \emph{geometrically degenerate} eigenvalue of $A$, that is, the dimension of the corresponding eigenspace is at least 2. 

The \emph{(hermitian) degeneracy variety} $\Sigma_{\text{herm}} \subset \text{Herm}(n)$ of matrices with coincident eigenvalues is a real algebraic subset of codimension 3. This variety $\Sigma_{\text{herm}}$ is fundamental in the characterization of degeneracy points, see below. 
The above description suggests that the relevant complex counterpart of $\Sigma_{\text{herm}} \subset \text{Herm}(n)$ is the \emph{geometric degeneracy variety} $\Sigma \subset \C^{n \times n}$ that consists of all complex matrices with at least one eigenvalue with at least two dimensional eigenspace. $\Sigma$ is a complex algebraic variety of codimension 3. Its vanishing ideal $I(\Sigma)$ was partially described in  \cite{DomokosHermitian}, where it is also proved that $\Sigma$ is the complexification of $\Sigma_{\text{herm}}$ (that is, $\Sigma$ is the smallest complex algebraic subset containing $\Sigma_{\text{herm}}$).

Recall that $p_0 \in M^3$ is a degeneracy point of a parameter-dependent Hamiltonian $H: M^3 \to \text{Herm}(n)$ if $A_0:=H(p_0) \in \Sigma_{\text{herm}}$.  In particular, the matrices $H(p_0)$ corresponding to Weyl points are non-singular points of $\Sigma_{\text{herm}}$ with transverse intersection of the map $H$ and  $\Sigma_{\text{herm}}$ at $H(p_0)$, while the matrices $H(p_0)$ corresponding to multifold degeneracy points are singular points of $\Sigma_{\text{herm}}$, see e.g. \cite{PinterSW, ArnoldSelMath1995}.

 Locally, by introducing a local chart in $M^3$ around $p_0$, we consider the Hamiltonian germ $H: (\R^3, 0) \to (\text{Herm}(n), A_0)$. The set of \emph{Weyl points}  (generic degeneracy points) $\mathbf{WP}$ of a generic perturbation $H_t$ of $H$ born from $p_0$ consists of the points $p \in \R^3$ close to $p_0=0$ with $H_t(p) \in \Sigma$. Although their number $\sharp \mathbf{WP}$ depends on the perturbation, the number of complex solutions does not depend on it. That is, considering the complexification $f=H_{\C}:  (\C^3, 0) \to (\C^{n \times n}, A_0)$ and the corresponding generic perturbation $f_t$, the set  $\mathbf{cWP}$ of \emph{complex Weyl points} consists of $p \in \C^3$ close to 0 with $f_t(p) \in \Sigma$. Obviously $\mathbf{WP} \subset \mathbf{cWP}$, hence, $\sharp \mathbf{WP} \leq \sharp \mathbf{cWP}$, and we will show that $\sharp \mathbf{cWP} $ does not depend on the perturbation, under the additional hypothesis that $0$ is an isolated degeneracy point in the complex sense, that is, $f^{-1}(\Sigma)=\{0\}$.

To determine $\sharp \mathbf{cWP} $, a natural idea from singularity theory is to compute the multiplicity of  $H_{\C}$ with respect to $\Sigma$ at $H(p_0)=A_0$, see Appendix~\ref{ss:intsect}. 
However, the vanishing ideal of $\Sigma$ is not well-understood in the literature, moreover, it will be pointed out that $\Sigma$ is not Cohen--Macaulay. This means, informally, that it does not behave nicely under perturbation. Instead we use the  variety $\Sigma' \subset \C^{n \times n}$ of matrices of rank at most $n-2$, and we reduce the investigation of $\Sigma$ to $\Sigma'$ by considering the minors $M_{ij}(f-\lambda \mathds{1})$, according to the description at the beginning of this subsection. This reduction has two purposes: (1) $\Sigma'$ is a determinantal variety, i.e. it has a nice vanishing ideal generated by the $(n-1) \times (n-1)$ minors, which is convenient to work with, (2) in particular, $\Sigma'$ is Cohen--Macaulay, hence it behaves well under perturbation. Determinantal varieties are well-understood, see e.g.~\cite{BrunsVetterBook}, and
also their multiplicities are known in general \cite{HarrisTu,HerzogTrung}.

\subsection{Preliminary example: spin-1 Hamiltonian}\label{ex:spin1}

    For $s=1$ the spin Hamiltian (discussed in Section~\ref{ss:exspin}) up to linear scaling of $x,y,z$ is
    \begin{equation}
        H(x,y,z)=\begin{pmatrix}
z&x-iy&0\\
x+iy&0&x-iy\\
0&x+iy&-z
\end{pmatrix}.
    \end{equation}
It is linear and it has a 3-fold degeneracy at $(x,y,z)=(0,0,0)$, $H(0,0,0)=0$. 

For the upper bound $\sharp \mathbf{cWP}$ of the number of Weyl points of a generic perturbation consider the ideal $J=I_2(H-\lambda \mathds{1}) \subset \mathcal{O}_4$ generated by the $2 \times 2$ minors $M_{ij}(H(x,y,z)-\lambda \mathds{1})$ of  $H(x,y,z)-\lambda \mathds{1}$ (see Section~\ref{ss:compweyl}). A basis of the quotient algebra $\mathcal{O}_4/J$ is formed by the 6 residue classes $[1]$, $[x]$, $[y]$, $[z]$, $[\lambda ]$, $[x^2]=[y^2]=[z^2]/2=[\lambda^2]/2$, verifying that $\sharp \mathbf{cWP}=\dim(\mathcal{O}_4/J)=3^2 \cdot (3^2-1)/12=6$.

 By Lemma~\ref{le:chernspin},  the Chern numbers of the three bands bands are $c_1(\eta_{-1})=-2$, $c_1(\eta_{0})=0$, $c_1(\eta_{1})=2$ (see Section~\ref{ss:chern} for definitions). Hence, the lower bound $\sharp_{\text{alg}} \mathbf{WP}$ for the number of real Weyl points is 4, according to Corollary~\ref{co:spinbounds}. Therefore, the possible number of real Weyl points is 4 or 6, since the parity cannot change by Proposition~\ref{pr:parity}.

For a (generic) perturbation $H_t$ of $H$ the (real and complex) Weyl points can be found as the solutions of the system of equations $M_{ij}(H_t(x,y,z)-\lambda \mathds{1})=0$. In particular, for
\begin{equation}
H_t(x,y,z)=\begin{pmatrix}
z&x-iy&0\\
x+iy&0&x-iy\\
0&x+iy&-z
\end{pmatrix}
+\begin{pmatrix}
0&t&0\\
t&0&-t\\
0&-t&0
\end{pmatrix}
\end{equation}
the 4 real solutions (i.e. the Weyl points $\mathbf{WP}$) are 
\begin{itemize}
    \item $(x,y,z)=(t, 0, \pm \sqrt{2}t)$ with degenerate eigenvalue $\lambda=-z=\mp \sqrt{2}t$,
    \item $(x, y, z)=(-t, 0, \pm \sqrt{2} t) $  with degenerate eigenvalue $\lambda=z=\pm \sqrt{2}t$.
\end{itemize}
 There are 2 more complex Weyl points $(x,y,z)=(0, \pm i, 0) $ with geometric degenerate eigenvalue $\lambda=0$. We verify that $\sharp_{\text{alg}} \mathbf{WP}=\sharp \mathbf{WP} =4 <6=\sharp \mathbf{cWP}$. This perturbation makes the lower bound sharp.

A more complicated perturbation 
\begin{equation}
H_t(x,y,z)=\begin{pmatrix}
z&x-iy&0\\
x+iy&0&x-iy\\
0&x+iy&-z
\end{pmatrix}
+\begin{pmatrix}
0&0&t\\
0&-t+t^3-2tz&0\\
t&0&0
\end{pmatrix}
\end{equation}
has 6 real Weyl points $\mathbf{WP}$: 
\begin{itemize}
    \item $(x,y,z)=(\pm t^2, 0,0)$, $\lambda=-t$,
    \item $(x,y,z)=(0, \pm t\sqrt{2-t^2},0)$, $\lambda=t$,
    \item $x=y=0$, \begin{equation}
        z= t^2 \left( \frac{2-2t^2 \pm \sqrt{2+t^2}}{1-4t^2}
        \right), \ \lambda=-t+t^3-2tz=
        -t \left( \frac{1-t^2 \pm 2t^2 \sqrt{2+t^2}}{1-4t^2}
        \right),
    \end{equation}
\end{itemize}
 Hence, $\sharp \mathbf{WP} =\sharp \mathbf{cWP}=6$ holds in this case, making the upper bound sharp.

\subsection{Summary of results} 

We summarize the results of the paper in more details.
\label{sec:summaryofresults}

For map germs $f:(\C^3, 0) \to (\C^{n \times n}, A_0)$, $A_0 \in \Sigma$ with $f^{-1}(\Sigma, A_0)=\{0\}$, we consider the number of pre-images (`complex Weyl points') $f_t^{-1}(\Sigma,A_0)$ born from $0$ of a generic perturbation\footnote{Throughout the paper we use the same notation for a map germ or set germ and a sufficiently small representative of it. To consider perturbations, we always fix a sufficiently small representative of the map germ (here $f$), the space of control parameters ($t$) and the set germ in the target ($(\Sigma, A_0)$). See details in Appendix~\ref{ss:intsect}} $f_t$ of $f$. We express this number $\sharp f_t^{-1}(\Sigma,A_0)$ in terms of the codimension of  the ideal $J \subset \mathcal{O}_4$ generated by the $(n-1) \times (n-1)$ minors $M_{ij}(f(x,y,z) -\lambda \mathds{1})$ of $f(x,y,z) -\lambda \mathds{1}$ in the algebra $\mathcal{O}_4$ of holomorphic map germs in 4 variables $x,y,z,\lambda$. More precisely,
\begin{equation}\label{eq:pullbackintro}
    \sharp f^{-1}_t(\Sigma, A_0; \lambda_0) =\dim_{\C} \frac{\mathcal{O}_4}{J},
\end{equation}
 where $(\Sigma, A_0; \lambda_0)$ denote the branch of the analytic set germ $(\Sigma, A_0)$ corresponding to the geometric degenerate eigenvalue $\lambda_0$ (see Equation~\eqref{eq:projectionlocal}), and $\sharp f^{-1}_t(\Sigma, A_0)$ is obtained by the sum for every geometric degenerate eigenvalue of $A_0$, see Theorem~\ref{th:pullback} and Remark~\ref{re:total}. 
In particular, $\sharp f_t^{-1}(\Sigma,A_0)$ does not depend on the perturbation. To derive Equation~\eqref{eq:pullbackintro}, we use the reduction (`lifting') of $\Sigma$ to $\Sigma'$. 

We apply Equation~ \eqref{eq:pullbackintro} for a degeneracy point $p_0 \in M^3$ of a parameter-dependent quantum system $H: M^3 \to \text{Herm}(n)$, $H(p_0)=A_0 \in \Sigma_{\text{herm}}$, more precisely, for the  Hamiltonian germ $H: (\R^3, 0) \to (\text{Herm}(n), A_0)$. If $p_0=0$ is an  isolated degeneracy point in the complex sense, that is, $H_{\C}^{-1}(\Sigma)=\{0\}$, then Equation \eqref{eq:pullbackintro} provides an upper bound $\sharp \mathbf{cWP}$ for the number of (real) Weyl points $\sharp \mathbf{WP}$ of a generic perturbation $H_t$ of $H$, see Corollary~\ref{co:pullback}. 

% We evaluate the complex multiplicity formula for a special case that is highly useful and relevant for many physics examples. 
We evaluate Equation~\eqref{eq:pullbackintro} in a special case defined by
the following conditions.
(1) We assume that $A_0 \in \Sigma$ has a \emph{strictly $k$-fold eigenvalue} $\lambda_0$, that is, both the algebraic and geometric multiplicity\footnote{Recall that the \emph{algebraic multiplicity} of an eigenvalue of a matrix is its multiplicity as a root of the characteristic polynomial, while its \emph{geometric multiplicity} is the dimension of the corresponding eigenspace. The algebraic multiplicity is greater than or equal to the geometric multiplicity.} 
    of $\lambda_0$ is equal to $k$. 
   (2) We assume that for the map germ $f: (\C^3, 0) \to (\C^{n \times n}, A_0)$ -- in addition to satisfying $f^{-1}(\Sigma, A_0; \lambda_0)=\{0\}$ --  the number $\sharp f_t^{-1}(\Sigma, A_0; \lambda_0)$ of pre-images of a generic perturbation $f_t$ of $f$ is equal to the multiplicity $\text{mult}(\Sigma, A_0; \lambda_0)$ of the analytic set germ $(\Sigma, A_0; \lambda_0)$. Then this number is
   \begin{equation}\label{eq:main}
       \sharp f_t^{-1}(\Sigma, A_0; \lambda_0)=\text{mult}(\Sigma, A_0; \lambda_0)=\frac{k^2(k^2-1)}{12},
   \end{equation}
   independently of the size $n$ of the matrices, see Theorem~\ref{th:multipl} and Corollary~\ref{co:generic}. Recall that condition (2) holds for almost every map, see Appendix~\ref{ss:intsect}. In particular, Equation~\eqref{eq:main} holds if $A_0=0$ (hence, $n=k$) and $f$ is a linear map with isolated degeneracy at 0, see Corollary~\ref{co:kanegyzet_symm}. 
   
   This result is applied to linear Hamiltonian systems $H: (\R^3, 0) \to (\text{Herm}(k), 0)$ with isolated degeneracy point at 0 in the complex sense. In this case 
   \begin{equation}
       \sharp \mathbf{WP} \leq \sharp \mathbf{cWP}=\frac{k^2(k^2-1)}{12}
   \end{equation}
    holds for the number of real and complex Weyl points of a generic perturbation $H_t$ of $H$, see Corollary~\ref{co:pyram}. In particular, this is the case for the spin Hamiltonian and for the electronic band structure of some materials, see Section~\ref{ss:exspin} and Section~\ref{sec:bandstructureexample}.
    
    We suggest the terminology \emph{$k$-fold Weyl point} for the most generic $k$-fold degeneracy points, i.e. those with $\sharp \mathbf{cWP}=k^2(k^2-1)/12$. Hence, the isolated degeneracy point at 0 (in complex sense) of a linear Hamiltonian is a $k$-fold Weyl point. 
    
    In Section~\ref{ss:degpoints} we also establish a lower bound $\sharp_{\text{alg}} \mathbf{WP}$ for the number of Weyl points $\sharp \mathbf{WP}$. In Section~\ref{ss:chern}, $\sharp_{\text{alg}} \mathbf{WP}$ is expressed in terms of the first Chern numbers of the eigenvector bundles, following the arguments in \cite{Bruno2006}.   Note that, in contrast of the two-fold case, $k$-fold Weyl points possibly differ from each other in the lower bound $\sharp_{\text{alg}} \mathbf{WP}$, see Remark~\ref{re:lowerweyl} and the example in Section~\ref{sec:bandstructureexample}.

     For the proof of formula~\eqref{eq:main}, we consider the $A_0=0$ case, then it is generalized for arbitrary $k$-fold degenerate $A_0$.  We show that $\text{mult}(\Sigma, 0)=\text{mult}(\Sigma',0)$, see Proposition~\ref{pr:primetilde}.  The fact that $\text{mult}(\Sigma',0)=k^2(k^2-1)/12$ is a special case of a more general result, indeed, $\text{mult}(V(I_r), 0)$ is known for the vanishing locus $V(I_r) \subset \C^{k \times k}$ of the ideal $I_r$ generated by the $r \times r$ minors, see Refs.~\cite{HarrisTu,HerzogTrung}
     ($r=k-1$ in our special case, that is, $\Sigma'=V(I_{n-1})$). In this paper we give an independent, more direct and more elementary proof for $\text{mult}(\Sigma',0)=k^2(k^2-1)/12$ by 
     computing the Hilbert series of the corresponding algebra using the \emph{Gulliksen--Neg{\aa}rd free resolution}, see Section~\ref{ss:proofgencase}. 

For arbitrary $k$-fold degenerate $A_0 \in \Sigma \subset \C^{n \times n}$ we show that $(\Sigma, A_0; \lambda_0)$ is an analytic trivial deformation of $(\Sigma, 0) \subset (\C^{k \times k}, 0)$, hence, $\text{mult}(\Sigma, A_0; \lambda_0)=\text{mult}(\Sigma,0)$, see Corollary~\ref{co:trivfam} and \ref{co:effmodmult}.
For this, we introduce a local parametrization we call \emph{complex Schrieffer--Wolff chart} of $\C^{n \times n}$ fitting to $(\Sigma, A_0; \lambda_0)$ around the strictly $k$-fold degenerate matrix $A_0$, see Section~\ref{ss:compSW}. The Schrieffer--Wolff transformation \cite{SW,BravyiSW,PinterSW} is a standard method used in quantum mechanics to decrease the dimension of the Hamiltonian around a degeneracy point. It is shown in Ref.~\cite{PinterSW} that the Schrieffer--Wolff transformation induces a local chart in $\text{Herm}(n)$, which fits to $(\Sigma_{\text{herm}}, A_0; \lambda_0)$. Therefore, the complex Schrieffer-Wolff chart is a generalization of this result.

    As a byproduct of our approach, we show that $\Sigma$ is not Cohen--Macaulay, see Theorem~\ref{co:notcohmac}. Hence, the reduction of $\Sigma$ to $\Sigma'$ cannot be avoided in Theorem~\ref{th:pullback}, in other words, the number of pre-images of a perturbation cannot be obtained in a simple way from the vanishing ideal $I(\Sigma)$ of $\Sigma$, see Theorem~\ref{th:notwork00}.

     Analogous results are formulated to real and complex symmetric matrix families. Although the diagonal case is trivial, it is also discussed to illustrate the methods and provide a more complete picture.  The three cases (general; complex symmetric and diagonal) are discussed simultaneously.

%\notegergo{Structure of paper-be?}   In Section~\ref{s:physex} we give a summary of the theory of degeneracy points of parameter dependent Hamiltonian systems. This summary includes the lower bound for $\sharp \mathbf{WP}$ in terms of the first Chern numbers of eigenvector bundles, and also the study of (non-generic) two-fold degeneracy points via the effective map germ $h: (\R^3, 0) \to (\R^3, 0)$ derived from the Schrieffer--Wolff transformation. See Section~\ref{ss:chern} and \ref{ss:twofold}, respectively.

%without Schrieffer--Wolff, since the real algebra we construct from the Hamiltonian $H$ essentially carries every information which is traditionally computed  from the effective map germ $h$.
   
   The results are illustrated on physical examples, including the spin Hamiltonian (Section~\ref{ss:exspin}) and an example from condensed matter physics describing the electronic band structure of a material (Section~\ref{sec:bandstructureexample}).  

     Moreover, we show that in case of two-fold degeneracy points a complete description of the type of degeneracy can be given using our methods. More precisely, two-fold degeneracy points are traditionally investigated via Schrieffer--Wolff transformation, which provides an \emph{effective map germ} $h=(h_1, h_2, h_3): (\R^3, 0) \to (\R^3, 0)$ characterizing the degeneracy, see Section~\ref{ss:efftwofold}. It is reasonable to consider $h$ up to contact equivalence, which is the same as its real local algebra $Q_0(h)=\mathcal{O}_3^{\R}/I(h_1, h_2, h_3)$ up the isomorphism. Using the above algebraic methods we construct the local algebra $Q_0(h)$ of $h$ directly from $H$  avoiding the SW, up to isomorphism of real algebras. Hence, this algebra carries the relevant information about the degeneracy point.  See Section~\ref{ss:twofold}.

\subsection{Structure of the paper}
The structure of the paper is the following. 

In Section~\ref{s:tools} the degeneracy varieties $\Sigma_{\text{herm}}$, $\Sigma$ (etc.), and the reduction of $\Sigma$ to $\Sigma'$ are introduced. The section finishes with the complex Schrieffer--Wolff chart and its consequences.

The proof of the main theorems are in Section~\ref{s:proofs}. We compute the multiplicity of the introduced complex analytic set germs, and the multiplicity of holomorphic map germs  with respect to these set germs. We show that $\Sigma$ is not Cohen--Macaulay.

In Section~\ref{s:physex} we give a summary of the theory of degeneracy points of parameter dependent Hamiltonian systems. This summary includes the lower bound for $\sharp \mathbf{WP}$ in terms of the first Chern numbers of eigenvector bundles, and also the study of (non-generic) two-fold degeneracy points via the effective map germ $h: (\R^3, 0) \to (\R^3, 0)$ derived from the Schrieffer--Wolff transformation.   The results are illustrated on physical examples, including the spin Hamiltonian (Section~\ref{ss:exspin}) and an example from condensed matter physics describing the electronic band structure of a material (Section~\ref{sec:bandstructureexample}).

In Appendix~\ref{a:app} the required notions of algebraic geometry are summarized. Especially, the multiplicity of holomorphic map germs with respect to complex analytic set germs is discussed in Section~\ref{ss:intsect}.

%Although for this we adapt classical notions from algebraic geometry, they are put together in an unusual way to fit our setup.
%In particular, we show that the variety of matrices with eigenvalues of geometric multiplicity at least two is not Cohen--Macaulay, hence, the computation of the multiplicity requires comparing our variety with classical determinantal varieties. The corresponding Hilbert sequence is computed via the Gulliksen--Negard resolution. Joint work in progress with Gyorgy Frank, Daniel Varjas, Alex Hof and Andras Palyi.

\section{Matrix varieties}\label{s:tools}

%Here we introduce the concepts we use in the proofs. The material of this section is very close to classical results, however, we need slightly unusual forms of them, which we did not find in the literature. From this reason we give the proofs for our versions.

In Section~\ref{ss:degcomp} we summarize the properties of the degeneracy variety $\Sigma_{\text{herm}} \subset \text{Herm}(n)$ and its complexification, the geometric degeneracy variety $\Sigma \subset \C^{n \times n}$. Then, in Section~\ref{ss:matrvar} we describe its relation with the determinantal variety $\Sigma'$ consists of matrices of rank at most $n-2$. Finally, in Section~\ref{ss:compSW} we give a suitable local parametrization of $\C^{n \times n}$ called complex Schrieffer--Wolff chart around strictly $k$-fold degenerate matrices $A_0 \in \Sigma$.

\subsection{Preliminaries: degeneracy variety and complexification}\label{ss:degcomp}

We start with a general (preliminary) description of the (hermitian) degeneracy variety $\Sigma_{\text{herm}} \subset \text{Herm}(n)$, which is a shortened and slightly modified version of \cite[Section 2]{PinterSW}. Then, the description of the complexification $\Sigma \subset \C^{n \times n}$ is based on the results in \cite{DomokosHermitian}. We also mention the case of real and complex symmetric matrices which have analogous characterization. 

 The Hermitian matrices (of size $n \times n$) are the complex matrices $H$ with $H=H^{\dagger}$. Their set $\mbox{Herm}(n)$ is a real vector space, it is a real subspace of  $\C^{n \times n}$ of all $ n \times n$ complex matrices. Although $\C^{n \times n}$ is a complex vector space (of dimension $n^2$), $\mbox{Herm}(n)$ inherits only real vector space structure, because it is not closed under the multiplication by the imaginary unit $i$. 
The dimension of $\mbox{Herm}(n)$ over $\R$ is $n^2$, i.e., $\mbox{Herm}(n) \cong \R^{n^2}$ as real vector spaces, a suitable basis is described in \cite{PinterSW}.

Hermitian matrices have real eigenvalues and they can be diagonalized by a unitary basis transformation.
Furthermore, for each $H \in  \mbox{Herm}(n)$, one can choose a unitary matrix $U \in U(n)$ such that 
\begin{equation}\label{eq:dia} \Lambda = U^{-1} \cdot H \cdot U \end{equation}
is diagonal with the eigenvalues in increasing order $\lambda_1 \leq \lambda_2 \leq \dots \leq \lambda_n$. The ordered eigenvalues are continuous functions $\lambda_i: \mbox{Herm}(n) \to \R$, as it can be proven e.g. using Weyl's inequality~\cite{Weyl1912}.
However, the diagonalizing matrix $U$ is not unique; each column, in fact, can be independently multiplied by an arbitrary phase factor. Moreover, in the case of degenerate eigenvalues $\lambda_{i+1} =\dots = \lambda_{i+k} $, the corresponding columns of $U$ form a unitary basis of the corresponding $k$-dimensional eigenspace. This basis can be transformed by any unitary action of $U(k)$ to obtain a different matrix $U'$ that still diagonalizes $H$.

The \emph{(hermitian) degeneracy set} $\Sigma_{\text{herm}}^{(n)} =\Sigma_{\text{herm}} \subset \mbox{Herm}(n)$ is the set of matrices with at least two coinciding eigenvalues,
\begin{equation}
    \Sigma_{\text{herm}}=\{A \in \text{Herm}(n) \ | \ \exists i, \  \lambda_i(A) =\lambda_{i+1}(A) \}.
\end{equation}
This subset $\Sigma_{\text{herm}}$ is a real subvariety, indeed, it can be defined
by one polynomial equation. Namely, let $\Delta(A)$ be the discriminant of the characteristic polynomial of $A$, it is a polynomial of the matrix entries $a_{ij}$ of $A$. Then $\Sigma_{\text{herm}}$  is the zero locus of $\Delta$, i.e., $\Sigma_{\text{herm}}=\Delta^{-1}(0)$.
By the Neumann--Wigner theorem \cite{Neumann1929} the codimension of $\Sigma_{\text{herm}}$ is 3. See two different proofs for it in \cite{PinterSW}, Table 1 and Corollary 3.1.6, see also see Corollary~\ref{co:twofoldsmooth} -- a third proof is indicated in Section~\ref{ss:matrvar} of this paper.

The subvariaty $\Sigma_{\text{herm}}$ can be decomposed into the disjoint union of different \emph{strata} based on the type of the degeneracy, i.e., which eigenvalues coincide \cite{ArnoldSelMath1995}. Each stratum can be labeled by an ordered partition of $n$ in the following  way. 
Let $\kappa=(k_1, \dots, k_l)$ be a sequence of integers $1 \leq k_i \leq n$ of length $1 \leq l < n$ such that $n=\sum_{i=1}^l k_i$, defining $\kappa$ as an ordered partition of $n$. 
The number of the ordered partitions of $n$ is $2^{n-1}$, (including the case with $l=n$). The associated stratum of $\Sigma_{\text{herm}} $ consists of the matrices with coinciding eigenvalues $\lambda_1 = \dots = \lambda_{k_1} <  \lambda_{k_1 + 1} = \dots = \lambda_{k_1+k_2} < \dots$.

The partition with $l=n$ and $k_i=1$ marks the complement of $\Sigma_{\text{herm}}$, the set of non-degenerate matrices. Each stratum  is a (not closed) smooth submanifold of $\mbox{Herm}(n)$, whose dimension is $n^2-\sum_{i=1}^l \left(k_i^2-1 \right)$ by the Neumann--Wigner theorem \cite{Neumann1929}, cf. \cite[Table 1]{PinterSW}. 
As a consequence of the continuity of the eigenvalues, the closure of the stratum labeled by $\kappa$ contains the higher degeneracies corresponding to the ordered partitions $\kappa'$ coarser than $\kappa$. 

 $\Sigma_{\text{herm}}$ and its real symmetric counterpart $\Sigma_{\text{sym}, \R}$ are studied in the literature from topological and physical point of view \cite{ArnoldSelMath1995},
and from  differential geometrical approach \cite{BreidingSymmetric, KozhasovMinimality}.

The complexification of $\text{Herm}(n) \cong \R^{n^2}$ is the space of $n \times n$ complex matrices $\C^{n \times n} \cong \C^{n^2}$.  What is the complexification of $\Sigma_{\text{herm}}$? A candidate would be the zero locus the discriminant of the characteristic polynomial $\Delta^{-1}(0) \subset \C^{n \times n}$, but it is too big, indeed, its codimension is 1. Although its intersection with $\text{Herm}(n)$ is the codimension 3 real variety $\Sigma_{\text{herm}}$, it is a real phenomenon caused by the fact that $\Delta$ is the sum of complete squares, see \cite{DomokosHermitian,DomokosDiscriminant,DomokosInvariant}.
%M. Domokos, Discriminant of symmetric matrices as a sum of squares and the orthogonal group, Comm. Pure Appl. Math. 64 (2011) 443–465.
%M. Domokos, Invariant theoretic characterization of subdiscriminants of matrices, Linear Multilinear Algebra (2013),http://dx.doi.org/10.1080/03081087.2012.762714, in press.
Note that $A \in \Delta^{-1}(0)$ if and only if $A$ has at least one eigenvalue $\lambda$ with algebraic multiplicity at least 2, that is, $\lambda$ is a multiple root of the characteristic polynomial of $A$. In this sense $\Delta^{-1}(0) \subset \C^{n \times n}$ is the \emph{algebraic degeneracy variety}, but it will be not used in the rest of this article.

Instead, we consider the \emph{geometric degeneracy variety} $\Sigma^{(n)}=\Sigma \subset \C^{n \times n}$, which consists of the matrices $A$ satisfying one, (therefore, all) of the following equivalent properties:
\begin{enumerate}
    \item $A$ has at least one eigenvalue $\lambda$ with geometric multiplicity at least 2, that is, the dimension of the corresponding eigenspace is at least 2.
    \item $A$ has at least two Jordan blocks for at least one eigenvalue $\lambda$.
    \item The degree of the minimal polynomial of $A$ is less than $n$.
    \item The powers $\mathds{1}_n, A, A^2, \dots, A^{n-1}$ of $A$ are linearly dependent. 
    \item Writing the above powers of $A$ as column vectors, the $n \times n$ minors of the resulting $n^2 \times n$ matrix are 0.
\end{enumerate}

Let us denote by $I_{(n)} \subset \C[A]$ the ideal generated by the minors described in point 5. Its zero locus $V(I_{(n)}) \subset \C^{n \times n}$ is equal to $\Sigma^{(n)}$, as sets. In \cite{DomokosHermitian}, $V(I_{(n)})$ is denoted by $(\C^{n \times n})_1$, and it is proved that
\begin{itemize}
    \item For $n=3$, $I_{(3)}$ is radical ideal, i.e. it is the vanishing ideal (reduced structure) of $\Sigma^{(3)} $. The analogous statement for $n >3$ is not known. See \cite[Cor. 4.7.]{DomokosHermitian}.
    \item For every $n$, the complexification of $\Sigma_{\text{herm}}$ is $\Sigma$. That is, considering all real polynomials in $\R[A]$ vanishing on $\Sigma_{\text{herm}} \subset \R^{n^2}$, their common zero locus in $\C^{n^2}$ is $\Sigma$. Equivalently, $\Sigma$ is the smallest complex  algebraic subset in $\C^{n \times n}$ containing $\Sigma_{\text{herm}}$, i.e. $\Sigma$ is the complex Zariski closure of $\Sigma_{\text{herm}}$. See \cite[Prop. 2.4.]{DomokosHermitian}.
\end{itemize}

In Section~\ref{ss:proof-notcohmac} we construct the the vanishing ideal of $\Sigma^{(n)}$ for every $n$ in a different way, and we prove that $\Sigma^{(n)}$ is not Cohen--Macaulay if $n \geq 3$. 

A stratification of $\Sigma$ can be given according to the Jordan block structure of the matrices. Since the hermitian matrices have only $1 \times 1$ Jordan blocks, $\Sigma_{\text{herm}} \subset \Sigma$ is contained in the union of the corresponding strata. However, as complex matrices, the ordering of the eigenvalues is not defined. In other words, two strata of $\Sigma_{\text{herm}}$ corresponding to different ordered partitions $\kappa$ and $\kappa'$ are in the same stratum of $\Sigma$ if $\kappa$ and $\kappa'$ define the same unordered partition. 

In Section~\ref{ss:matrvar} we show that each eigenvalue $\lambda$ of $A$ with geometric multiplicity at least 2 defines
an analytic subsetgerm $(\Sigma, A; \lambda) \subset (\Sigma, A)$ which is the union of irreducible components. If $\lambda$ is a strictly $k$-fold eigenvaule (i.e. both the algebraic and geometric multiplicity  is $k$), then we describe $(\Sigma, A; \lambda)$ in Section~\ref{ss:compSW}, in particular we show that it is one irreducible component.

Similar description can be given for the real symmetric degeneracy variety $\Sigma_{ \text{sym}, \R} \subset \text{Symm}_{\R}(n) \cong \R^{(n+1)n/2}$. Its codimension is 2. Its stratification also corresponds to the ordered partitions $\kappa$, but the codimensions of the strata are different than in the hermitian case. As proved in \cite{DomokosHermitian}, the complexification (complex Zariski closure) of $\Sigma_{ \text{sym}, \R}$ is $\Sigma_{\text{sym}, \C}=\Sigma_{\text{sym}} \subset \text{Symm}_{\C}(n)$.

\subsection{Lifting the geometric degeneracy variety}\label{ss:matrvar}

Here we describe the relation of $\Sigma$ with the determinantal variety $\Sigma'$ consists of matrices of rank at most $n-2$. As an intermediate step we introduce the lifted degeneracy variety $\widetilde{\Sigma} \subset \C^{n \times n} \times \C$. We summarize the properties $\Sigma'$, $\widetilde{\Sigma}$ and the natural projection $\widetilde{\Sigma} \to \Sigma$.

Consider the following complex varieties 
\begin{eqnarray}
 \text{Matrices of corank at least 2:} &  \Sigma'^{(n)}=\Sigma'=\{A \in \C^{n \times n} \ | \ \text{rk}(A) \leq n-2 \} \\
 \text{Lifted geometric degeneracy variety:}   & \widetilde{\Sigma}^{(n)}=\widetilde{\Sigma}=\{(A, \lambda) \in \C^{n \times n} \times \C \ | \ A-\lambda \mathds{1} \in \Sigma' \} 
\end{eqnarray}
Then an equivalent characterization of the geometric degeneracy variety $\Sigma$ is 
  \begin{equation}\label{eq:sigmanew} 
  \Sigma=\{ A \in \C^{n \times n} \ | \ \exists \lambda \in \C, \ (A,\lambda) \in \widetilde{\Sigma} \}=
  \{ A \in \C^{n \times n} \ | \ \exists \lambda \in \C, \ \text{rk}(A-\lambda \mathds{1}) \leq n-2 \}.
\end{equation}

By classical results, $\Sigma' \subset \C^{n \times n}$ is an affine variety of codimension 4 \cite[Proposition 1.1]{BrunsVetterBook}.
% Brunns-Wetter: Determinantel rings
Let $\C{[A]}$ denote the polynomial ring on $n^2$ variables $A=(a_{ij})_{i,j=1, \dots, n}$. Let $I_{n-1} \subset \C{[A]}$ denote the ideal generated by the $(n-1) \times (n-1)$ minors $M_{ij}(A)$. Obviously, $\Sigma'=V(I_{n-1})$. Moreover, by \cite[Theorem
2.10]{BrunsVetterBook}, $I_{n-1}$  is a prime ideal. Therefore, 
\begin{itemize}
    \item $I_{n-1}$ defines the reduced structure of $\Sigma'$, see also Theorem
2.11 in \cite{BrunsVetterBook}, i.e., $I(\Sigma)=I_{n-1}$.
\item $\Sigma'$ is irreducible, see also Proposition 1.1 in \cite{BrunsVetterBook}.
\end{itemize}
The non-singular points of $\Sigma'$ are the matrices $A$ with $\text{rk} (A)=n-2$, see Proposition 1.1 in \cite{BrunsVetterBook}.
% Brunns, Vetter: Determinantal rings
Moreover, $\Sigma'$ is the of singular points of the set of singular matrices $\{A \in \C^{n \times n} \ | \ \det (A)=0\}$.

We also consider $\Sigma'$ locally around the origin, that is, the complex analytic germ $(\Sigma',0) \subset (\C^{n\times n}, 0)$. Slightly abusing the notation, we also denote by $I_{n-1}$ the ideal  in $\mathcal{O}_{n \times n}$. It defines a reduced structure of $(\Sigma',0)$. Since $I_{n-1}$ is homogeneous, $(\Sigma',0)$ is irreducible as complex analytic set germ, which follows from its algebraic irreducibility and homogeneity by Serre's GAGA principle \cite{SerreGAGA}.
% Serre: Géometrie Algébrique et Géométrie Analytique

$\widetilde{\Sigma}$ is a trivial deformation of $\Sigma'$, i.e. $\widetilde{\Sigma}\cong \Sigma' \times \C$. In fact, for each fixed $\lambda \in \C$, the set of matrices $A$ satisfying $(A, \lambda) \in \widetilde{\Sigma}$ is the translated copy $\Sigma'+\lambda \mathds{1}$ of $\Sigma'$:
\begin{equation}
    (A, \lambda) \in \widetilde{\Sigma} \ \Leftrightarrow \ A \in \Sigma'+\lambda \mathds{1}.
\end{equation}
Therefore, the codimension of $\widetilde{\Sigma}$ in $\C^{n \times n} \times \C$ is $4$, and its reduced structure in $\C[A][\lambda]$ (or, locally, in $ \mathcal{O}_{n \times n +1}$) is given by the ideal $I(\widetilde{\Sigma})$ generated by the $(n-1) \times (n-1) $ minors of $M_{ij}(A-\lambda \mathds{1})$ of $A-\lambda \mathds{1}$. In other words, this is the  pull-back ideal $\rho^*(I_{n-1}) \cdot \C[A][\lambda]$, where $\rho(A, \lambda)=A-\lambda \mathds{1}$.

As a conclusion of the corresponding properties of $\Sigma'$, $\widetilde{\Sigma}$ is also irreducible, and its non-singular points are the pairs $(A, \lambda)$ with $\text{rk}(A- \lambda \mathds{1})=n-2$, that is, the dimension of the eigenspace of $A$ corresponding to $\lambda$ is equal to 2. 

%Moreover, $\widetilde{\Sigma}$ is the set of singular points of the eigenvalue relation (dispersion relation) $\{(A, \lambda) \in \C^{n \times n} \times \C \ | \ \det(A-\lambda \mathds{1})=0 \}$

The vanishing ideals of $\Sigma'$ and $\widetilde{\Sigma}$ are homogeneous ideals of degree $n-1$. Indeed, the minor determinants are degree $n-1$ homogeneous polynomials of the matrix entries, and, for $\widetilde{\Sigma}$, the matrix entries are homogeneous linear polynomials of the variables $a_{ij}$ and $\lambda$. In other words, the homogeneity of the ideals is equivalent with the fact that $A \in \Sigma'$ implies $c \cdot A \in \Sigma'$, and $(A, \lambda) \in \widetilde{\Sigma}$ implies $(c \cdot A, c \cdot \lambda) \in \widetilde{\Sigma} $ for all $c  \in \C$. 

This also means that the space germs $(\Sigma', 0)$ and $(\widetilde{\Sigma}, (0,0))$ are homogeneous. 
However, the space germs $(\Sigma', A_0)$ and $(\widetilde{\Sigma}, (A_0, \lambda_0))$ are not homogeneous, if $A_0 \neq 0$. Algebraically it means, for $(\Sigma', A_0)$, that the minors of $A+A_0$ are not homogeneous polynomials of the matrix entries of $A$. Geometrically it means that $A+A_0 \in \Sigma'$ does not imply $c \cdot A +A_0 \in \Sigma'$ in general. Similarly for $\widetilde{\Sigma}$.

$\Sigma'$ and $\widetilde{\Sigma}$ (and also their local versions) are determinantal varieties. That is, their vanishing ideals are generated by $r \times r$ minors of a $p \times q$ matrix over the Cohen-Macaulay rings $\C[A]$ and $\C[A][ \lambda]$, respectively ($\mathcal{O}_{n \times n}$ and $\mathcal{O}_{n \times n +1}$ in the local version), for some integer $r$  such that their codimension is $(p-r+1)(q-r+1)$; in this concrete situation, $p=q=n$, $r=n-1$ and the codimension is $(n-(n-1)+1)^2 = 2^2 = 4$. See \cite[Section 18.5]{EisenbudCommutativeAlgebra}, also \cite{BrunsVetterBook} and \cite[Theorem C.6]{MondBook}.

Since determinantal varieties  are Cohen--Macaulay \cite[Theorem 2.5.]{BrunsVetterBook}, $\Sigma'$ and $\widetilde{\Sigma}$ are Cohen--Macaulay. Therefore, they behave nicely under perturbation of map germs in sense Section~\ref{ss:intsect}, Proposition~\ref{pr:everypert}.

Recall from the previous section that $\Sigma$ is the set of matrices $A$ with at least one eigenvalue $\lambda$ of geometric multiplicity at least 2, that is, the dimension of the eigenspace corresponding to $\lambda$ is at least 2. Since $\Sigma=V(I_{(n)})$, it is a complex algebraic subset. However, it is known only for $n=3$ that $I_{(n)}$ is the vanishing ideal of $\Sigma$, that is, $I(\Sigma)=I_{(n)}$. In general, $\sqrt{I_{(n)}}=I(\Sigma)$ holds.

Let $p: \C^{n \times n} \times \C \to \C^{n \times n}$ denote the projection onto the first coordinate, $p(A, \lambda)=A$. By Equation~\eqref{eq:sigmanew}, the image of the restriction $p|_{\widetilde{\Sigma}}$ is $\Sigma$. Consider the restriction
\begin{equation}
    p|_{\widetilde{\Sigma}}: \widetilde{\Sigma} \to \Sigma
\end{equation}
onto its image. Observe that each $A \in \Sigma$ has a preimage $(A,\lambda)$ for every eigenvalue $\lambda$ with geometric multiplicity at least 2. Hence, 
non-injectivity of $p|_{\widetilde{\Sigma}}$ happens at points $(A, \lambda), (A, \mu) \in \widetilde{\Sigma}$ with $\lambda \neq \mu$, that is, for $A \in \Sigma$ with more than one 
geometric degenerate eigenvalue. 
By the same argument, each point $A \in \Sigma$ has finitely many preimages in $\widetilde{\Sigma}$. 

What is more, we show that $p|_{\widetilde{\Sigma}}$ is \emph{finite} and \emph{generically 1-to-1} (later, in next section). We will use these stronger properties to characterize the local components of $\Sigma$, to determine the vanishing ideal $I(\Sigma)$ of $\Sigma$  via $p|_{\widetilde{\Sigma}}$ and to prove that $\Sigma$ is not Cohen--Macaulay in Section~\ref{ss:proof-notcohmac}. 

For finiteness, recall the induced homomorphism between the quotient algebras
\begin{eqnarray}\label{eq:inducedhomo}
    (p|_{\widetilde{\Sigma}})^*: \mathcal{O}_{\Sigma}=\frac{\C[A]}{I(\Sigma)} & \to & \frac{\C[A][\lambda]}{I(\widetilde{\Sigma})}=\mathcal{O}_{\widetilde{\Sigma}} \\
    (p|_{\widetilde{\Sigma}})^*(h) &=& h \circ  p|_{\widetilde{\Sigma}} \\
    (p|_{\widetilde{\Sigma}})^*(h)(A, \lambda) &=& h(A)
    ,
\end{eqnarray}
and that $p|_{\widetilde{\Sigma}}$ is called \emph{finite} if $\mathcal{O}_{\widetilde{\Sigma}}$ becomes a finitely generated module over $\mathcal{O}_{\Sigma}$ via $(p|_{\widetilde{\Sigma}})^*$. Finiteness implies that each pre-image $(p|_{\widetilde{\Sigma}})^{-1}(A)$ is finite \cite[Sec.~5.3]{Shafarevich}.

\begin{prop}\label{pr:projectiongen1to1}
$p|_{\widetilde{\Sigma}}$ 
is finite.
\end{prop}

\begin{proof}
$\mathcal{O}_{\widetilde{\Sigma}}$ is a  finitely generated module over $\mathcal{O}_{\Sigma}$ if and only if it is finitely generated over $\C [A]$, that is, if $p|_{\widetilde{\Sigma}}$, considered as $\widetilde{\Sigma} \to \C^{n \times n}$, is finite. In this version, $(p|_{\widetilde{\Sigma}})^*(h)=h $ modulo $I(\widetilde{\Sigma})$ holds for a polynomial $h \in \C[A] \subset \C[A][\lambda]$. Hence, $\mathcal{O}_{\widetilde{\Sigma}}$ is generated over $\C[A]$ by the powers of $\lambda$  modulo $I(\widetilde{\Sigma})$. We have to show that there is a power $\lambda^K$ which is the combination of lower powers with coefficients in $\C[A]$, modulo $I(\widetilde{\Sigma})$.

Consider the minor $M_{nn}(A-\lambda \mathds{1}) \in I(\widetilde{\Sigma})$. Since it is a monic polynomial in $\lambda$ of degree $n-1$, it provides $\lambda^{n-1}$ as the combination of lower powers with coefficients in $\C[A]$, modulo $I(\widetilde{\Sigma})$, proving the statement.
\end{proof}

Therefore, we can see again that $\Sigma$ is an algebraic subset, since it is the image of a finite polynomial map \cite[Theorem D.6]{MondBook}, \cite[Sec.~1.5.3, Finite Maps]{Shafarevich}. This also shows that $\dim(\widetilde{\Sigma})=\dim(\Sigma)$, hence, the codimension of $\Sigma $ in $\C^{n \times n}$ is 3. Moreover, since $\widetilde{\Sigma}$ is Cohen--Macaulay, it is pure-dimensional (that is, its dimension is the same at all points), hence, $\Sigma$ is pure-dimensional as well.

The vanishing ideal $I(\Sigma)$ of $\Sigma$ in $\C[A]$  is obviously homogeneous, since $A \in \Sigma$ implies $c \cdot A \in \Sigma$ for every $A$ and $c \in \C$. This also means that the localization $(\Sigma, 0) \subset \C^{n \times n}$ is homogeneous (it is a cone), but  the space germs $(\Sigma, A_0)$ are not homogeneous in general for $A_0 \neq 0$. In Section~\ref{ss:proof-notcohmac} we determine $I(\Sigma)$ as the kernel of the induced homomorphism $(p|_{\widetilde{\Sigma}})^*$, and we prove that $\Sigma$ is not Cohen--Macaulay.

%\notealex{Have we given an argument anywhere that $\Sigma$ is actually a closed algebraic set? It will follow by knowing $p|_{\widetilde{\Sigma}}$ is finite, e.g., but this should be spelled out.} 

 %The authors do not know whether $\Sigma$ is determinantal or Cohen-Macaulay. These properties cannot be deduced directly from the construction. However, these properties are not needed for the proofs. \textcolor{blue}{Probably we will know the ideal and the Cohen--Macaulay property.} \notealex{The ideal will be the kernel of the map $\C[A] \to \mathcal{O}_{\widetilde{\Sigma}}$ I mention above, although this may not be particularly helpful to know in practice.}

 A finite morphism is also locally finite, hence, the map germ
     \begin{equation}\label{eq:projectionlocal}
     p|_{(\widetilde{\Sigma}, (A_0, \lambda_0))}: (\widetilde{\Sigma}, (A_0, \lambda_0)) \to (\Sigma, A_0)
 \end{equation}
 is finite for $(A_0, \lambda_0) \in \widetilde{\Sigma}$. Let $(\Sigma, A_0; \lambda_0) \subset (\Sigma, A_0)$ denote its image, i.e. the local branch of the projection $p|_{\widetilde{\Sigma}}$. By Remmert’s proper mapping \cite[Theorem D.6]{MondBook}, $(\Sigma, A_0; \lambda_0) \subset (\Sigma, A_0)$ is an analytic subsetgerm of  $(\Sigma, A_0)$ of the same dimension $\dim (\Sigma, A_0; \lambda_0)=n^2-3$, and it is also pure-dimensional. We have the following consequence.

 \begin{cor}\label{co:unionofcomp} \begin{enumerate}
     \item $(\Sigma, A_0; \lambda_0)$ is the union of irreducible components of $(\Sigma, A_0)$.\footnote{The definition of irreducible components of an affine algebraic set or analytic set germ, can be found in the standard textbooks of algebraic geometry, for example, in \cite{HartshoneBook}. However, in case of pure-dimensional analytic set germs, an equivalent definition is the following: an irreducible component is a minimal pure-dimensional analytic subsetgerms of the same dimension.}
     \item If a matrix $A_0 \in \Sigma$ has more than one geometric degenerate eigenvalues, say, $\lambda_0 \neq \mu_0$ such that $(A_0, \lambda_0), (A_0, \mu_0) \in \widetilde{\Sigma}$, then the analytic set germ $(\Sigma, A_0)$ is reducible.
 \end{enumerate}
      \end{cor}

  \begin{rem}\label{re:everyirred}
      Moreover, every irreducible component of $(\Sigma, A_0)$ is contained in $(\Sigma, A_0; \lambda_0)$ for a geometric degenerate eigenvalue $\lambda_0$ of $A_0$. Indeed, from the finiteness of $p|_{\widetilde{\Sigma}}$ it can be deduced by Nakayama lemma \cite[C.1]{MondBook}  that every $A_0 \in \Sigma$ has an open neighborhood $U \subset \Sigma$ whose preimage decomposes as a disjoint union $(p|_{\widetilde{\Sigma}})^{-1}(U)=\bigsqcup_j V_j$, where $V_j$ are open neighborhoods of the preimages $(p|_{\widetilde{\Sigma}})^{-1}(A_0)=\{(A_0, \lambda_j)\}$.   Since $\bigsqcup_j V_j \to U$ is surjective, every component is contained in the image of a  $V_j$. 

      We will show  that $(\Sigma, A_0; \lambda_0)$ is an irreducible component of $(\Sigma, A_0)$, if $\lambda_0$ is a strictly $k$-fold eigenvalue of $A_0$, that is, both the algebraic and geometric multiplicities are $k$, see Corollary~\ref{cor:irred}.
 \end{rem}

 Intuitively, $(\Sigma, A_0; \lambda_0)$ is the set of matrices near $A_0$ which have an at least two-fold degenerate eigenvalue near $\lambda_0$ (but the other geometric degenerate eigenvalues of $A_0$ are possible dissolved).

Similar analysis can be done in the complex symmetric and diagonal cases, involving the following complex algebraic subsets:
\begin{itemize}
    \item $\Sigma'_{\text{sym}} =\Sigma' \cap \text{Symm}_{\C} (n)$ of codimension 3 in $\text{Symm}_{\C} (n) \cong \C^{(n+1)n/2}$.
    \item $\Sigma'_{\text{diag}} =\Sigma' \cap \text{Diag}_{\C} (n)$ of codimension 2 in $\text{Diag}_{\C} (n) \cong \C^n$.
   \item $\widetilde{\Sigma}_{\text{sym}} =\widetilde{\Sigma} \cap (\text{Symm}_{\C}(n) \times \C)$ of codimension 3 in $\text{Symm}_{\C}(n) \times \C\cong \C^{(n+1)n/2+1}$.
     \item $\widetilde{\Sigma}_{\text{diag}} =\widetilde{\Sigma} \cap (\text{Diag}_{\C}(n) \times \C)$ of codimension 2 in $\text{Diag}_{\C}(n) \times \C\cong \C^{n+1}$.
      \item $\Sigma_{\text{sym}} =\Sigma \cap \text{Symm}_{\C} (n)$ of codimension 2 in $\text{Symm}_{\C} (n) \cong \C^{(n+1)n/2}$.
      \item $\Sigma_{\text{diag}} =\Sigma \cap \text{Diag}_{\C} (n)$ of codimension 1 in $\text{Diag}_{\C} (n) \cong \C^{n}$.
\end{itemize}

If the size of the ambient matrix space $n \times n$ is important, we use the notations $\Sigma^{(n)}$, $\widetilde{\Sigma}^{(n)}$, $\Sigma'^{(n)}$, and similarly for the symmetric and diagonal versions.

The symmetric determinantal variety $\Sigma'_{\text{symm}} \subset \text{Symm}_{\C}(n)$ is irreducible, Cohen--Macaulay, its vanishing ideal is the ideal generated by the $(n-1) \times (n-1) $ minors of symmetric matrices, it is a prime ideal. Moreover, the singular points of $\Sigma'_{\text{symm}}$ are the matrices with corank at least 3. See e.g. \cite{HarrisTu,HosonoTakagi}.
%Kutz, R., Cohen–Macaulay rings generated by minors of symmetric matrices, J. Algebra 58 (1979), 80–92.
%Harris, J., and Tu, L. W., On symmetric and skew-symmetric determinantal varieties, Topology 26 (1987), 217–249.
%Józefiak, T., Lascoux, A., and Pragacz, P., Symmetric determinantal ideals, in Determinantal Ideals, Progress in Mathematics, vol. 264, Birkhäuser (2008), pp. 105–128. 
%Hosono–Takagi & Tyurin, in their survey of symmetric determinantal loci, assert
%SHOULD BE CHECKED!!!

Moreover, in the real vector spaces $\text{Herm}(n)$ and $\text{Symm}_{\R}(n)$ we also define the analogous sets, which are real algebraic subsets in these cases. They have the same dimensions and codimensions over $\R$ as their complexifications over $\C$. 

%\notegergo{What do we need? Determinantal, Cohen--Macaulay? Reduced? Irreducibility? Smooth points? Proposition 2.2.6.?}

\subsection{Complex Schrieffer--Wolff chart}\label{ss:compSW}

 Let $A_0 \in \Sigma^{(n)}$ be an $n \times n$ matrix with a \emph{strictly $k$-fold degenerate eigenvalue} $\lambda_0$, that is, both the algebraic and geometric multiplicity of $\lambda_0$ is $k$. The main purpose of this section is to show that the germ $(\Sigma^{(n)}, A_0; \lambda_0)$  is an analytic trivial deformation of the `effective model' $(\Sigma^{(k)}, 0)$, and similar statement holds for symmetric matrices. For this, we provide a local chart of $\C^{n \times n}$ in a neighborhood of $A_0$ which `fits' to $(\Sigma^{(n)}, A_0; \lambda_0)$. This chart will be the generalization of the so-called Schrieffer--Wolff transformation.

The Schrieffer--Wolff transformation (SW) is a commonly used method in quantum mechanics to decrease the size of the Hamiltonian around a degeneracy point
 \cite{SW,BravyiSW,PinterSW,PymablockPaper,PymablockRelease}.
In \cite{PinterSW} it was pointed out that the SW induces a local chart of $\text{Herm}(n)$ in a neighborhood of a $k$-fold degeneracy point $A_0 \in \Sigma_{\text{herm}}$, called SW chart. Here we generalize this result -- and in fact, the corresponding SW decomposition -- to general complex matrices and complex symmetric matrices. Note that the real symmetric case is included in the original hermitian version.
%, and also that the approximate method of computing the SW transformation is already formulated to the complex case in \cite{???}. %Az a cikk, amiben a hatvanysort altalanositjak komplexre. Ha nincs meg, hagyjuk ki.

Let $A_0 \in \Sigma$ be a matrix with a strictly $k$-fold degenerate eigenvalue $\lambda_0$.  $A_0$ can be  partially diagonalized by a general linear basis transformation  $G_0 \in GL_n(\C)$ in $\C^n$, that is, 
\begin{equation}\label{eq:normform}
 G_0^{-1} \cdot   A_0 \cdot G_0= \begin{pmatrix}
        \lambda_0 \mathds{1}_{k} & 0 \\
        0 & \widetilde{A}_{0} \\
    \end{pmatrix} :=A_0',
\end{equation}
where $\mathds{1}_{k}$ is the $ k \times k$ identity matrix, $\widetilde{A}_0 \in \C^{(n-k)\times (n-k)}$ satisfies that $\widetilde{A}_{0}-\lambda_0 \mathds{1}_{n-k}$ has maximal rank $n-k$.
 For example if $A_0'$ is the Jordan normal form of $A_0$, then it
satisfies these properties.

Consider the map
\begin{equation}\label{eq:swdec}
   g: (S,C,A_{\text{eff}}) \mapsto A=g(S,C,A_{\text{eff}})=
   G_0 \cdot e^{S} \cdot (A_0' +C+A_{\text{eff}}) \cdot e^{-S} \cdot G_0^{-1}, 
\end{equation}
where
\begin{enumerate}
    \item $S \in \C^{n \times n}$ has zero $k \times k$ upper-left and $(n-k) \times (n-k)$ lower-right diagonal blocks, that is,
    \begin{equation}
        S=\begin{pmatrix}
        0 & S_{12} \\
        S_{21} & 0 \\
    \end{pmatrix}
    \end{equation}
    with $S_{21} \in \C^{(n-k) \times k}$ and $S_{12} \in \C^{k \times (n-k)}$.
    \item $C$ possibly has non-zero elements only in the $(n-k) \times (n-k)$ lower-right diagonal block and the other entries are zero, that is,
       \begin{equation}
        C=\begin{pmatrix}
        0 & 0 \\
        0 & C_{22} \\
    \end{pmatrix}
    \end{equation}
    with $C_{22} \in \C^{(n-k)\times (n-k)}$.
  % \item $T$ has only a diagonal $k \times k$ upper-left block, and the other entries are zero, that is,
  %    \begin{equation}
   %     T=\begin{pmatrix}
   %     \lambda \mathds{1} & 0 \\
   %     0 & 0 \\
   % \end{pmatrix},
   % \end{equation} 
    \item  $A_{\text{eff}}$ has only a $k \times k$ upper-left block, and the other entries are zero, that is,
    \begin{equation}
        A_{\text{eff}}=\begin{pmatrix}
        \widetilde{A}_{\text{eff}} & 0 \\
        0 & 0 \\
    \end{pmatrix}
    \end{equation}
    with $\widetilde{A}_{\text{eff}} \in \C^{k \times k}$. $A_{\text{eff}}$ is the analogue of the `effective Hamiltonian' in the hermitian case.
    %, $\text{tr} (A_{\text{eff}})=\text{tr} (\widetilde{A}_{\text{eff}})=0$.
\end{enumerate}
Observe that the domain of $g$ has the same dimension $n^2$ as its target $A \in \C^{n \times n}$, in fact,
\begin{equation}
    (S,C,A_{\text{eff}}) \in \C^{2k(n-k)} \times \C^{(n-k)^2}  \times \C^{k^2} \cong \C^{n^2}.
\end{equation}

\begin{thm}[Complex Schrieffer--Wolff chart]\label{th:compsw} \begin{enumerate}[label=(\alph*)]
 In the above setup,   \item The map 
$
    g: \C^{n^2} \to \C^{n \times n}
$
is a complex analytic diffeomorphism (i.e. biholomorphsm)  between a neighborhood $\mathcal{U}_0$ of $0 \in \C^{n^2}$ and a neighborhood $\mathcal{V}_0$ of $A_0 \in \C^{n \times n}$.
\item For matrices $A=g(S, C, A_{\text{eff}}) \in \mathcal{V}_0$, $A \in (\Sigma^{(n)}, A_0; \lambda_0)$ holds if and only if $A_{\text{eff}} \in (\Sigma^{(k)}, 0)$. (Here, slightly abusing  the notation, $(\Sigma^{(n)}, A_0; \lambda_0)$ and $(\Sigma^{k}, 0)$ denote the representatives of the corresponding analytic set germs obtained as the intersection with $\mathcal{V}_0$ and $\mathcal{U}_0$, respectively.)
\item If in addition $\lambda_0=0$, that is, $A_0 \in \Sigma'^{(n)}$ is a strictly corank-$k$ matrix, then $A \in (\Sigma'^{(n)}, A_0)$ if and only if $A_{\text{eff}} \in (\Sigma'^{(k)}, 0)$.
\end{enumerate}
\end{thm}

\begin{proof} (a)
By the holomorphic inverse function theorem \cite[p. 18.]{GriffithsHarris1978}
%Griffiths, Phillip; Harris, Joseph (1978), Principles of Algebraic Geometry CSEKKK!!!
it is enough to show that the Jacobian $(dg)_0$ of $g: \C^{n^2} \to \C^{n^2}$ at 0 has maximal rank $n^2$. Equivalently, we show that for every collection $(S,C, A_{\text{eff}}) \neq (0,0,0)$, the first order term of the one variable map $t \to g(tS,tC, tA_{\text{eff}})$ is non-zero. Computing the coefficient of the first order term we get 
\begin{eqnarray}
    \frac{d}{dt}g(tS,tC, tA_{\text{eff}})|_{t=0} &=&
    G_0 \cdot([S, A_0']+C+A_{\text{eff}}) \cdot G_0^{-1} \\
    &=&
    G_0 \cdot
    \begin{pmatrix}
        \widetilde{A}_{\text{eff}} & S_{12} \cdot (\widetilde{A}_0 -\lambda_0 \mathds{1}_{n-k}) \\
      (  \lambda_0 \mathds{1}_{n-k}-\widetilde{A}_0) \cdot S_{21} & C_{22} \\
    \end{pmatrix}
    \cdot G_0^{-1}.   
\end{eqnarray}
This cannot be zero, since $(S,C, A_{\text{eff}}) \neq (0,0,0)$ and $\widetilde{A}_{0}-\lambda_0 \mathds{1}_{n-k}$ is invertible, which implies that the off diagonal blocks are non-zero if $S \neq 0$. This proves the theorem.

(b) Each eigenvalue $\lambda$ of $A=g(S, C, A_{\text{eff}}) \in \mathcal{V}_0$ is either
 $\lambda=\lambda_0 + \mu$, where $\mu$ is an eigenvalue $A_{\text{eff}}$, or
   an eigenvalue of $\widetilde{A}_0+C_{22}$.
Hence, the eigenvalues of $A$ `born from $\lambda_0$' are in form $\lambda=\lambda_0 + \mu$. 
%More precisely, $(A, \lambda) \in V_0 \subset \mathcal{E}$ if $\lambda=\lambda_0+\mu$ with an eigenvalue $\mu$ of $A_{\text{eff}}$, recall the notations of Proposition~\ref{pr:charactcomp}, in particular, the eigenvalue relation $\mathcal{E}$. 
Note that $\mu$ can be assumed to be arbitrary small by the choice of the neighborhoods $\mathcal{U}_0$ and $\mathcal{V}_0$. 
Hence, $A \in (\Sigma^{(n)}, A_0; \lambda_0)$ if and only if  $\lambda=\lambda_0+\mu$ is a geometric degenerate eigenvalue of $A$, if and only if $\mu$ is a geometric degenerate eigenvalue of $A_{\text{eff}}$, that is, $A_{\text{eff}} \in (\Sigma^{(k)}, 0)$.

The proof of point (c) is very similar, with the additional restriction $\mu=0$ for a geometric degenerate eigenvalue of $A_{\text{eff}}$.
\end{proof}

We call Equation~\eqref{eq:swdec} (more precisely, the inverse of the map $g$) the \emph{(complex) SW decomposition of the matrix $A \in \mathcal{V}_0$}. Theorem~\ref{th:compsw} guarantees that locally  there is a unique SW decomposition, depending on $A$ in analytic way.

 The hermitian version of Theorem~\ref{th:compsw} is analogous, see \cite[Theorem 3.1.2]{PinterSW}. In this case $A_0$ can be diagonalized by a unitary matrix $G_0$, and $C$ and $A_{\text{eff}}$ are hermitian, $S$ is anti-hermitian (that is, $\overline{S}^T+S=0$ or equivalently, $iS $ is hermitian, hence $e^{S}$ is unitary). In this case, our proof for point (a) of Theorem~\ref{th:compsw} literally shows that the map $g$ is a diffeomorphism between neighborhoods $0 \in \R^{n^2}$ and $A_0 \in \text{Herm}(n) \cong \R^{n^2}$, see also \cite{PinterSW}. Note that in \cite{PinterSW} a bit different notation is used, namely, $H_{\text{eff}}$ denotes the \emph{traceless} effective Hamiltonian of a hermitian matrix $H$, its importance is summarized in Section~\ref{ss:efftwofold}. Also note that in \cite{PinterSW} the exponent  the exponent of $e$ is $\pm iS$, since $S$ is chosen to be hermitian instead of anti-hermitian.
 
 The hermitian version also covers the real symmetric case, where $G_0$ is real orthogonal matrix, $C$ and $A_{\text{eff}}$ are real symmetric matrices and $S$ is a real skew-symmetric matrix (that is, $S^T+S=0$, hence $e^{S}$ is real orthogonal). In this case $g$ is a local diffeomorphism between $\R^{(n+1)n/2}$ (at $0$) and $\text{Symm}_{\R}(n) \cong \R^{(n+1)n/2}$. Our proof for  point (a) of Theorem~\ref{th:compsw} literally works also in this case, for convenience we check that the dimensions are equal:
 \begin{equation}\label{eq:equidimsym}
      (S,C,A_{\text{eff}}) \in \C^{k(n-k)} \times \C^{(n-k+1)(n-k)/2}  \times \C^{(k+1)k/2} \cong \C^{(n+1)n/2} \cong \text{Symm}_{\R}(n).
 \end{equation}

 For the complex symmetric case, recall that complex symmetric matrices are not diagonalizable in general. Furthermore, if $A_0 \in \text{Symm}_{\C}(n)$ and $G $ is a complex orthogonal matrix (that is, $G^T \cdot G=\mathds{1}$), then $G^{-1} \cdot A_0 \cdot G \in \text{Symm}_{\C}(n)$, indeed, $(G^{-1} \cdot A_0 \cdot G)^T=G^{-1} \cdot A_0 \cdot G$. Recall that for complex skew-symmetric matrices $S$ (that is, $S^T+S=0$), $e^S$ is complex orthogonal.

 \begin{thm}
     Let $A_0 \in \Sigma_{\text{sym}}^{(n)} \subset \text{Symm}_{\C}(n)$ be a complex symmetric matrix with a strictly $k$-fold degenerate eigenvalue $\lambda_0$. Then there is a complex orthogonal matrix $G_0$ such that 
     \begin{equation}
 G_0^{-1} \cdot   A_0 \cdot G_0= \begin{pmatrix}
        \lambda_0 \mathds{1}_{k} & 0 \\
        0 & \widetilde{A}_{0} \\
    \end{pmatrix} :=A_0',
    \end{equation}
where $\widetilde{A}_0 \in \text{Symm}_{\C}(n-k)$ and $\widetilde{A}_{0}-\lambda_0 \mathds{1}_{n-k}$ has maximal rank $n-k$
 \end{thm}

 \begin{proof} Since $\lambda_0$ is a strictly $k$-fold degenerate eigenvalue of $A_0$, $A_0$ is similar to an $A_0'$ as in Equation~\eqref{eq:normform} by a general linear basis transformation. Since every complex matrix is similar to a symmetric one by \cite[Corollary 1]{Gantmacher1960}, $A_0'$ can be assumed to be symmetric. Moreover, similar symmetric matrices are also orthogonally similar by 
    \cite[Theorem 4.]{Gantmacher1960}, hence, the basis transformation $G_0$ can be chosen to be orthogonal. This proves the statement.
    
    Alternatively, $A_0'$ can be chosen to be the normal form of symmetric matrices constructed in 
    %\cite{CRAVEN, COMp. SYMM. MATR.}
    Ref.~\cite{Craven}. 
    This normal form satisfies the required properties.  
 \end{proof}

 Consider the map $g$ in Equation~\eqref{eq:swdec} for complex symmetric matrices, with $A_0$, $C$ and $A_{\text{eff}}$ complex symmetric, $G_0$ complex orthogonal and $S$ complex skew-symmetric matrices. The dimension calculation in equation~\eqref{eq:equidimsym} shows that $g$ is a map from $\C^{(n+1)n/2}$ to $\text{Symm}_{\C}(n) \cong \C^{(n+1)n/2}$. The proof of Theorem~\ref{th:compsw} literally applies in the complex symmetric case, leading to the following theorem.

\begin{thm}[Complex symmetric Schrieffer--Wolff chart]\label{th:symmsw} \begin{enumerate}[label=(\alph*)]
    \item The map 
\begin{equation}
    g: \C^{(n+1)n/2} \to \text{Symm}_{\C}(n)
\end{equation}
is a complex analytic diffeomorphism (i.e. biholomorphsm)  between a neighborhood $\mathcal{U}_0$ of $0 \in \C^{(n+1)n/2}$ and a neighborhood $\mathcal{V}_0$ of $A_0 \in \text{Symm}_{\C}(n)$.
\item For matrices $A=g(S, C, A_{\text{eff}}) \in \mathcal{V}_0$, $A \in (\Sigma_{\text{sym}}^{(n)}, A_0; \lambda_0)$ holds if and only if $A_{\text{eff}} \in (\Sigma^{(k)}_{\text{sym}}, 0)$. %(Here, slightly abusing  the notation, $(\Sigma^{(n)}, A_0; \lambda_0)$ and $(\Sigma^{k}, 0)$ denote sufficiently small representatives of the corresponding analytic set germs.)
\item If in addition $\lambda_0=0$, that is, $A_0 \in \Sigma'^{(n)}_{\text{sym}}$ is a strictly corank-$k$ matrix, then $A \in (\Sigma'^{(n)}_{\text{sym}}, A_0)$ if and only if $A_{\text{eff}} \in (\Sigma'^{(k)}_{\text{sym}}, 0)$.
\end{enumerate}
\end{thm}

\begin{cor}[Analytic trivial deformations]\label{co:trivfam} Let $A_0 \in \Sigma_{\bullet}^{(n)}$ be a matrix with a strictly $k$-fold degenerate eigenvalue $\lambda_0$, and let $B_0 \in \Sigma_{\bullet}'^{(n)}$ be a strictly corank $k$ matrix, where the lower index $\bullet $ is either empty (general case) or `sym' or `diag' (symmetric and diagonal cases, respectively). We have analytic isomorphisms 
\begin{eqnarray}
    (\Sigma_{\bullet}^{(n)}, A_0; \lambda_0) &\cong& (\Sigma_{\bullet}^{(k)}  \times \C^{\delta(\bullet)}, 0), \label{eqa:sigma} \\
    (\Sigma'^{(n)}_{\bullet}, B_0) &\cong & (\Sigma'^{(k)}_{\bullet}  \times \C^{\delta(\bullet)}, 0),
   \label{eqa:prime} \\ 
    (\widetilde{\Sigma}_{\bullet}^{(n)}, (A_0, \lambda_0)) &\cong& (\widetilde{\Sigma}_{\bullet}^{(k)}  \times \C^{\delta(\bullet)}, (0,0)),  \label{eqa:tilde}
\end{eqnarray}
 where, in each case, the dimension $\delta(\bullet)$ of the family parameter is equal to the difference of the dimensions of the ambient spaces, namely, 
    \begin{eqnarray}
      \delta(\text{empty}) = \dim (\C^{n \times n})-\dim (\C^{k \times k}) & =& n^2-k^2, \\ 
       \delta(\text{`sym'}) = \dim (\text{Symm}_{\C}(n))-\dim (\text{Symm}_{\C}(k))& =& \binom{n+1}{2}-\binom{k+1}{2} \\
      \delta(\text{`diag'}) =  \dim (\text{Diag}_{\C}(n))-\dim (\text{Diag}_{\C}(k))&=& n-k.
    \end{eqnarray}
    \end{cor}

    \begin{proof} First we prove the general case. By points (a) and (b) of the complex SW chart theorem~\ref{th:compsw} (see also Equation~\ref{eq:swdec}), $(\Sigma^{(n)}, A_0; \lambda_0)$ is parametrized as $g(S, C, A_{\text{eff}})$  with $A_{\text{eff}} \in (\Sigma^{(k)}  , 0)$ and $(S, C) \in \C^{n^2-k^2}$. This proves Equation~\eqref{eqa:sigma}. 

Equation~\eqref{eqa:prime} follows from points (a) and (c) of the complex SW chart theorem~\ref{th:compsw} in the same way, but in this case $A_{\text{eff}} \in (\Sigma'^{(k)}  , 0)$.

Equation~\eqref{eqa:tilde} follows from the isomorphism 
\begin{eqnarray}\label{eq:isomtildeprime}
(\widetilde{\Sigma}^{(n)}, (A_0, \lambda_0))  &\cong& (\Sigma'^{(n)} \times \C, (A_0- \lambda_0 \mathds{1}, 0)) \\
 (A, \lambda) &\mapsto& (A-\lambda \mathds{1}, \lambda-\lambda_0).
\end{eqnarray}
Indeed, by point (b), the right hand side of equation~\eqref{eq:isomtildeprime} is isomorphic to $(\Sigma'^{(k)}  \times \C \times \C^{n^2-k^2}, 0 )$, which is isomorphic to $(\widetilde{\Sigma}^{(k)}  \times \C^{n^2-k^2}, (0,0))$, since equation~\eqref{eq:isomtildeprime} can be applied again. 

The symmetric cases can be proved in the same way, using the symmetric SW chart theorem~\ref{th:symmsw}, and that $(S, C) \in \C^{((n+1)n-(k+1)k)/2}$ in the symmetric case.

The diagonal case is straight forward. It can be considered a special case of either the general or the symmetric case with, $S=0$, $A_0$, $C$ and $A_{\text{eff}}$ are diagonal. In this case the family parameter is $C \in \C^{n-k}$.
\end{proof}

\begin{rem}
    For fixed $(S, C)$, $A_{\text{eff}} \mapsto g(S, C, A_{\text{eff}})$ parametrizes a plane (affine subspace) transverse to $\Sigma$, identified with the $k \times k$  matrices. 
    $(\Sigma^{(k)}, 0)$ is already a one parameter trivial family, indeed, there is an isomorphism
    \begin{eqnarray}
        (\Sigma_0^{(k)} \times \C, 0) &\to&  (\Sigma^{(k)}, 0) \\
        (A, \lambda) &\mapsto & A+\lambda \mathds{1},
    \end{eqnarray}
   where $\Sigma_0^{(k)}$ denotes the traceless part  of $\Sigma^{(k)}$. In contrast, $(\Sigma'^{(k)}, 0)$ is not a trivial family, since in this case the trace cannot be varied. The same holds in the symmetric and diagonal cases.
\end{rem}

\begin{rem}
    Corollary~\ref{co:trivfam} also holds in the hermitian and real symmetric cases with a small modification, namely, $\C$ should be replaced by $\R$. The proof is the same, using the hermitian and real symmetric versions of the SW chart.
\end{rem}

\begin{cor}[Irreducible components]\label{cor:irred}
If $\lambda_0$ is a strictly $k$-fold eigenvalue of $A_0 \in \Sigma^{(n)}$, then $(\Sigma, A_0; \lambda_0)$ is an irreducible component of $(\Sigma, A_0)$. The same holds in the symmetric case.
\end{cor}

\begin{proof}
    By Corollary~\ref{co:trivfam} it is enough to prove that $(\Sigma^{(k)}, 0)$ is irreducible. This follows if we show that $(\widetilde{\Sigma}, (0,0))$ is irreducible. Since $(\widetilde{\Sigma}, (0,0)) \cong (\Sigma' \times \C, 0)$ and $(\Sigma', 0)$ is irreducible by \cite[Proposition 1.1]{BrunsVetterBook} (and \cite{SerreGAGA} to switch to the local version), the proof is done. The symmetric case is proved in the same way.
\end{proof}

\begin{cor}[Strictly two-fold degeneracy]\label{co:twofoldsmooth}
If $\lambda_0$ is a strictly two-fold eigenvalue of $A_0 \in \Sigma^{(n)}$, then
\begin{enumerate}[label=(\alph*)]
    \item $(\Sigma, A_0; \lambda_0)$ is non-singular, and
    \item the projection  $p|_{(\widetilde{\Sigma}, (A_0, \lambda_0))}: (\widetilde{\Sigma}, (A_0, \lambda_0)) \to (\Sigma, A_0; \lambda_0)$ is a biholomorphism. 
\end{enumerate}
  
\end{cor}

\begin{proof}
  (a)  Since  $\Sigma^{(2)}=\{\lambda \mathds{1}_2 \ | \ \lambda \in \C \}$, the statement follows from Corollary~\ref{co:trivfam}.

  (b) Recall that $(\widetilde{\Sigma}, (A_0, \lambda_0))$ is non-singular if  $\lambda_0$ is 2-fold geometric degenerate. Consider a sufficiently small open neighborhood $\mathcal{V} \subset {\Sigma}$ of $A_0$ (a representative of the germ $(\Sigma, A_0; \lambda_0)$), and its preimage $\mathcal{U} \subset \widetilde{\Sigma}$. These sets  $\mathcal{U}$ and $\mathcal{V}$ are non-singular, and $p|_{\mathcal{U}}: \mathcal{U} \to \mathcal{V}$ is a bijection. Indeed, for $(A, \lambda) \in \mathcal{U}$, $\lambda$ is a geometric degenerate eigenvalue of $A$ close to $\lambda_0$, but $\lambda_0$ is 2-fold geometric degenerate, hence, $(A, \lambda), (A, \mu) \in \mathcal{U}$ with $\lambda \neq \mu$ cannot happen. By \cite[p. 19]{GriffithsHarris1978}, a bijective holomorphic map between non-singular spaces is a biholomorphism. \end{proof}

\begin{lem}\label{le:strictly2}
     The set $\Xi$ of matrices $A \in \Sigma$ with only one geometric degenerate eigenvalue $\lambda$ which is strictly 2-fold degenerate is an open and dense subset  of $\Sigma$ (with respect to the Euclidean topology).
\end{lem}

\begin{proof}
    $\Xi \subset \Sigma$ is obviously  an open subset. Indeed, the eigenvalues vary continuously \cite[Ch.II., Sec. 1.]{kato2013perturbation}, hence, any matrix $A_0 \in \Sigma$ has a neighborhood in $\Sigma$ such that different eigenvalues of $A_0$ do not become equal to each other inside this neighborhood. If in particular $A_0 \in \Xi$, then this neighborhhod is contained in $\Xi$.

    We show that $\Xi \subset \Sigma$ is dense.  By Proposition 2.1. and Remark 2.3. in \cite{DomokosHermitian}, the set of diagonalizable matrices $\Delta$ is dense in $\Sigma$. Inside $\Delta$, the set $\Delta'$ of matrices with only one geometric degenerate eigenvalue which is strictly 2-fold is dense in $\Delta$, hence, in $\Sigma$. Indeed, starting from a matrix $A_0 \in \Delta$, by an arbitrary small modification of the eigenvalues we obtain a matrix $A \in \Delta'$. Since $\Delta' \subset \Xi$, this proves the lemma.
\end{proof}

\begin{rem}
    Some important comments:
    \begin{itemize}
        \item On may try to prove the density from the Jordan normal form. Starting from the Jordan normal form of an $A_0 \in \Sigma$, it can be reached by an arbitrary small modification of the eigenvalues (with the same value in each Jordan block) that the resulted matrix $A \in \Sigma$ has exactly one geometric degenerate eigenvalue $\lambda$, and its geometric multiplicity is $2$. However, we cannot guarantee that easily that $\lambda$ is a strictly 2-fold degenerate eigenvalue of $A$.
        \item The set of diagonalizable matrices $\Delta$ is not open in $\Sigma$, for example, 
        \begin{equation}
\begin{pmatrix}
    0 & t & 0 \\
    0 & 0 & 0 \\
    0 & 0 & 0 \\
\end{pmatrix}
        \end{equation}
        is in $\Sigma$ for all $t$, but for $t \neq 0$ it is not in $\Delta$. Also if $A_0 \in \Delta'$, there are non-diagonalizable degenerate matrices arbitrary close to $A_0$, but they have only a two-fold algebraic degeneracy, hence, they are not in $\Sigma$.
    \end{itemize}
\end{rem}

\begin{cor}[The projection is generically one-to-one]\label{co:genonetoone} 
    Consider the set $\Upsilon$ of points $A \in \Sigma$ over which the projection $p|_{\widetilde{\Sigma}}: \widetilde{\Sigma} \to \Sigma$ is  an analytic isomorphism (biholomorphism) between smooth analytic set germs, that is: 
    \begin{itemize}
       \item $A$ has exactly one preimage $(A, \lambda)$,
       \item $(\Sigma, A)$ and $(\widetilde{\Sigma}, (A, \lambda))$ are non-singular germs,
       \item $p|_{\widetilde{\Sigma}}: (\widetilde{\Sigma}, (A, \lambda)) \to (\Sigma, A)$ is non-singular.
    \end{itemize}
     This set $\Upsilon \subset \Sigma$ is open and dense (with respect to the Euclidean topology).
\end{cor}

\begin{proof}
  By the above arguments,  $\Xi \subset \Upsilon $ and it is open and dense in $\Sigma$.
\end{proof}

We have similar characterization in the hermitian case, for the local components of $\Sigma_{\text{herm}}$. For a $k$-fold degenerate eigenvalue $\lambda_0$ of $A_0 \in \Sigma_{\text{herm}}$ we define the real analytic set germs $(\Sigma_{\text{herm}}, A_0;\lambda_0) \subset(\Sigma_{\text{herm}}, A_0)$ via the hermitian SW chart theorem, more precisely, the hermitian version of point (b) of Theorem~\ref{th:compsw}: $(\Sigma_{\text{herm}}, A_0;\lambda_0)$ is the image of $g(S, C, A_{\text{eff}})$ with $S, C$ arbitrary (in the given form), and $A_{\text{eff}} \in \Sigma_{\text{herm}}^{(k)}$. Hence, $(\Sigma_{\text{herm}}, A_0;\lambda_0)$ has the direct product structure in the same way as described in Corollary~\ref{co:trivfam} for the complex case. We use this definition from two reasons (1) in the hermitian case the SW is always defined in contrast of the complex case, where we assumed that $\lambda_0$ is strictly $k$-fold degenerate, (2)  using the projection $p|_{\widetilde{\Sigma}_{\text{herm}}}: \widetilde{\Sigma}_{\text{herm}} \to {\Sigma}_{\text{herm}}$, $p(A, \lambda)=A$ for the definition would require a detailed  analysis involving real algebraic geometry.

\begin{prop}\label{pr:hermcomp}
\begin{enumerate}[label=(\alph*)]
\item %Consider the projection $p|_{\widetilde{\Sigma}_{\text{herm}}}: \widetilde{\Sigma}_{\text{herm}} \to {\Sigma}_{\text{herm}}$, $p(A, \lambda)=A$. The image of $(\widetilde{\Sigma}_{\text{herm}}, (A, \lambda))$, denoted by 
$(\Sigma_{\text{herm}}, A_0;\lambda_0)$ is the union of irreducible components of $(\Sigma_{\text{herm}}, A_0)$.
\item If $A_0$ has more then one degenerate eigenvalue, that is, in its label $\kappa=(k_1, \dots, k_l)$ more than one $k_i$ is bigger than 1, then the germ $(\Sigma_{\text{herm}}, A_0)$ of $\Sigma_{\text{herm}} $ at $A_0$ is reducible.
\item $(\Sigma_{\text{herm}}, A_0;\lambda_0)$ is smooth (non-singular) if and only if $\lambda_0$ is a two-fold degenerate eigenvalue of $A_0$.
\item The smooth  points of $\Sigma_{\text{herm}}$ are the degenerate matrices with exactly 2 coincident eigenvalues. 
$\Sigma_{\text{herm}}$ is singular at all other points.
\item If $\lambda_0$ is a two-fold degenerate eigenvalue of $A_0$, then the projection $p|_{\widetilde{\Sigma}_{\text{herm}}}: \widetilde{\Sigma}_{\text{herm}} \to {\Sigma}_{\text{herm}}$, $p(A, \lambda)=A$ is a local analytic diffeomorphism between the non-singular real analytic set germs $(\widetilde{\Sigma}_{\text{herm}}, (A_0, \lambda_0))$ and $(\Sigma_{\text{herm}}, A_0;\lambda_0)$.
\end{enumerate}
\end{prop}

\begin{proof}
(a) $(\Sigma_{\text{herm}}, A_0;\lambda_0) \subset(\Sigma_{\text{herm}}, A_0)$ is the union of irreducible components, since it is a
real analytic subsetgerm of the same dimension. (b) follows directly.

(c) For a two-fold degenerate eigenvalue $\lambda_0$ of $A_0$, $(\Sigma_{\text{herm}}, A_0;\lambda_0)$ is smooth, indeed, the proof of point (a) of Corollary~\ref{co:twofoldsmooth} can be repeated. If the multiplicity of $\lambda_0$ is higher than 2, the link of $(\Sigma, A_0)$ (that is, its intersection with a small sphere around $A_0$ in $\text{Herm}(n)$) is not a sphere, but it is a union of complex projective spaces, see e.g. \cite{ArnoldSelMath1995}, hence $(\Sigma, A_0)$ is singular. This holds also for $(\Sigma_{\text{herm}}, A_0;\lambda_0)$, if the multiplicity $k$ of $\lambda_0$ is higher than 2, indeed, this germ is a trivial deformation of $(\Sigma^{(k)}_{\text{herm}}, 0)$, which is singular by the above argument. This also proves (d).

(e) Follows from point (b) of Corollary~\ref{co:twofoldsmooth}, since the hermition projection is this case is a restriction of a biholomorphism (the complex projection) to real analytic submanifolds. \end{proof}

\section{Multiplicity of the geometric degeneracy set and map germs}\label{s:proofs}

We give formulas for the number of `complex Weyl points' ($\Sigma$-points) of perturbations of holomorphic map germs $f: (\C^3, 0) \to (\C^{n \times n}, A_0)$, $A_0 \in \Sigma$. A general formula is provided in Section~\ref{ss:proof-pullback}. In case of $A_0$ with a strictly $k$-fold  eigenvalue $\lambda_0$ we give a numerical formula in terms of $k$, see Section~\ref{ss:proof-multi}.  These formulas are based on the reduction of $\Sigma$ to $\Sigma'$ via the lifting $p|_{\widetilde{\Sigma}}: \widetilde{\Sigma} \to \Sigma$, and the complex SW chart we introduced. Similar analysis is done for the symmetric and diagonal cases. 

Although the multiplicity of $(\Sigma_{\bullet}',0)$ follows from  more general formulas for determinantal varieties, we present here a direct computation based on free resolutions in the symmetric and general cases, see Section~\ref{ss:proofdiag}--\ref{ss:proofsymm2}. In Section~\ref{ss:proof-notcohmac} we show that $\Sigma$ is not Cohen--Macaulay. 

Throughout this chapter we consider intersections of complex analytic set germs and (small perturbations of) of holomorphic map germs. This requires to take representatives of the germs, wich we do not distinguish from the germ in notation. The dimension of an algebra always means here its dimension as a vector space over the base field, which is the complex field $\C$ in this section. In Appendix~\ref{a:app} we collected the wildly used concepts from local analytic and algebraic geometry.

\subsection{Number of complex Weyl points}\label{ss:proof-pullback}
Here we give an algebraic formula (in terms of the codimension of a suitable ideal) for the number of complex Weyl points, which works for any $A_0$, not only for those with a strictly $k$-fold eigenvalues.

Recall that in the space $\C^{n \times n}$ of complex matrices of type $n \times n $, $\Sigma=\Sigma^{(n)} \subset \C^{n \times n}$ denote the set of matrices having at least one eigenvalue $\lambda_0$ with geometric multiplicity at least 2, that is, the dimension of the eigenspace corresponding to $\lambda_0$ is at least 2. $\Sigma$ is an algebraic variety, see Section~\ref{s:tools}. Let $A_0 \in \Sigma$ be a matrix. For a fixed eigenvalue $\lambda_0$ of $A_0$ with geometric multiplicity at least 2,  $(\Sigma, A_0; \lambda_0)$ denotes the union of components of the analytic set germ $(\Sigma, A_0)$ corresponding to the eigenvalue $\lambda_0$, see Corollaries \ref{co:unionofcomp} and \ref{cor:irred}.

Let $f: (\C^3, 0) \to (\C^{n \times n}, A_0)$ be a holomorphic map germ with $f(0)=A_0 \in \Sigma$. Assume that $f$ is isolated with respect to $(\Sigma, A_0; \lambda_0)$, that is, $f^{-1}(\Sigma, A_0; \lambda_0) =\{0\}$ holds for a sufficiently small representative of $f$.

Let $f_t$ be a perturbation of $f$ generic (transverse) with respect to $(\Sigma, A_0; \lambda_0)$, that is, it intersects $(\Sigma, A_0; \lambda_0)$ transversely at non-singular points. Then the elements of $ f_t^{-1}(\Sigma, A_0; \lambda_0)$ can be interpreted as the complex Weyl points (i.e. generic degeneracy points) of $f_t$ \emph{born from the the degeneracy corresponding to $\lambda_0$.}

\begin{thm}\label{th:pullback}
     
In the above setup,
    \begin{equation}\label{eq:pullback}
        \sharp f_t^{-1}(\Sigma, A_0; \lambda_0)=\dim \frac{\mathcal{O}_4}{J},
    \end{equation}
    where $J \subset \mathcal{O}_4$ is the ideal generated by the $(n-1) \times (n-1) $ minors \begin{equation}
        M_{ij}(f(x)-(\lambda +\lambda_0) \mathds{1})
    \end{equation}
    of $f(x)-(\lambda +\lambda_0)\mathds{1}$.

   In particular, the number of complex Weyl points $\sharp f_t^{-1}(\Sigma, A_0; \lambda_0)$ born from $\lambda_0$ is the same for every generic perturbation $f_t$ of $f$. 
\end{thm}

\begin{proof}
Define the map germ \begin{eqnarray}
    \widetilde{f}: (\C^{3} \times \C, (0,0)) &\to& ( \C^{n \times n} \times \C, (A_0,\lambda_0)) \\
   \widetilde{f}(x, \lambda) & =& (f(x), \lambda+\lambda_0 ).
\end{eqnarray}

 A perturbation $f_t$ of $f$ induces a perturbation $\widetilde{f}_t$ of $\widetilde{f}$ as 
 \begin{equation}
     \widetilde{f}_t(x, \lambda):=(f_t(x),\lambda + \lambda_0 ).
 \end{equation}

  Since $f_t$ is generic with respect to $(\Sigma, A_0; \lambda_0)$, it implies that $\widetilde{f}_t$ is generic with respect to $(\widetilde{\Sigma}, (A_0, \lambda_0))$.  This is because the projection \begin{eqnarray}
      p|_{(\widetilde{\Sigma}, (A_0, \lambda_0))}: (\widetilde{\Sigma}, (A_0, \lambda_0)) &\to& (\Sigma, A_0; \lambda_0) \\
      p(A, \lambda)&=&A
  \end{eqnarray}
   generically 1-to-1, see Corollary~\ref{co:genonetoone}. This implies that $\sharp \widetilde{f}_t^{-1}(\widetilde{\Sigma}, (A_0, \lambda_0))=\sharp f_t^{-1}(\Sigma, A_0; \lambda_0) $.

Observe that  $J=\widetilde{f}^{*}(I(\widetilde{\Sigma}, (A_0, \lambda_0))) \cdot \mathcal{O}_4$, indeed, the the minors $M_{ij}(A-\lambda \mathds{1})$ generate the vanishing ideal $I(\widetilde{\Sigma}, (A_0, \lambda_0))$ of $(\widetilde{\Sigma}, (A_0, \lambda_0)) \subset (\C^{n \times n+1}, (A_0, \lambda_0))$. Since $(\widetilde{\Sigma}, (A_0, \lambda_0))$ is Cohen--Macaulay,  by Proposition~\ref{pr:everypert} we have
\begin{equation}
\sharp \widetilde{f}_t^{-1}(\widetilde{\Sigma}, (A_0, \lambda_0))= \dim \frac{\mathcal{O}_4}{\widetilde{f}^{*}(I(\widetilde{\Sigma}, (A_0, \lambda_0))) \cdot \mathcal{O}_4}=\dim \frac{\mathcal{O}_4}{J},
\end{equation}
proving the theorem.
\end{proof}

\begin{rem}
Although the proof of Theorem~\ref{th:pullback} is based on the general theory explained in Section~\ref{ss:intsect}, we emphasize that the ideal $J$ used in the theorem is not the pull-back of the vanishing ideal of $(\Sigma, A_0; \lambda_0)$. Indeed, the zero locus of the minors  $M_{ij}(A)$ is $\Sigma'=\Sigma'^{(n)} \subset \C^{n \times n}$, that is, the set of the matrices $A$ of corank at least 2, see Section~\ref{ss:matrvar}. Hence, with the notation $f'(x,\lambda)=f(x)-(\lambda +\lambda_0)\mathds{1}$ we have that 
\begin{equation}
    J=f'^{*}(I(\Sigma', A_0-\lambda_0 \mathds{1})) \cdot \mathcal{O}_4=\widetilde{f}^{*}(I{(\widetilde{\Sigma}, (A_0, \lambda_0))}) \cdot \mathcal{O}_4.
\end{equation}
%$J$ equals to the pull-back ideal $f'^{*}(I(\Sigma', A_0-\lambda_0 \mathds{1})) \cdot \mathcal{O}_4$, where $I(\Sigma', A_0-\lambda_0 \mathds{1})$ is the ideal of the germ of $\Sigma'$ at $A_0-\lambda_0 \mathds{1}$ in $\mathcal{O}_{(\C^{n \times n}, A_0-\lambda_0 \mathds{1})}$, and $f'(x,\lambda)=f(x)-(\lambda +\lambda_0)\mathds{1}$.
We will see in Section~\ref{ss:proof-notcohmac} that the pull-back ideal $f^*(I(\Sigma, A_0; \lambda_0)) \cdot \mathcal{O}_3 \subset \mathcal{O}_3$ cannot be used directly to compute the number $\sharp f_t^{-1}(\Sigma, A_0; \lambda_0)$, because $\Sigma$ is not Cohen--Macaulay, see Theorem~\ref{co:notcohmac} and \ref{th:notwork00}.
\end{rem}

\begin{remark}\label{re:total}
  We get the total number of complex Weyl points $\sharp f_t^{-1}(\Sigma, A_0) $ born from the degeneracy at $A_0$ as the sum of the right hand sides of Equation~\eqref{eq:pullback} for every eigenvalue $\lambda_0$ of $A_0$ with geometric multiplicity at least 2. Indeed, the components of $(\Sigma, A_0) $ correspond to these eigenvalues by Remark~\ref{re:everyirred}.
\end{remark}

\begin{rem}
    The finiteness of $\dim (\mathcal{O}_4/J)$ is equivalent with the fact that $f$ is isolated with respect to $(\Sigma, A_0, \lambda_0)$. Indeed, both are equivalent with the fact that the vanishing locus $V(J) \subset (\C^4, 0)$ of $J$ is only one point (the origin). See Example~\ref{ex:onepoint}, cf. \cite[Theorem D.5.]{MondBook}.
\end{rem}

The analogues of Theorem~\ref{th:pullback} for symmetric and diagonal matrix families can be formulated and proved similarly. In these cases, $f: (\C^2, 0) \to (\text{Symm}_{\C}(n), A_0)$, $A_0 \in \Sigma_{\text{sym}}$ and $f: (\C, 0) \to \text{Diag}_{\C}(n)$, $A_0 \in \Sigma_{\text{diag}}$, respectively,  and the perturbation $f_t$ is considered in the same target space. Indeed,  the dimension of the source should be equal to the codimension of $\Sigma_{\bullet} $ to obtain isolated intersection points generically. The definition of the ideal $J$ is literally the same as in Theorem~\ref{th:pullback}, but it is an ideal in $\mathcal{O}_3$ in the symmetric case and in $\mathcal{O}_2$ in the diagonal case.

\subsection[Multiplicity at strictly k-fold degeneracy points]{Multiplicity at strictly $k$-fold degeneracy points}\label{ss:proof-multi}

Here we restrict the setup of Section~\ref{ss:proof-pullback} to map germs $f: (\C^3, 0) \to (\C^{n \times n}, A_0)$, where $\lambda_0$ is a strictly $k$-fold eigenvalue of $A_0 \in \Sigma$ (that is, both the algebraic and geometric multiplicity of $\lambda_0$ is $k$) and $f$ is \emph{generic} with respect to $(\Sigma, A_0; \lambda_0)$. This latter condition means that for a  perturbation $f_t$ of $f$ which is generic with respect $(\Sigma, A_0; \lambda_0)$, the number of preimages  $\sharp f_t^{-1}(\Sigma, A_0; \lambda_0)$ is equal to the multiplicity $\text{mult}(\Sigma, A_0; \lambda_0)$, see Appendix~\ref{ss:intsect}.

In this context, the main goal of this and the next subsections is to prove the Multiplicity Theorem~\ref{th:multipl}. In this subsection we give equivalent reformulations for this theorem, and the proofs are in the next subsections, separately for the three cases. 

\begin{thm}[Multiplicity Theorem]\label{th:multipl}
   Let $A_0 \in \Sigma^{(n)}$ be a matrix with a strictly $k$-fold degenerate eigenvalue $\lambda_0$. Let $B_0 \in \Sigma'^{(n)}$ be a strictly corank $k$ matrix.  Then,
    \begin{enumerate}
        \item \begin{equation}
            \text{mult}(\Sigma^{(n)}, A_0; \lambda_0)=
            \text{mult}(\widetilde{\Sigma}^{(n)}, (A_0, \lambda_0))=
            \text{mult}{(\Sigma'^{(n)}, B_0)}=\frac{k^2(k^2-1)}{12}.
        \end{equation}
        \item If furthermore $A_0$ and $B_0$ are complex symmetric matrices, then
        \begin{equation}
            \text{mult}(\Sigma_{\text{sym}}^{(n)}, A_0; \lambda_0)=
            \text{mult}(\widetilde{\Sigma}^{(n)}_{\text{sym}}, (A_0, \lambda_0))=
            \text{mult}{(\Sigma'^{(n)}_{\text{sym}}, B_0)}= \frac{k(k^2-1)}{6}.
            \end{equation}
            \item If furthermore $B_0$ and $A_0$ are complex diagonal matrices, then
        \begin{equation}
            \text{mult}(\Sigma_{\text{diag}}^{(n)}, A_0; \lambda_0)=
            \text{mult}(\widetilde{\Sigma}^{(n)}_{\text{diag}}, (A_0, \lambda_0))=
            \text{mult}{(\Sigma'^{(n)}_{\text{diag}}, B_0)}=\frac{k (k-1)}{2}.
        \end{equation}
    \end{enumerate}
    In particular, the above multiplicities do not depend on the size $n \times n$ of the ambient matrix spaces.
\end{thm}

Let $f: (\C^3, 0) \to (\C^{n \times n}, A_0)$ be a map germ as in Therorem~\ref{th:pullback}, that is, $f(0)=A_0 \in \Sigma$ and $f$ is isolated with respect to $(\Sigma, A_0; \lambda_0)$, that is, $f^{-1}(\Sigma, A_0; \lambda_0) =\{0\}$ holds for a sufficiently small representative of $f$. Assume furthermore that $\lambda_0$ is a strictly $k$-fold eigenvalue of $A_0$. %and $f$ is \emph{generic} with respect to $(\Sigma^{(n)}, A_0; \lambda_0)$. 
Recall from Section~\ref{ss:homo} that such map germ $f$ is generic with respect to $(\Sigma, A_0; \lambda_0)$ if 
\begin{equation}
 \sharp f_t^{-1}(\Sigma, A_0; \lambda_0)=\text{mult}(\Sigma, A_0; \lambda_0)
\end{equation}
 holds for any perturbation of $f$ generic with respect to $(\Sigma, A_0; \lambda_0)$. %Consequently, for generic germs, the number of complex Weyl points is equal to $k^2(k^2-1)/12$. 
Hence, we have the following consequence of point 1. of Theorem~\ref{th:multipl}.

\begin{cor}\label{co:generic}
     In the above set-up,
     \begin{equation}
           \sharp f_t^{-1}(\Sigma, A_0; \lambda_0)=\frac{k^2(k^2-1)}{12}.
     \end{equation}
\end{cor}
The analogue of this corollary holds for 2-parameter generic smmyetric and 1-parameter generic diagonal families, according to the codimension of $\Sigma_{\text{sym}} \subset \text{Symm}_{\C}(n)$ (which is 2) and of $\Sigma_{\text{diag}} \subset \text{Diag}_{\C}(n)$ (which is 1). 

If in addition $n=k$ and $A_0=0$ (hence, $\lambda_0=0$), $(\Sigma, 0;0)=(\Sigma, 0)$ is homogeneous. By Corollary~\ref{co:finitegeneric2}, a linear map $f: (\C^3, 0) \to (\C^{k \times k}, 0)$ is generic with respect to $(\Sigma, 0)$ if and only if it is isolated with respect to $(\Sigma, 0)$. Hence, we have the following consequence  of Theorem~\ref{th:multipl} (we formulate it in all the three cases because of its importance in the proof and in the applications).

\begin{cor}\label{co:kanegyzet_symm}
\begin{enumerate}
\item For a linear family $f: (\C^3, 0) \to (\C^{k \times k}, 0)$ isolated with respect to $(\Sigma, 0)$ and for a generic perturbation $f_t$ of $f$  which is generic with respect to $(\Sigma, 0)$, we have 
 \begin{equation}
           \sharp f_t^{-1}(\Sigma, 0)=\frac{k^2(k^2-1)}{12}.
     \end{equation} 
    \item For a linear family $f: (\C^2, 0) \to (\text{Symm}_{\C}(k), 0)$ isolated with respect to $(\Sigma_{\text{sym}}, 0)$ and for a generic perturbation $f_t$ of $f$ in $\text{Symm}_{\C}(k) $ which is generic with respect to $(\Sigma_{\text{sym}}, 0)$, we have 
 \begin{equation}
           \sharp f_t^{-1}(\Sigma_{\text{sym}}, 0)=\frac{k(k^2-1)}{6}.
     \end{equation}
    \item  For a linear family $f: (\C, 0) \to (\text{Diag}_{\C}(k), 0)$ isolated with respect to $(\Sigma_{\text{diag}}, 0)$ and for a generic perturbation $f_t$ of $f$ in $\text{Diag}_{\C}(k) $ which is generic with respect to $(\Sigma_{\text{diag}}, 0)$, we have 
 \begin{equation}
           \sharp f_t^{-1}(\Sigma_{\text{diag}}, 0)=\frac{k(k-1)}{2}.
     \end{equation}
\end{enumerate}
     \end{cor}

     Of course, the analogue of Corollaries \ref{co:generic} and \ref{co:kanegyzet_symm} can be formulated for   $\Sigma'_{\bullet}$ and $\widetilde{\Sigma}_{\bullet}$ in context of map germs into their ambient spaces, where  $\bullet$ denotes each of the lower indices $\text{diag}$, $\text{sym}$ or empty (general case).

     However, the Multiplicity Theorem~\ref{th:multipl} can be reduced to the special case formulated Corollary~\ref{co:kanegyzet_symm} in two steps.  On the one hand, by Corollary~\ref{co:trivfam}, $(\Sigma_{\bullet}, A_0;\lambda_0)$, $(\widetilde{\Sigma}_{\bullet}, (A_0, \lambda_0))$ and $(\Sigma'_{\bullet}, B_0)$ are trivial deformations for an $A_0$ with a strictly $k$-fold degenerate eigenvalue $\lambda_0$, and for a strictly corank-2 matrix $B_0$, where  $\bullet$ denotes each of the lower indices $\text{diag}$, $\text{sym}$ or empty (general case). Therefore, the following multiplicities equal.

     \begin{cor}\label{co:effmodmult} In the above set up,
         \begin{eqnarray}
             \text{mult}(\Sigma_{\bullet}^{(n)}, A_0; \lambda_0)&=& \text{mult}(\Sigma_{\bullet}^{(k)}, 0), \\
            \text{mult}(\widetilde{\Sigma}_{\bullet}^{(n)}, (A_0, \lambda_0))&=& \text{mult}(\widetilde{\Sigma}_{\bullet}^{(k)}, (0,0)), \\
            \text{mult}{(\Sigma_{\bullet}'^{(n)}, B_0)}&=& \text{mult}{(\Sigma_{\bullet}'^{(k)}, 0)}
         \end{eqnarray}
     \end{cor}

     On the other hand, we show that the right hand sides are equal to each other.

     \begin{prop}\label{pr:primetilde} In the above set up, for each fixed type  $\bullet=\{\}, \text{sym}, \text{diag}$ we have
         \begin{equation}
             \text{mult}(\Sigma_{\bullet}^{(k)}, 0)=\text{mult}(\widetilde{\Sigma}_{\bullet}^{(k)}, (0,0))=\text{mult}{(\Sigma_{\bullet}'^{(k)}, 0)}.
         \end{equation}
     \end{prop}

     \begin{proof}
      We formulate the proof in the general case, and the symmetric and diagonal cases can be proved in the same way.
      
      Consider a linear map $f: (\C^3, 0) \to (\C^{k \times k}, 0)$ , and define 
      \begin{eqnarray}
          \widetilde{f}: (\C^3 \times \C, 0) &\to& (\C^{k \times k} \times \C, 0) \\
          \widetilde{f}(x, \lambda) &=& (f(x), \lambda)           
      \end{eqnarray}
      and 
       \begin{eqnarray}
    f': (\C^3 \times \C, 0) &\to& (\C^{k \times k} , 0) \\
          f'(x, \lambda) &=& (f(x)- \lambda \mathds{1}_k)           
      \end{eqnarray}
      Observe that the following are equivalent:
      \begin{enumerate}
          \item $f$ is isolated with respect to $(\Sigma, 0)$,
          \item $\widetilde{f}$ is isolated with respect to $(\widetilde{\Sigma}, (0,0))$, 
          \item $f'$ is isolated with respect to $(\Sigma',0)$.
      \end{enumerate} 
      Indeed, $f(x)-\lambda \mathds{1}_k \in \Sigma'$ holds for $(x, \lambda)$ close to $(0,0)$ if and only if  $(f(x), \lambda) \in \widetilde{\Sigma}$, if and only if $f(x) \in \Sigma$ and $\lambda$ is an eigenvalue of $f(x) $ with geometric multiplicity at least 2. That is, $\text{rk}(f(x)-\lambda \mathds{1}_k) \leq n-2$, equivalently, $f'(x) \in \Sigma'$. Hence, points 1--3. are equivalent with the fact that $f(x) \in \Sigma$ implies $x=0$ (if $x$ is close to 0).

   By Corollary~\ref{co:finitegeneric2}, a linear map is generic with respect to an analytic set germ if and only if it is isolated with respect to it. Hence, choosing an $f$ isolated with respect to $\Sigma$, $\text{mult}(\Sigma, 0)=\sharp f_t^{-1}(\Sigma, 0)$ holds for any perturbation $f_t$ of $f$ generic with respect to $\Sigma$. Such perturbation determine a perturbation $\widetilde{f}_t(x, \lambda)=(f_t(x), \lambda)$ of $\widetilde{f}$ and $f_t'(x, \lambda)=f_t(x)-\lambda \mathds{1}_k$ of $f'$. 

   Since by Corollary~\ref{co:genonetoone} the projection $\widetilde{\Sigma} \to \Sigma$ is generically 1-to-1, by a generic choice of $f_t$ it can be assumed that the perturbation $\widetilde{f}_t$ is generic with respect to $(\widetilde{\Sigma}, (0,0))$ and $f'_t$ is generic with respect to $(\Sigma', 0)$, hence,  
   \begin{equation}
      \sharp f_t^{-1}(\Sigma, 0)= \sharp \widetilde{f}_t^{-1} (\widetilde{\Sigma}, (0,0))= \sharp f'^{-1}_t(\Sigma', 0).
   \end{equation}
   Since (by definition) these terms are equal to $\text{mult}(\Sigma, 0)$, $\text{mult}(\widetilde{\Sigma}, (0,0))$ and $\text{mult}(\Sigma', 0)$, respectively, the proposition is proved.
     \end{proof}

Therefore, by Corollary~\ref{co:effmodmult} and Proposition~\ref{pr:primetilde}, to prove Theorem~\ref{th:multipl} it is enough to show that $\text{mult}(\Sigma'_{\bullet}, 0)$ is equal to $k^2(k^2-1)/12$ in the general case, $k(k^2-1)/6$ in the symmetric case and $k(k-1)/2$ in the diagonal case. But $(\Sigma'_{\bullet}, 0)$ is Cohen--Macaulay, hence, its multiplicity can be expressed by the dimension of a quotient algebra corresponding to a linear map, isolated with respect to $(\Sigma'_{\bullet}, 0)$. Namely, the following Reduced Multiplicity Theorem implies Multiplicity Theorem~\ref{th:multipl}.

\begin{thm}[Reduced Multiplicity Theorem]\label{th:multiprime}
\begin{enumerate}
\item For a linear family $f: (\C^4, 0) \to (\C^{k \times k}, 0)$ isolated with respect to $(\Sigma', 0)$ and for a generic perturbation $f_t$ of $f$   which is generic with respect to $(\Sigma', 0)$, we have 
 \begin{equation}
           \dim \frac{\mathcal{O}_4}{f^*(I_{k-1}) \cdot \mathcal{O}_4}=\frac{k^2(k^2-1)}{12},
     \end{equation} 
    \item For a linear family $f: (\C^3, 0) \to (\text{Symm}_{\C}(n), 0)$ isolated with respect to $(\Sigma'_{\text{sym}}, 0)$ and for a generic perturbation $f_t$ of $f$ in $\text{Symm}_{\C}(n) $ which is generic with respect to $(\Sigma'_{\text{sym}}, 0)$, we have 
 \begin{equation}
            \dim \frac{\mathcal{O}_3}{f^*(I_{k-1}) \cdot \mathcal{O}_3}=\frac{k(k^2-1)}{6}.
     \end{equation}
    \item  For a linear family $f: (\C^2, 0) \to (\text{Diag}_{\C}(n), 0)$ isolated with respect to $(\Sigma'_{\text{diag}}, 0)$ and for a generic perturbation $f_t$ of $f$ in $\text{Diag}_{\C}(n) $ which is generic with respect to $(\Sigma'_{\text{diag}}, 0)$, we have 
 \begin{equation}
           \dim \frac{\mathcal{O}_2}{f^*(I_{k-1}) \cdot \mathcal{O}_2}=\frac{k(k-1)}{2}.
     \end{equation}
\end{enumerate}
Here $I_{k-1}$ denotes the ideal generated by the $(k-1) \times (k-1)$ minors in the local algebras of each of the three matrix spaces $(\C^{n \times n}, 0)$, $(\text{Symm}_{\C}(n), 0)$ and $(\text{Diag}_{\C}(n), 0)$.
\end{thm}

\begin{rem}
This reduction actually depends on the fact that the projection $(\widetilde{\Sigma}, (0,0)) \to (\Sigma, 0)$, $(A, \lambda) \mapsto A$ is `generic' in sense of the proof of Proposition~\ref{pr:primetilde}, that is, a generic linear map $f$ with respect to $\Sigma$ induces generic linear maps $\widetilde{f}$ with respect to $\widetilde{\Sigma}$ and $f'$ with respect to $\Sigma'$.
\end{rem}

 The general (1.) and symmetric (2.) cases of Reduced Multiplicity Theorem~\ref{th:multiprime} are special cases of more general results. Namely, formulas for the multiplicity are known for the vanishing locus of the $(n-r) \times (m-r)$ minors in $\C^{n \times m}$, see e.g. \cite[Proposition 12.]{HarrisTu} or \cite{HerzogTrung}.
 These formulas give the multiplicity of $(\Sigma'_{\bullet}, 0)$ for $n=m$ and $r=1$ (in \cite{HarrisTu} the symmetric case is also discussed).
 
 In the next subsections we provide an independent, more direct and  more elementary proof in our cases.
First we discuss the diagonal case from different perspectives, serving as an elementary illustration of these concepts.

\subsection{Proof of the diagonal case}\label{ss:proofdiag}

The diagonal case serves as a baby version to illustrate our methods. We prove point 3. of Multiplicity Theorem~\ref{th:multipl} in essentially two different ways. 

First we prove point 3. of Corollary~\ref{co:kanegyzet_symm}, that is, we compute directly the multiplicity of $(\Sigma_{\text{diag}}, 0)$. This implies point 3. of the Multiplicity Theorem~\ref{th:multipl} via Corollary~\ref{co:effmodmult} (about the multiplicity at other points) and Proposition~\ref{pr:primetilde} (about the equality of the multiplicities of $\Sigma_{\text{diag}}$, $\widetilde{\Sigma}_{\text{diag}}$ and $\Sigma'_{\text{diag}}$).

\begin{proof}[First proof of point 3. of Corollary~\ref{co:kanegyzet_symm}]

Introduce the simplified notation $A=\text{diag}(a_{ii})_i=(a_i)_i$ for the entries of a diagonal matrix $A \in \text{Diag}_{\C}(k) \cong \C^k$. Then 
\begin{equation}
    \Sigma_{\text{diag}}=\{ A \in \C^k \ | \ \exists i \neq j \ a_i=a_j \}=
    \{ A \in \C^k \ | \ \prod_{1 \leq i < j \leq k} (a_i-a_j)=0 \}.
\end{equation}
Hence, in the diagonal case, $\Sigma_{\text{diag}}$ is the hypersurface defined as the vanishing locus of the polynomial $F(A):=\prod_{1 \leq i < j \leq k} (a_i-a_j)$. In particular, $\Sigma_{\text{diag}}$ is Cohen--Macaulay.

Let $f: (\C, 0) \to (\C^k, 0)$ be a linear map, generic (isolated) with respect to $\Sigma_{\text{diag}}$. That is, $f(x)=(a_i  x)_i $, where $a_i \in \C$,  $a_i \neq a_i$ for $i \neq j$. A perturbation in special form $f_t(x)=(a_ix+b_i)_i$ ($b_i \in \C$) is generic with respect to $(\Sigma, 0)$ if and only if the solutions $x_{ij}$ of the equations $a_ix+b_i =a_jx+b_j$ for $i \neq j$ are pairwise different, that is, the lines intersect each other at different points. These solutions $x_{ij}$ are exactly the elements of the set 
\begin{equation}
    f_t^{-1}(\Sigma_{\text{diag}}, 0)=\{x \in \C \ | \ \prod_{1 \leq i < j \leq k} ((a_ix-b_i)-(a_jx-b_j))=0 \},
\end{equation}
 and their number is $\binom{k}{2}=k(k-1)/2$, hence, this is equal to $\text{mult}(\Sigma_{\text{diag}}, 0)$, by definition. This finishes the proof. \end{proof}

\begin{proof}[Second proof of point 3. of Corollary~\ref{co:kanegyzet_symm}]
 Alternatively, since $\Sigma_{\text{diag}}$ is Cohen--Macaulay, we can consider the pull-back 
 \begin{equation}
      f^*(F)(x)=F(f(x))=\prod_{1 \leq i < j \leq k} (a_ix-a_jx)=\prod_{1 \leq i < j \leq k} (a_i-a_j) x^{k(k-1)/2} \in \mathcal{O}_1.
 \end{equation}
 Then, 
 \begin{equation}
     \text{mult}(\Sigma_{\text{diag}}, 0)=\dim \frac{\mathcal{O}_1}{f^*(F) \cdot \mathcal{O}_1}=\binom{k}{2}=\frac{k(k-1)}{2}.
 \end{equation}

 \end{proof}

  Next we prove point 3. of the Reduced Multiplicity Theorem~\ref{th:multiprime}, that is, we compute the multiplicity of $(\Sigma'_{\text{diag}}, 0)$. This implies point 3. of the Multiplicity Theorem~\ref{th:multipl} via Corollary~\ref{co:effmodmult} (about the multiplicity at other points) and Proposition~\ref{pr:primetilde} (about the equality of the multiplicities of $\Sigma_{\text{diag}}$, $\widetilde{\Sigma}_{\text{diag}}$ and $\Sigma'_{\text{diag}}$).

\begin{proof}[Proof of point 3. of the Reduced Multiplicity Theorem~\ref{th:multiprime}]  Consider
 \begin{equation}
     \Sigma'_{\text{diag}} =\{ A \in \C^k \ | \ \exists i \neq j \ a_i=a_j=0 \} =\{ A \in \C^k \ | \ \forall i \ 
   %  a_1 a_2 \dots \hat{a}_i \dots a_k
   \prod_{j \neq i} a_j  =0 .\} \end{equation}
  %   Here the notation $\hat{a}_i$ means that $a_i$ is missing from the product. 
     Note that $\prod_{j \neq i} a_j=M_{ii}(A)$ is the $(k-1) \times (k-1)$ minor corresponding to $a_{ii}$. The other minors $M_{ij}(A)$, $i \neq j$ are automatically 0 in the diagonal case. In other words, the vanishing ideal of $\Sigma'_{\text{diag}}$ is the ideal $I_{k-1}=I\langle M_{ii} \rangle \subset \C^k$ generated by the minors $M_{ii}$. %\notegergo{Cf. Example C.8 in Mond-Nuno book. What does it mean `set theoretically complete intersection'? Is $\Sigma'_{\text{diag}}$ this? See remark below.}
     
     Consider a linear map 
     \begin{eqnarray}
         f: \C^2 &\to& \C^k \\
         f(x, y) &=& (a_i x +b_i y)_i.
     \end{eqnarray}
     
     $f$ is isolated with respect to $\Sigma'_{\text{diag}}$, that is, $f^{-1}(\Sigma'_{\text{diag}})=\{0\}$ if and only if the vectors $(a_i, b_i)$ are pairwise linear independent, for different values of $1 \leq i \leq k$. Perturb $f$ as $f_t(x,y)=(a_i x+b_iy +c_i t)_i$. For a fixed $t$ (close but not equal to 0), $f_t^{-1}(\Sigma'_{\text{diag}})$ is the union of the set of solutions of the systems 
     \begin{equation}
       \left.  \begin{array}{ccc}
             a_i x+b_iy +c_i t & = & 0 \\
             a_j x+b_jy +c_j t & = & 0 
         \end{array} \right\}
     \end{equation}
     for each pair of indices $i < j$. Because of the linear independence of the vectors $(a_i, b_i)$, each system has exactly one solution, hence, their number is
     \begin{equation}
         \text{mult}(\Sigma'_{\text{diag}}, 0)= \sharp f_t^{-1}(\Sigma'_{\text{diag}}) =\binom{k}{2}=\frac{k(k-1)}{2},
     \end{equation}
     proving the statement.
     \end{proof}
     
     Alternatively, we can avoid perturbation by considering the codimension of the pull-back ideal $I_{k-1}(f):=f^*(I\langle M_{ii} \rangle) \cdot \mathcal{O}_2$. It is the ideal in $\mathcal{O}_2$ generated by the functions $M_{ii} \circ f $, which are homogeneous polynomials of degree $k-1$ in two variables $x$ and $y$. We show below that  $M_{ii} \circ f $ form a basis in the complex vector space $\C[x, y]_{k-1} \cong \C^{k}$ of homogeneous polynomials of degree $k-1$, if $f$ is isolated with respect to $\Sigma'_{\text{diag}}$. This implies that every monomial of degree at least $k-1$ is contained in $I_{k-1}(f)$, hence, $I_{k-1}(f)$ is the space of holomorphic function germs of order at least $k-1$. Therefore, as vector spaces,
     \begin{equation}
\frac{\mathcal{O}_2}{I_{k-1}(f)} \cong \bigoplus_{d=0}^{k-2} \C[x,y]_d,
     \end{equation}
     hence its dimension is $\sum_{d=0}^{k-2} (d+1)=k(k-1)/2$. Since we obtained $\text{mult}(\Sigma'_{\text{diag}}, 0)$, this verifies that $(\Sigma'_{\text{diag}}, 0)$ is Cohen--Macaulay -- or, we obtained a second proof for point 3. of the Reduced Multiplicity Theorem~\ref{th:multiprime}, if we assume that  $(\Sigma'_{\text{diag}}, 0)$ is Cohen--Macaulay. 
     
     There is still left to prove that
     \begin{prop}
         If $f$ is isolated with respect to $\Sigma'_{\text{diag}}$, then the functions $M_{ii} \circ f $ form a basis in the complex vector space $\C[x, y]_{k-1}$ of homogeneous polynomials of degree $k-1$.
     \end{prop}

     \begin{proof}
         We show that an arbitrary $g \in \C[x, y]_{k-1}$ can be written as the linear combination of the polynomials $M_{ii} \circ f $. As every 2-variable homogeneous polynomial, $g$ decomposes as the product of linear polynomials $g_j$, that is, $g= \prod_{j=1}^{k-1} g_j$. (Indeed, $g^{-1}(0)$ is the union of 1 dimensional subspaces in $\C^2$.) Let us introduce the notations $m_i:=M_{ii} \circ f $ and $f_i=a_ix+b_iy$, then $m_i=\prod_{j \neq i} f_j$ is the decomposition of $m_i$ as the product of linear terms. 

         For $i=1, \dots, k-1$, $m_i$ and $m_{i+1}$ has $k-2$ common factors $f_j$ ($j \neq i, i+1 $) and two different factors $f_i$ and $f_{i+1}$. Because of the assumption that $f$ is isolated with respect to $\Sigma'_{\text{diag}}$, $f_i$ and $f_{i+1}$ are linearly independent, hence, there are complex coefficients $s$ and $t$ such that $sf_i+tf_{i+1}=g_1$. We define
         \begin{equation}
             m_i^{(2)}:=tm_i+s m_{i+1}=g_1  \cdot \prod_{j \neq i, i+1} f_j.
         \end{equation}
         We obtained $k-1$ new homogeneous polynomials $ m_i^{(2)}$ of degree $k-1$ ($i=1, \dots, k-1$). Each pair $m_i^{(2)}$, $m_{i+1}^{(2)}$, has $k-2$ common factors $g_1$ and $f_j$ ($j \neq i, i+1, i+2$), and they have two different factors: $f_i$ in $m_{i+1}^{(2)}$ and $f_{i+2}$ in $m_{i}^{(2)}$. We repeat the above argument with a choice of (new) coefficients $s$ and $t$ satisfying $sf_i+tf_{i+2}=g_2$ to obtain
          \begin{equation}
             m_i^{(3)}:=tm_i^{(2)}+s m_{i+1}^{(2)}=g_1 g_2 \cdot \prod_{ j \neq i, i+1, i+2} f_j,
         \end{equation}
         for $i=1, \dots, k-2$. This argument can be repeated inductively in steps $l=1, \dots, k$:
         \begin{equation}
              m_i^{(l+1)}:=tm_i^{(l)}+s m_{i+1}^{(l)}=\prod_{\alpha=1}^{l} g_{\alpha} \cdot \prod_{ j \neq i,\dots, i+l} f_j,
         \end{equation}
with coefficients $s$, $t$ such that $s f_i + tf_{i+l}=g_{l}$. Then, by induction, $m_1^{(k)}=g$ is the linear combination of $m_i$ ($i=1, \dots, k$).
         \end{proof}

         \begin{remark}
             In the diagonal case we can see directly that $I_{k-1}$ is a radical ideal, hence, it defines the reduced structure of $\Sigma'_{\text{diag}}$.  A germ $g \in \mathcal{O}_k$ is in $I_{k-1}$ if and only if every monomial of $g$ with non-zero coefficient contains at least $k-1$ variables (with non-zero exponent). Hence, the complement of $I_{k-1}$ in  $\mathcal{O}_k$ is closed under exponentiation, that is, if $g \in \mathcal{O}_k \setminus I_{k-1}$, then $g^j \in \mathcal{O}_k \setminus I_{k-1}$ also holds for all $j \in \N$. This implies that $\sqrt{I_{k-1}}=I_{k-1}$. 

             In contrast, one can consider the ideal $I \langle g_1, g_2 \rangle \subset \mathcal{O}_k$ generated by the two elements $g_1=\prod_{i=1}^k a_i$ and $g_2=\sum_{i=1}^k M_{ii}$ (see \cite[Example C.8]{MondBook} for $k=3$). The vanishing locus of $I \langle g_1, g_2 \rangle$ is also $(\Sigma'_{\text{diag}}, 0)$, as a set germ. Indeed, $g_1=0$ if and only if at least one of the variables is 0, say, $a_j=0$, then $M_{ii}$ is automatically 0, if $i \neq j$, hence $g_2=0$ means that $M_{jj}=0$, that is, another $a_{i}$ should also be 0 ($i \neq j$). Therefore, $V(I \langle g_1, g_2 \rangle)=V(I_{k-1})=(\Sigma'_{\text{diag}}, 0)$, and, in contrast of $I_{k-1}$, $I \langle g_1, g_2 \rangle$ is generated by $\text{codim}(\Sigma'_{\text{diag}}, 0)=2$ elements. But $I \langle g_1, g_2 \rangle$ is not a radical ideal, $\sqrt{I \langle g_1, g_2 \rangle}=I_{k-1}$. This means that $(\Sigma'_{\text{diag}}, 0)$ is \emph{set-theoretically a complete intersection}, but not a complete intersection. %\notegergo{Correct? Does it imply Cohen--Macaulay?} \notealex{This is correct, but it isn't obvious to me that it implies Cohen-Macaulayness.}
         \end{remark}

%\notegergo{Continue. Many problems, I wrote them to Gyuri on chat.}

\subsection{Proof for the symmetric case (part 1)}\label{ss:symmpart1} We prove point 2. of the Reduced Multiplicity Theorem~\ref{th:multiprime}, that is, we compute the multiplicity of $(\Sigma'_{\text{sym}}, 0)$. This implies point 2. of the Multiplicity Theorem~\ref{th:multipl} via Corollary~\ref{co:effmodmult} (about the multiplicity at other points) and Proposition~\ref{pr:primetilde} (about the equality of the multiplicities of $\Sigma_{\text{sym}}$, $\widetilde{\Sigma}_{\text{sym}}$ and $\Sigma'_{\text{sym}}$).

\begin{proof}[Proof of point 2. of the Reduced Multiplicity Theorem~\ref{th:multiprime}, part 1] The vanishing ideal of $\Sigma'_{\text{sym}} \subset \text{Symm}_{\C}(k) $ in the polynomial ring $\C[a_{ij}]$ ($i \leq j$, since $a_{ij}=a_{ji}$ because of the symmetry) is generated by the $(k-1) \times (k-1)$ minors $M_{ij}$. Since $M_{ij}=M_{ji}$, the number of different minors is equal to the number of pairs $(i, j)$ with $1 \leq i \leq j \leq k$, namely, $(k+1)k/2$.

Consider a linear map $f=(f_{ij})_{ij}:\C^3 \to \text{Symm}_{\C}(k)$ isolated with respect to $(\Sigma'_{\text{sym}}, 0)$. %\notegergo{do we have a simple characterization? does it mean that $f_{ij}$ are linearly independent? NO!}
%Assume that $f$ is isolated with respect to $(\Sigma'_{\text{sym}}, 0)$, this happens if and only if the entries $f_{ij}(x,y,z)$ and $f_{i'j'}(x,y,z)$ are linearly independent for different  index pairs $(i,j) \neq (i',j')$. 
The pull-back ideal $I_{k-1}(f):=f^{*}(I \langle M_{ij} \rangle )\cdot \mathcal{O}_3$ is generated by the homogeneous polynomials $m_{ij}:=M_{ij} \circ f$ of degree $k-1$. The vector space $\C[x,y,z]_{k-1}$ of such polynomials has dimension $(k+1)k/2$, that is, the number of degree $k-1$ monomials of three variables. If the minors $m_{ij} $ happen to be linearly independent, then they form a basis in $\C[x,y,z]_{k-1}$. Hence, every monomial of degree at least $k-1$ is contained in the ideal $I_{k-1}(f)$, that is, $I_{k-1}(f)$ is equal to the space of holomorphic functions germs of order at least $k-1$. Therefore, as vector spaces,
     \begin{equation}
\frac{\mathcal{O}_3}{I_{k-1}(f)} \cong \bigoplus_{d=0}^{k-2} \C[x,y,z]_d,
     \end{equation}
     hence its dimension is 
     \begin{equation}\label{eq:binoms}
         \dim \frac{\mathcal{O}_3}{I_{k-1}(f)}=\sum_{d=0}^{k-2} \binom{d+2}{2}=\binom{k+1}{3} = \frac{k(k^2-1)}{6}.
     \end{equation}
     This proves point 2. of the Reduced Multiplicity Theorem~\ref{th:multiprime}, once the linear independence of the minors $m_{ij}$ is verified. This will be done in Section~\ref{ss:proofsymm2} using the technics introduced for the generic case, namely, free resolution and syzygies. 
     \end{proof}

     \begin{rem}
         For Equation~\eqref{eq:binoms}  we used two combinatorial observations, which will be important in the next subsection as well.
         \begin{enumerate}
             \item The number of degree $d$ monomials of $m$ variable is $\binom{d+m-1}{m-1 }$.
             \item \begin{equation}
                 \sum_{d=0}^l \binom{d+c}{c}= \binom{l+c+1}{c+1} .
             \end{equation}
         \end{enumerate}
     \end{rem}

\subsection{Proof for the general case}\label{ss:proofgencase}
For the general case (point 1. of the Reduced Multiplicity Theorem~\ref{th:multiprime}) we use the Gulliksen--Neg{\aa}rd free resolution \cite{GulNeg}, 
see also \cite[Section 3.1]{MondGor}, \cite[Section 2]{BrunsVetterBook}.
 We start with a motivation by comparing the general case with the symmetric and diagonal cases.

Let $f=(f_{ij})_{ij}: \C^4 \to \C^{k \times k} $ be a linear map isolated with respect to $(\Sigma, 0)$. Note that $f$ can be considered as an element of $\mathcal{O}_4^{k \times k}$, that is, a matrix with entries $f_{ij} \in \mathcal{O}_4$. This motivates the simplified notation $M_{ij}(f)=M_{ij} \circ f$ for the $(k-1) \times (k-1)$ minors $M_{ij} \circ f$ of $f(x,y,z,w)$. We also use the simplified notation $(\mathcal{O}_4)_d=\C[x,y,z,w]_d$ for the vector space of degree $d$ homogeneous polynomials. 

 The pull-back ideal $I_{k-1}(f):=f^{*} (I \langle M_{ij} \rangle ) \cdot \mathcal{O}_4 \subset \mathcal{O}_4$ is generated by these minors $M_{ij}(f)$, which are homogeneous polynomials of degree $k-1$, and their number is $k^2$. The dimension of the space $(\mathcal{O}_4)_{k-1}$ of degree $k-1$ homogeneous polynomials is $\binom{ k+2}{3}=(k+2)(k+1)k/6$. Observe that $(k+2)(k+1)k/6 >k^2$, if $k>2$, indeed, the difference $(k+2)(k+1)k/6 -k^2$ has its three roots at $k=0, 1, 2$ and it is monotonic increasing for $k>2$. Therefore, in contrast of the symmetric and diagonal cases, the minors $M_{ij} (f)$ do not span the vector space $(\mathcal{O}_4)_{k-1}$.

Before introducing the Gulliksen--Neg{\aa}rd free resolution, as a motivation, we present a naive estimation for the dimension of  $\mathcal{O}_4/ I_{k-1}(f)$. 
All the homogeneous monomials of degree less than $k-1$ are not contained in the ideal $I_{k-1}(f)$, hence, they represent linearly independent elements of the quotient $\mathcal{O}_4/I_{k-1}(f)$. Therefore,  
    \begin{equation}
         \dim \frac{\mathcal{O}_4}{I_{k-1}(f)} \geq \sum_{d=0}^{k-2} \dim ((\mathcal{O}_4)_d)=
 \sum_{d=0}^{k-2} \binom{d+3}{3}=\binom{k+2}{4}.
     \end{equation}
     
     In degree $k-1$, the dimension of $I_{k-1}(f) \cap (\mathcal{O}_4)_{k-1}$ is at most $k^2$. In fact, if the minors $M_{ij} ( f)$ happen to be linearly independent (as we will see, they are), then they form a basis in $I_{k-1}(f) \cap (\mathcal{O}_4)_{k-1}$ (otherwise the dimension is less than $k^2$). Therefore,
     \begin{equation}
         \dim \frac{\mathcal{O}_4}{I_{k-1}(f)} \geq \sum_{d=0}^{k-1} \dim ((\mathcal{O}_4)_d)-\dim (I_{k-1}(f) \cap (\mathcal{O}_4)_{k-1}) \geq
         \sum_{d=0}^{k-1} \binom{d+3}{3}  -k^2=\binom{k+3}{4}-k^2.
     \end{equation}

     In degree $k$, the vector space $I_{k-1}(f) \cap (\mathcal{O}_4)_{k}$ is spanned by the products of the minors $M_{ij}(f)$ and the variables $x,y, z, w$. However, these products are linearly dependent by the cofactor formulas for each pair $1 \leq i \neq j \leq k$
     \begin{equation}\label{eq:cofact}
         \sum_{l=1}^k (-1)^{i+l} f_{il}  \cdot M_{jl}(f)=0 \text{ and } \sum_{l=1}^k (-1)^{i+l} f_{li}  \cdot M_{lj}(f)=0,
     \end{equation}
     and, from the expansion of $\det(f)$,
     \begin{equation}\label{eq:cofactdet}
         \sum_{l=1}^k ((-1)^{i+l} f_{il}  \cdot M_{il}(f)- (-1)^{j+l}f_{jl}  \cdot M_{jl}(f))=0 \text{ and } \sum_{l=1}^k ((-1)^{i+l} f_{li}  \cdot M_{li}(f)- (-1)^{j+l} f_{lj}  \cdot M_{lj}(f)))=0.
     \end{equation}
     The relations \eqref{eq:cofactdet} are not independent, indeed, each one is a linear combination of the $(i,j)=(i,i+1)$ cases ($i=1, \dots, k-1$). With this restriction, equations~\eqref{eq:cofact} and \eqref{eq:cofactdet} express $2(k^2-k)+2(k-1)=2(k^2-1)$ relations between the products $f_{i'j'} M_{ij}$. Assuming that they are linearly independent (as we will see, they are, and there are no other relations), we have
     \begin{equation}
         \dim \frac{\mathcal{O}_4}{I_{k-1}(f)} \geq 
         \sum_{d=0}^{k} \binom{d+3}{3}  -k^2- (4k^2-2(k^2-1))=\binom{k+4}{4}-3k^2-2.
     \end{equation}

%     \notegergo{Correct? It is equal to the multiplicity for  $k=2,3,4$, by wolfram alpha.   Are the products $f_{i'j'} M_{ij}$ and the relations between them linearly independent? I'm  confused because already $f_{ij}$ are not independent. But I think the inequality still holds if the relations are not independent. Right?}

One can continue this  analysis by considering the generators, relations between them, relations between relations etc. in degree $k+1, k+2$ and so on. These iterated relations are called \emph{syzygies}, and the \emph{free resolution} is  a useful tool to manage them. In our particular case, the Gulliksen--Neg{\aa}rd free resolution of $\mathcal{O}_4/I_{k-1}(f)$ is the sequence $\text{GN}(f)$
\begin{equation}\label{eq:gnf}
\text{GN}(f): \ 0 \to \mathcal{O}_4 \xrightarrow{d_4} \mathcal{O}_4^{k \times k} \xrightarrow{d_3} sl_{k} ( \mathcal{O}_4) \oplus sl_{k} (\mathcal{O}_4) \xrightarrow{d_2} \mathcal{O}_4^{k \times k} \xrightarrow{d_1} \mathcal{O}_4 \xrightarrow{d_0} \frac{\mathcal{O}_4}{I_{k-1}(f)} \to 0.
\end{equation}
Here $sl_{k} ( \mathcal{O}_4)=\{g \in \mathcal{O}_4^{k \times k} \ | \ \text{tr}(g)=0 \}$ is the space of special linear matrices, and the homomorphisms are defined as
\begin{itemize}
    \item $d_1(g)=\text{tr}(f^*g)$, where $f^*$ is the adjoint (signed cofactor matrix) of $f$, that is, $(f^*)_{ij}=(-1)^{i+j} M_{ji}(f)$,
    \item $d_2(g,h)=fg-hf$,
    \item $d_3(g)=(gf-\frac{\text{tr}(gf)}{k} \mathds{1}_k, fg-\frac{\text{tr}(fg)}{k} \mathds{1}_k)$,
    \item $d_4(g)=gf^*$.
\end{itemize}

To interpret the homomorphisms $d_1$, observe that the image of an element $g=(g_{ij})_{ij} \in \mathcal{O}_4^{k \times k}$ is 
\begin{equation}\label{eq:interprd1}
d_1(g)=\text{tr}(f^*g)=\sum_{i, j=1}^k (-1)^{i+j} g_{ij} M_{ij}(f) .
\end{equation}
Therefore, the image of $d_1$ is the ideal $I_{k-1}(f) \subset \mathcal{O}_4$, implying that $d_0 \circ d_1=0$, moreover, $\text{im}(d_{1})=\text{ker} (d_0)$.

To interpret the homomorphism $d_2$, consider the following generators of $sl_{k} ( \mathcal{O}_4) \oplus sl_{k}(\mathcal{O}_4)$ over $\mathcal{O}_4$:  $(e_{ij}, 0), (0,e_{ij})$ for $i \neq j$ and $ (e_{ii}-e_{jj}, 0), (0, e_{ii}-e_{jj})$, where $e_{ij} \in \mathcal{O}_4^{k \times k} $ is the matrix whose $ij$ entry is 1 and all the others are 0. 
Then $d_2(e_{ij}, 0) \in \mathcal{O}_4^{k \times k}$ is the matrix whose $j$-th column is equal to the $i$-th column of $f$, and $d_2(0,e_{ij}) \in \mathcal{O}_4^{k \times k}$ is the matrix whose $i$-th row is $(-1)$ times the $j$-th row of $f$. 

For $i \neq j$, the compositions $d_1(d_2(e_{ij},0))$ and $d_1(d_2(0,e_{ij})$ are exactly the expressions on the left hand side of the cofactor formulas~\eqref{eq:cofact}, while $d_1(d_2(e_{ii}-e_{jj}, 0))$ and $d_1(d_2(0, e_{ii}-e_{jj}))$ are the left sides of the formulas~\eqref{eq:cofactdet}, hence, they are 0. This implies that $d_1 \circ d_2=0$, since the images of the generators are 0.

$\text{GN}(f)$ is a \emph{free resolution} of $\mathcal{O}_4/I_{k-1}(f)$, which means the following:
\begin{itemize}
    \item All the terms except $\mathcal{O}_4/I_{k-1}(f)$ are free $\mathcal{O}_4$-modules, that is, isomorphic to $\mathcal{O}_4^K$ for some integer $K$. The only non-trivial one is $sl_{k} ( \mathcal{O}_4) \oplus sl_{k}(\mathcal{O}_4)$, but it is in fact isomorphic to $\mathcal{O}_4^{2(k^2-1)}$.
    \item The arrows are $\mathcal{O}_4$-module homomorphisms, that is, they are linear maps over $\mathcal{O}_4$.
    \item $\text{GN}(f)$ is a \emph{chain complex}, that is, the compositions of the neighboring arrows are 0. Above we already saw that $d_0 \circ d_1=0$ and $d_1 \circ d_2=0$, and it can be checked similarly for each composition.
    \item Moreover, $\text{GN}(f)$ is an \emph{exact sequence (acyclic chain complex)}, that is, $\text{im}(d_{i+1})=\text{ker} (d_i)$ holds for each pair of neighboring arrows. We already saw it above for $i=0$, and the proof for the other pairs of arrows can be found in \cite{BrunsVetterBook}. 
\end{itemize}

In other words, for example, the image of $sl_{k} ( \mathcal{O}_4) \oplus sl_{k} (\mathcal{O}_4)$ via $d_2$ exactly consists of  the relations between the minors $M_{ij}(f)$. These are generated by the $2(k^2-1)$ linearly independent cofactor formulas in Equations~\eqref{eq:cofact} and \eqref{eq:cofactdet}, according to the fact that $sl_{k} ( \mathcal{O}_4) \oplus sl_{k}(\mathcal{O}_4) \cong \mathcal{O}_4^{2(k^2-1)}$ (as $\mathcal{O}_4$-modules). The previous term $\mathcal{O}_4^{k \times k}$, via $d_3$, represents the algebraic relations between these cofactor formulas, and so on.

\begin{proof}[Proof of point 1. of the Reduced Multiplicity Theorem~\ref{th:multiprime}] The arrows of $\text{GN}(f)$ are homogeneous homomorphisms, i.e. they map homogeneous parts to homogeneous parts with shifted degree. Hence, we have exact sequences of complex vector spaces:
\begin{equation}
 \ 0 \to (\mathcal{O}_4)_{d-2k} \xrightarrow{d_4} (\mathcal{O}_4^{k^2})_{d-k-1} \xrightarrow{d_3} (\mathcal{O}_4^{2k^2-2})_{d-k} \xrightarrow{d_2} (\mathcal{O}_4^{k^2})_{d-k+1} \xrightarrow{d_1} (\mathcal{O}_4)_d \xrightarrow{d_0} \left( \frac{\mathcal{O}_4}{I_{k-1}(f)} \right)_d \to 0.
\end{equation}

Recall that for an exact sequence of vector spaces \begin{equation}
    0 \to V_K \xrightarrow{d_K} V_{K-1} \xrightarrow{d_{K-1}} \dots  \xrightarrow{d_2} V_1 \xrightarrow{d_1} 0,
\end{equation}
the alternate sum of the dimensions is 0, that is, $\sum_{i=1}^K (-1)^{i} \dim (V_i)=0$. This is because each $V_i$ decomposes as $V_i \cong \text{ker}(d_i) \oplus (V_i/\text{ker}(d_i))$, and, by the homomorphism theorem, $V_i/\text{ker}(d_i) \cong \text{im}(d_i)$. By the exactness,  $\text{im}(d_i)=\text{ker}(d_{i-1})$, $\text{ker}(d_{1})=V_1$ and $\text{ker}(d_{K})=0$. Hence, \begin{equation}
    \sum_{i=1}^K (-1)^{i} \dim (V_i)= \sum_{i=1}^K (-1)^{i} (\dim (\text{ker}(d_i))+ \dim (\text{ker}(d_{i-1}))=0.
\end{equation} 

Therefore, 
\begin{equation}\label{eq:sequence}
   \dim \left( \frac{\mathcal{O}_4}{I_{k-1}(f)} \right)_d = \dim (\mathcal{O}_4)_d - \dim (\mathcal{O}_4^{k^2})_{d-k+1}+ \dim (\mathcal{O}_4^{2k^2-2})_{d-k}- \dim (\mathcal{O}_4^{k^2})_{d-k-1} + \dim (\mathcal{O}_4)_{d-2k}.
\end{equation}
Moreover, if $\dim ( \mathcal{O}_4/I_{k-1}(f) )_{d_0}=0$ holds for a degree $d_0$, then $\dim ( \mathcal{O}_4/I_{k-1}(f) )_{d}=0$ holds for every $d \geq d_0$, since this means that $(\mathcal{O}_4)_{d_0} \subset I_{k-1}(f)$. In this case, as vector spaces, 
\begin{equation}
     \frac{\mathcal{O}_4}{I_{k-1}(f)} \cong \bigoplus_{d=0}^{d_0} \left( \frac{\mathcal{O}_4}{I_{k-1}(f)} \right)_d,
\end{equation}
hence, 
\begin{equation}
    \dim \left( \frac{\mathcal{O}_4}{I_{k-1}(f)} \right) = \sum_{d=0}^{d_0} \dim \left( \frac{\mathcal{O}_4}{I_{k-1}(f)} \right)_d.
\end{equation}

By taking into account that 
\begin{equation}
\dim (\mathcal{O}_4^N)_d=N \cdot \binom{d+3}{3}=\frac{N(d+3)(d+2)(d+1)}{6},
\end{equation}
the following table contains the values of the terms in Equation~\eqref{eq:sequence}:

\begin{small}
\begin{equation}\label{eq:table}
\renewcommand{\arraystretch}{2}
\begin{array}{*{6}{|c}{|c|}}
%{ccccccc}
\hline
\text{dimension} & (\mathcal{O}_4)_{d-2k} & (\mathcal{O}_4^{k^2})_{d-k-1} & (\mathcal{O}_4^{2k^2-2})_{d-k} & (\mathcal{O}_4^{k^2})_{d-k+1} & (\mathcal{O}_4)_d & \left( \frac{\mathcal{O}_4}{I_{k-1}(f)} \right)_d \\ \hline
0 \leq d \leq k-2 & 0 & 0 & 0 & 0 & \binom{ d+3}{3} & \binom{d+3}{3} \\
d=k-1 & 0 & 0 & 0 & k^2 \binom{ d-k+4 }{3} & \binom{ d+3}{3} & \binom{2k-1-d}{3} \\
d=k & 0 & 0 & (2k^2-2) \binom{d-k+3}{3} & k^2 \binom{ d-k+4}{3} & \binom{ d+3}{3} & \binom{2k-1-d}{3} \\
k+1 \leq d \leq 2k-1 & 0 & k^2 \binom{ d-k+2}{3} & (2k^2-2) \binom{d-k+3}{3} & k^2 \binom{ d-k+4}{3} & \binom{ d+3}{3} & \binom{2k-1-d}{3} \\
2k \leq d & \binom{d-2k+3}{3} & k^2 \binom{ d-k+2}{3} & (2k^2-2) \binom{d-k+3}{3} & k^2 \binom{ d-k+4}{3} & \binom{ d+3}{3} & 0 \\ \hline
%\sum_{d=0}^{2k}: & 1 &   
%k^2 {k+3 \choose 4}& 
%(2k^2-2) {k+4 \choose 4}& 
%k^2 {k+5 \choose 4} & 
%{2k+4 \choose 4}& \frac{k^2(k^2-1)}{12} \
\sum_{d=0}^{2k-4} & 0 &   
k^2 \binom{k-1}{4}& 
(2k^2-2) \binom{k}{4}& 
k^2 \binom{k+1}{4} & 
\binom{2k}{4}& \frac{k^2(k^2-1)}{12} \\
\hline
\end{array}
\end{equation}
\end{small}

This is the table for $k \geq 4$, for $k=2, 3$ a similar table can be obtained.
The table shows that $d_0=2k-3$ can be chosen, and it proves the statement that $\dim ( \mathcal{O}_4/I_{k-1}(f))=k^2(k^2-1)/12$. 
\end{proof}

\begin{rem}\label{re:cohmacgoren}
    \begin{itemize}
        \item The Gulliksen--Neg{\aa}rd free resolution $\text{GN}(f)$ exists for any holomnorphic germ $f: (\C^m, 0) \to (\C^{k \times k}, 0)$, i.e., $f \in \mathcal{O}_m^{k \times k}$, just $\mathcal{O}_4$ has to be replaced by $\mathcal{O}_m$ in Equation~\eqref{eq:gnf}. $\text{GN}(f)$ is always a complex, and it is exact if $f^{-1}( \Sigma', 0)$ has its maximal codimension 4. In our case, that is, $f \in  \mathcal{O}_4^{k \times k} $ linear, this means that $f$ is isolated with respect to $(\Sigma', 0)$, or equivalently, $f$ is generic with respect to $(\Sigma', 0)$.
        \item In our case, the exactness of  $\text{GN}(f)$ implies that the $(k-1) \times (k-1)$ minors $M_{ij}(f)$ are linearly independent. Indeed, by Equation~\eqref{eq:interprd1}, a nontrivial zero linear combination would mean a degree 0 element (a non-zero constant matrix) $ g \in (\mathcal{O}_4^{k \times k})_0$ such that $d_1(g)=0$, that is, $ g \in \text{ker}(d_1)$. But $\text{ker}(d_1)=\text{im}(d_2)$ by the exactness, and the lowest degree here is 1.   Similar arguments show that the cofactor formulas~\eqref{eq:cofact} and \eqref{eq:cofactdet} are also linearly independent, and all the linear relations between the products of the variables and the minors are linear combinations of the cofactor formulas.
        \item Observe that the lenght of $\text{GN}(f)$, that is, the number of non-zero terms on the left of $\mathcal{O}_4$, is 4, the same as the codimension of $f^{-1}(\Sigma')$. As pointed out in \cite{MondGor}, this fact implies that $(\Sigma',0)$ is Cohen--Macaulay. This implication can be proved via the Auslander--Buchsbaum formula \cite[Chapter 19]{EisenbudCommutativeAlgebra}. 
        \item It is also mentioned in \cite{MondGor} that the sequence $\text{GN}(f)$ is self-dual, which means that $(\Sigma',0)$ is \emph{Gorenstein}, which is a property stronger than being Cohen--Macaulay. This implies that the dimensions in the last column of Table~\ref{eq:table} (i.e., the \emph{Hilbert sequence} of $\mathcal{O}_4/I_{k-1}(f)$) satisfies the symmetry
        \begin{equation}
          \dim  \left( \frac{\mathcal{O}_4}{I_{k-1}(f)} \right)_d= \dim  \left( \frac{\mathcal{O}_4}{I_{k-1}(f)} \right)_{2(k-2)-d},
        \end{equation}
        as it can be verified by direct computation.
    \end{itemize}
\end{rem}

\subsection{Proof for the symmetric case (part 2)}\label{ss:proofsymm2} For the symmetric case in Section~\ref{ss:symmpart1} there is still left to prove that the minors $M_{ij}(f) \in (\mathcal{O}_3)_{k-1} $ are linearly independent. This follows from the \emph{J\'{o}zefiak free resolution} of $\mathcal{O}_3/I_{k-1}(f)$ for symmetric families $f$ with $\text{codim}(f^{-1}(\Sigma'_{\text{sym}}, 0))=3$, see Refs.~\cite{MondGor,jozefiak1978ideals}.
% T. J´ozefiak, Ideals generated by minors of a symmetric matrix , Coment. Math. Helvetici 53 (1978), 595–607.

Although the Gulliksen--Neg{\aa}rd free resolution $\text{GN}(f)$ can be defined for such symmetric matrix families via the inclusion $\text{Symm}_{\C}(k) \subset \C^{k \times k}$, in this case $\text{GN}(f)$ is not necessarily exact, because the codimension condition $\text{codim}(f^{-1}(\Sigma', 0))=4$ does not hold. Instead, we use the J\'{o}zefiak free resolution, which is
\begin{equation}
\text{JO}(f): \ 0 \to \text{sk}_k(\mathcal{O}_3) \xrightarrow{d_3} \text{sl}_k(\mathcal{O}_3) \xrightarrow{d_2} \text{symm}_k(\mathcal{O}_3)  \xrightarrow{d_1} \mathcal{O}_3 \xrightarrow{d_0} \frac{\mathcal{O}_3}{I_{k-1}(f)} \to 0,
\end{equation}
where $\text{sk}_k(\mathcal{O}_3), \text{sl}_k(\mathcal{O}_3) $ and $\text{symm}_k(\mathcal{O}_3)$ denote the spaces of skew-symmetric ($g+g^T=0$), special linear ($\text{tr}(g)=0$) and symmetric ($g^T=g$) matrices of size $k \times k$ over $\mathcal{O}_3$, respectively, and the homomorphisms are defined as
\begin{itemize}
    \item $d_1(g)=f^*g$,
    \item $d_2(g)=fg+g^Tf$,
    \item $d_3(g)=gf$.
\end{itemize}
In our case, i.e. $f \in (\text{symm}_k(\mathcal{O}_3))_1$ is isolated with respect to $(\Sigma'_{\text{sym}}, 0)$, the codimension condition $\text{codim}(f^{-1}(\Sigma'_{\text{sym}}, 0))=3$ is satisfied, hence $\text{JO}(f)$ is exact \cite{jozefiak1978ideals}. 

Then the linear independence of the minors can be proved in the same way as it was done for general matrix families, see Remark~\ref{re:cohmacgoren}. The (non-trivial) linear combinations of the minors $M_{ij}(f) \in (\mathcal{O}_3)_{k-1} $ are exactly the images $d_1(g)$ of the non-zero constant matrix families $g \in (\text{symm}_k(\mathcal{O}_3))_{0}$. Since  $\text{im}(d_2)=\text{ker}(d_1)$ does not contain any degree 0 element, $d_1(g) \neq 0$ in $\mathcal{O}_3$. This proves that the minors are linearly independent, and finishes the proof of point 2. of the Reduced Multiplicity Theorem~\ref{th:multiprime}.

Note that the length of $\text{JO}(f)$ (which is 3) is equal to the codimension of $f^{-1}(\Sigma'_{\text{sym}}, 0)$, proving that $(\Sigma'_{\text{sym}}, 0)$ is Cohen--Macaulay by the Auslander--Buchsbaum formula. However, $\text{JO}(f)$ is not self-dual, indeed, the dimensions $\dim(\mathcal{O}_3/I_{k-1}(f))$ are not symmetric, in other words, $(\Sigma'_{\text{sym}}, 0)$ is not Gorenstein.

\subsection[The geometric degeneracy variety Sigma is not Cohen--Macaulay]{The geometric degeneracy variety $\Sigma$ is not Cohen--Macaulay}\label{ss:proof-notcohmac}

Recall that all the varieties $\Sigma'_{\bullet}$ are Cohen--Macaulay, where the index $\bullet$  denotes either sym, diag, or nothing (general case), and this implies that $\widetilde{\Sigma}_{\bullet}$ are also Cohen--Macaulay (indeed, $\widetilde{\Sigma}_{\bullet} \cong \Sigma'_{\bullet} \times \C$). The Cohen--Macaulay property of $\Sigma'$ and $\Sigma'_{\text{sym}}$ is a classical result \cite{BrunsVetterBook}, it follows for example from the free resolutions Gulliksen--Neg{\aa}rd and  J\'{o}zefiak respectively, more precisely, from the fact that the length of these free resolutions agrees with the codimension of $\Sigma'$ and $\Sigma'_{\text{sym}}$, respectively. Moreover, $\Sigma'$ is also Gorenstein, which does not hold for $\Sigma'_{\text{sym}}$. See Section~\ref{ss:proofsymm2} and Remark~\ref{re:cohmacgoren}. 

The Cohen--Macaulay property of $\Sigma'_{\text{diag}}$ follows directly from the computation of the multiplicity, see Section~\ref{ss:proofdiag}. In the diagonal case, $\Sigma_{\text{diag}}$ is also Cohen--Macaulay, moreover, it is complete intersection, in fact, a hypersurface.

In contrast, here we show that $\Sigma$ and  $\Sigma_{\text{sym}}$ are not Cohen--Macaulay. We describe our argument for $\Sigma$, and $\Sigma_{\text{sym}}$ will be similar.

First we determine the vanishing ideal (reduced structure) $I(\Sigma, 0) \subset \mathcal{O}_{n \times n}$ of $(\Sigma, 0) \subset (\C^{n \times n}, 0)$. Recall from Section~\ref{ss:degcomp} that in \cite{DomokosHermitian}  an ideal $I_{(n)} \subset \C[A]$ is given whose vanishing locus $V(I_{(n)}) \subset \C^{n \times n}$ is equal to $\Sigma$ as a set, however it is proved only for $n=3$ that this ideal $I_{(n)}$ is the vanishing ideal of $\Sigma$ (that is, $I_{(n)}$ is a radical ideal). We use a different construction than \cite{DomokosHermitian}, but, of course, for $n=3$ our ideal $I(\Sigma^{(3)})$ coincide with $I_{(3)}$ given in  \cite{DomokosHermitian}, since both are the vanishing ideals of $\Sigma^{(3)}$ (and the local analytic version $I(\Sigma^{(3)}, 0)$ is the ideal in $\mathcal{O}_{n \times n}$ generated by $I(\Sigma^{(3)}) \subset \C[A] \subset \mathcal{O}_{n \times n}$). 

 By Proposition~\ref{pr:projectiongen1to1} and Corollary~\ref{co:genonetoone}, $\Sigma$ is finite and generically 1-to-1 image of $\widetilde{\Sigma} \subset \C^{n \times n} \times \C$ via the projection $p:  \C^{n \times n} \times \C \to  \C^{n \times n} $, $p(A, \lambda)=A$. Consider the restriction $p|_{\widetilde{\Sigma}}: \widetilde{\Sigma} \to \C^{n \times n}$. The induced homomorphism is
\begin{eqnarray}
    (p|_{\widetilde{\Sigma}})^*: \mathcal{O}_{n \times n} &\to& \mathcal{O}_{(\widetilde{\Sigma}, (0,0))} \\
    (p|_{\widetilde{\Sigma}})^*(g)(A, \lambda)&=&g(A),
\end{eqnarray}
for $g \in \mathcal{O}_{n \times n}$ and $(A, \lambda) \in (\widetilde{\Sigma}, (0,0))$, cf. Equation~\eqref{eq:inducedhomo} and the proof of Proposition~\ref{pr:projectiongen1to1}.

\begin{thm} The vanishing ideal $I(\Sigma,0) \subset \mathcal{O}_{n \times n}$ of $(\Sigma, 0) \subset (\C^{n \times n}, 0)$ is
\begin{equation}
    I(\Sigma,0)=\text{ker}(p|_{(\widetilde{\Sigma},(0,0))})^*= I(\widetilde{\Sigma}, (0,0)) \cap \mathcal{O}_{n \times n},
\end{equation}
and similarly, 
\begin{equation}
    I(\Sigma)=\text{ker}(p|_{\widetilde{\Sigma}})^*= I(\widetilde{\Sigma}) \cap \C[A].
\end{equation}
\end{thm}

\begin{proof}
    The first equality follows from the fact that $p|_{\widetilde{\Sigma}}$ 
 is a finite, generically 1-to-1 onto its image $(\Sigma,0)$. For the second equality observe that 
 \begin{equation}
     \mathcal{O}_{(\widetilde{\Sigma}, (0,0))}=\frac{\mathcal{O}_{n^2+1}}{I(\widetilde{\Sigma}, (0,0))},
 \end{equation}
 hence, $(p|_{\widetilde{\Sigma}})^*(g)=0$ holds for $g \in \mathcal{O}_{n \times n}$ if and only if $g \in I(\widetilde{\Sigma}, (0,0))$, proving the theorem. \end{proof}

Recall from Section~\ref{ss:matrvar} that $I(\widetilde{\Sigma}, (0,0))$ is the ideal in $\mathcal{O}_{n^2+1}$ generated by the $(n-1) \times (n-1)$ minors $M_{ij}(A-\lambda \mathds{1})$. Hence, the elements of $I(\Sigma, 0)$ are the combinations of $M_{ij}(A-\lambda \mathds{1})$ with coefficients in $\mathcal{O}_{n^2+1}$ which do not depend on $\lambda$ (i.e., they are obtained by the elimination of $\lambda$).

For example, for $n=3$ 
\begin{equation}
    a_{31} \cdot M_{12} +a_{12} \cdot M_{13}=a_{31} a_{21} a_{33}-a_{31}^2 a_{23} + a_{21}^2 a_{32} - a_{21} a_{31} a_{22} \in I(\Sigma^{(3)}, 0).
\end{equation}
This elimination method suggests the following proposition.

\begin{lemma}\label{le:ord3}
Every element of $I(\Sigma^{(3)}, 0)$ has order at least 3. 
\end{lemma}

\begin{proof}
    Let $I_{(n)} \subset \C[A]$ be the ideal given in \cite{DomokosHermitian} with vanishing locus $V(I_{(n)})=\Sigma^{(n)} \subset \C^n$, see Section~\ref{ss:degcomp}. It is proved in \cite{DomokosInvariant}
    %M. Domokos, Invariant theoretic characterization of subdiscriminants of matrices, Linear Multilinear Algebra (2013), http://dx.doi.org/10.1080/03081087.2012.762714, in press.
    that the minimal degree of a non-zero homogeneous component is $\binom{n}{2}$, which is equal to $3$ for $n=3$. See also the discussion on \cite[pg. 3965]{DomokosHermitian}.
    %Domokos: Hermitian matrices with a bounded number of eigenvalues
    Moreover, it is proved in \cite{DomokosHermitian} that $I_{(3)}$ is the reduced structure of $\Sigma^{(3)}$, that is, $I(\Sigma^{(3)})=I_{(3)}$, therefore, $I(\Sigma^{(3)}, 0) \subset \mathcal{O}_{3 \times 3}$ is the ideal generated by $I_{(3)}$. 
\end{proof}

\begin{cor}\label{co:notcoh3}
    $(\Sigma^{(3)}, 0)$ is not Cohen--Macaulay.
\end{cor}

\begin{proof}
    Consider a linear map $f: (\C^3,0) \to (\C^{3 \times 3}, 0)$ isolated with respect to $\Sigma^{(3)}$. For a generic perturbation $f_t$, the number of complex Weyl points is $\sharp f_t^{-1}(\Sigma^{(3)}, 0)=\text{mult}(\Sigma^{(3)}, 0)=6$, see Corollary~\ref{co:kanegyzet_symm}. On the other hand, 
    \begin{equation}\label{eq:three}
        \dim \frac{\mathcal{O}_3}{f^*(I(\Sigma^{(3)}, 0)) \cdot \mathcal{O}_3} \geq \sum_{d=0}^2 \dim (\C[x,y,z]_d)=10,
    \end{equation}
    since the minimal order in $I(\Sigma^{(3)}, 0)$ is at least 3 and $f$ is linear. Since the left side of equation~\eqref{eq:three} is not equal to $\text{mult}(\Sigma^{(3)}, 0)$, this proves that  $(\Sigma^{(3)}, 0)$ is not Cohen--Macaulay by Corollary~\ref{co:finitegeneric2}.% \notegergo{I'm a bit confused. Does the previous proposition imply that there no order $\leq 2$ elements in the pull-back ideal?}
\end{proof}

\begin{thm}
\label{co:notcohmac}
%\begin{enumerate} 
%\item 
The affine variety $\Sigma^{(k)}$ and the analytic set germs $(\Sigma^{(k)}, 0)$, $(\Sigma^{(n)}, A_0; \lambda_0)$, $(\Sigma^{(n)}, A_0)$ are not Cohen--Macaulay for $k \geq 3$, where $\lambda_0$ is a strictly $k$-fold degenerate eigenvalue of $A_0 \in \Sigma^{(n)}$.
%\item The affine variety $\Sigma^{(k)}$ is not Cohen--Macaulay for $k \geq 3$.
%\end{enumerate}
\end{thm}

\begin{proof}
First we show that the germ $(\Sigma^{(k)}, 0) \subset (\C^{k \times k}, 0)$ is Cohen--Macaulay if and only if the affine variety $\Sigma^{(k)} \subset \C^{k \times k}$ is Cohen--Macaulay. This is because 
\begin{enumerate}
    \item being Cohen--Macaulay is an `open property', that is, a Cohen--Macaulay analytic set germ has a representative whose germ at every point is Cohen--Macaulay \cite[Prop. 18.8.]{EisenbudCommutativeAlgebra}.
    \item moreover, since $\Sigma^{(k)}$ is homogeneous, any small representative of $(\Sigma^{(k)}, 0)$  represents all the germs of $\Sigma^{(k)}$, that is, the germs $(\Sigma^{(k)}, A)$ and $(\Sigma^{(k)}, tA)$ are isomorphic for $0 \neq A \in \Sigma^{(k)}$, $0 \neq t \in \C$. See also Remark~\ref{re:alex}.
\end{enumerate}
If $k \geq 3$, then any small representative of  $(\Sigma^{(k)}, 0)$ has a point $A$ with  a strictly $3$-fold degenerate eigenvalue such that all the other eigenvalues are non-degenerate (in particular, $A$ is diagonalizable). Then $(\Sigma^{(k)}, A)$ is a trivial deformation of $(\Sigma^{(3)}, 0)$, which is not Cohen--Macaulay. This proves that $\Sigma^{(k)}$ and $(\Sigma^{(k)}, 0)$ are not Cohen--Macaulay. 

Since $(\Sigma^{(n)}, A_0; \lambda_0)$ is a trivial deformation of $(\Sigma^{(k)}, 0)$, it is not Cohen--Macaulay. Since $(\Sigma^{(n)}, A_0; \lambda_0)$ is a component of $(\Sigma^{(n)}, A_0)$, this is not Cohen--Macaulay, proving the corollary. 
\end{proof}

As a consequence, in Theorem~\ref{th:pullback} the right hand side cannot be replaced by the codimension of the pull back of the ideal of $(\Sigma^{(n)}, A_0; \lambda_0)$. More precisely,

\begin{thm}\label{th:notwork00} Let $f: (\C^3, 0) \to (\C^{n \times n}, A_0)$ be a holomorphic map germ such that $f(0)=A_0 \in \Sigma$ has a strictly $k$-fold eigenvalue $\lambda_0$. Assume that $f$ is isolated with respect to $(\Sigma, A_0; \lambda_0)$ and consider a perturbation $f_t$ of $f$ generic with respect to $(\Sigma, A_0; \lambda_0)$. Then, in general,
    \begin{equation}
        \sharp f_t^{-1}(\Sigma, A_0; \lambda_0) \neq \dim \frac{\mathcal{O}_3}{f^*(I(\Sigma, A_0; \lambda_0)) \cdot \mathcal{O}_3}.
    \end{equation}
\end{thm}

We have analogous statements for the symmetric case as well, shortly summarized it below.

\begin{thm}
\label{pr:notcohmacsym}
%\begin{enumerate} 
%\item 
The affine variety $\Sigma_{\text{sym}}^{(k)}$ and the analytic set germs $(\Sigma_{\text{sym}}^{(k)}, 0)$, $(\Sigma_{\text{sym}}^{(n)}, A_0; \lambda_0)$, $(\Sigma_{\text{sym}}^{(n)}, A_0)$ are not Cohen--Macaulay for $k \geq 3$, where $\lambda_0$ is a strictly $k$-fold degenerate eigenvalue of $A_0 \in \Sigma^{(n)}_{\text{sym}}$.
%\item The affine variety $\Sigma^{(k)}$ is not Cohen--Macaulay for $k \geq 3$.
%\end{enumerate}
\end{thm}

\begin{proof}[Sketch of the proof] The proof is similar to the general case, here we focus on the differences. We only have to show that $(\Sigma_{\text{sym}}^{(3)}, 0)$ is not Cohen--Macaulay, and this implies the theorem in the same way as the proof of Theorem~\ref{co:notcohmac} is based on Corollary~\ref{co:notcoh3}. Like in the proof of Corollary~\ref{co:notcoh3}, consider a linear map $f: (\C^2,0) \to (\text{Symm}_{\C}(3), 0)$ isolated with respect to $\Sigma_{\text{sym}}^{(3)}$. For a generic perturbation $f_t$, the number of complex Weyl points is $\sharp f_t^{-1}(\Sigma^{(3)}_{\text{sym}}, 0)=\text{mult}(\Sigma^{(3)}_{\text{sym}}, 0)=4$, see Corollary~\ref{co:kanegyzet_symm}. It can be proved in the same way as Lemma~\ref{le:ord3} that every element of $I(\Sigma^{(3)}_{\text{sym}}, 0)$ has order at least 3. Hence, 
    \begin{equation}\label{eq:three2}
        \dim \frac{\mathcal{O}_2}{f^*(I(\Sigma^{(3)}, 0)) \cdot \mathcal{O}_2} \geq \sum_{d=0}^2 \dim (\C[x,y]_d)=6,
    \end{equation}
    since $f$ is linear. Since the left side of equation~\eqref{eq:three2} is not equal to $\text{mult}(\Sigma^{(3)}_{\text{sym}}, 0)$, this proves that  $(\Sigma^{(3)}_{\text{sym}}, 0)$ is not Cohen--Macaulay by Corollary~\ref{co:finitegeneric2}, and this finishes the proof of this theorem.
\end{proof}

%\subsection{Linear maps with non-isolated intersection}

\section{Parameter dependent Hamiltonian systems}\label{s:physex}

We apply the results of the previous sections for  parameter dependent Hamiltonian systems $H: M^3 \to \text{Herm}(n)$ to give an upper bound for the number  of Weyl points $ \mathbf{WP}$ born from a multifold degeneracy point $p_0 \in M^3$ for a generic perturbation. In Section~\ref{ss:degpoints} we introduce the notion of degeneracy points, real and complex Weyl points, and the concept of the upper and lower bounds ($\sharp \mathbf{cWP}$ and $\sharp_{\text{alg}} \mathbf{WP}$, respectively) for $\sharp \mathbf{WP}$. In Section~\ref{ss:efftwofold} we summarize the analogous results for two-fold degeneracy points based on our previous papers \cite{PinterSW, BirthQ}. In Section~\ref{ss:chern} the lower bound $\sharp_{\text{alg}} \mathbf{WP}$ is expressed in terms of the first Chern numbers $c_1(\eta_j)$ of the eigenvector bundles $\eta_j$ (evaluated on a small 2-sphere surrounding $p_0$). In Section~\ref{ss:compweyl} the upper bound $\sharp \mathbf{cWP}$ is computed based on the results in Section~\ref{ss:proof-multi}. 

The results are illustrated on two physical examples, a spinful particle in an external magnetic field is discussed in Section~\ref{ss:exspin}, and an example for electronic band structures of crystalline materials is discussed in Section~\ref{sec:bandstructureexample}. In Section~\ref{ss:twofold} we show that our methods can be apllied for the characterization of two-fold degeneracy points without Schrieffer--Wolff.

\subsection{Degeneracy points of parameter dependent Hamiltonian}\label{ss:degpoints}

Let $H: M^3 \to \text{Herm}(n)$ be a parameter dependent Hamiltonian (a smooth map) with a degeneracy at $p_0 \in M^3$, that is, $A_0:=H(p_0) \in \Sigma_{\text{herm}}$. Assume that $\lambda_0:=\lambda_{l+1}$ is a $k$-fold degenerate eigenvalue of $A_0$, that is, 
\begin{equation}
    \lambda_1 \leq \lambda_2 \leq \dots \leq \lambda_l < \lambda_{l+1} = \lambda_{l+2} =  \dots = \lambda_{l+k} < \lambda_{l+k+1} \leq \dots \leq \lambda_n
\end{equation} 
holds for the eigenvalues of $A_0$.
 Locally, by introducing a local chart in $M^3$ around $p_0=0$, we consider the Hamiltonian germ $H: (\R^3, 0) \to (\text{Herm}(n), A_0)$. Assume that the germ $H$ is analytic\footnote{Although we formulate the statements for analytic maps, which is required to consider the complexification, many of them can be generalized   in smooth ($C^{\infty}$) context.}.

If a perturbation $H_t$ of $H$ is generic (transverse) to the branch $(\Sigma_{\text{herm}}, A_0; \lambda_0)$,  the elements of the preimage $\mathbf{WP}:=H_t^{-1}(\Sigma_{\text{herm}}, A_0; \lambda_0)$ are the (real) \emph{Weyl points} of the perturbation $H_t$  born from the degenerate eigenvalue $\lambda_0$. This definition coincide with the usual definition of the Weyl points in the physicist literature, see \cite[Theorem 3.4.5]{PinterSW}. Note also that in the physics literature, this terminology is ambiguous: many articles refer to any isolated two-fold degeneracy
point as a ‘Weyl point’, not only to the generic degeneracy points.

The \emph{complexification} $f:=H_{\mathbb{C}}:  (\mathbb{C}^3, 0) \to (\mathbb{C}^{n \times n}, A_0)$ of $H$  is the complex holomorphic map germ defined by the same power series as $H$, and its values can be arbitrary complex matrices (not necessarily hermitian). 

A perturbation $H_t$ of $H$ induces a perturbation $f_t$ of the complexification. If $f_t$ is generic with respect to the branch $(\Sigma, A_0; \lambda_0) \subset (\Sigma, A_0)$ corresponding to the geometric degenerate eigenvalue $\lambda_0$ (defined by Equation~\eqref{eq:projectionlocal}),  we call the elements of the preimage $\mathbf{cWP}:=f_t^{-1}(\Sigma, A_0; \lambda_0)$  \emph{complex Weyl points} of the perturbation $H_t$ (more precisely, $f_t$) born from the degenerate eigenvalue $\lambda_0$. Obviously, $\mathbf{WP} \subset \mathbf{cWP}$.

 By the properties of transversality, real and complex Weyl points are isolated two-fold degeneracy points. See \cite[Theorem 3.4.7]{PinterSW} in the real case, the complex case is similar. Let $\sharp \mathbf{WP}$  and $\sharp \mathbf{cWP}$ denote the number of Weyl points and the number of complex Weyl points, respectively, if the Hamiltonian $H$ (with $H(0)=A_0$ and its $k$-fold degenerate eigenvalue $\lambda_0$) and its perturbation $H_t$ is fixed. 
 
 Moreover, real Weyl points are endowed with a sign, also called a \emph{topological charge}, which is the sign of the transverse intersection of $H$ and $\Sigma_{\text{herm}}$. More precisely, an orientation can be defined for each smooth stratum of $(\Sigma_{\text{herm}}, A_0; \lambda_0)$. A smooth stratum, denoted by $\Sigma_{j,j+1}$ consists of the matrices with two-fold degenerate eigenvalue $\lambda_{j}=\lambda_{j+1}$, $j=l+1, \dots, l+k-1$. An orientation of $\Sigma_{j,j+1}$  defines a sign for the corresponding Weyl points $(x,y,z)$ with $H(x,y,z) \in \Sigma_{j,j+1}$, which is defined to be $+1$ if the Jacobian of $H$ at $(x,y,z)$ plus the tangent space of $\Sigma_{j,j+1}$ at $H(x,y,z)$ together defines positive orientation of $\text{Herm}(n) \cong \R^{n^2}$ determined by a fixed basis, see \cite{PinterSW}. Otherwise the sign of the Weyl point is $-1$. 
 
 In the next subsections we reformulate the sign of the Weyl points in terms of the effective germ and also in terms of the Chern numbers, which approache provides a coherent way to define the orientation for each $\Sigma_{j,j+1}$, but in this subsection it is enough that an orientation can be fixed for each $j$ independently, see below.   Note that the complex Weyl points do not have a sign, or their sign can be considered to be positive.

 Let $\sharp_{\text{alg}} \mathbf{WP}_{j, j+1}$ denote the signed number (algebraic number) of the Weyl points of $H_t$ on $\Sigma_{j,j+1}$, i.e. corresponding to the two-fold degeneracy $\lambda_{j}=\lambda_{j+1}$. We define 
 \begin{equation}\label{eq:algnumb}
 \sharp_{\text{alg}} \mathbf{WP}:= \sum_{j=l+1}^{l+k-1} |\sharp_{\text{alg}} \mathbf{WP}_{j, j+1}|,
 \end{equation}
 which, does not depend on the choice of the orientations on $\Sigma_{j,j+1}$.
 
   The main observation is the following.

     \begin{prop}\label{pr:bound1}
     \begin{equation}\label{eq:bound1}
      \sharp_{\text{alg}} \mathbf{WP}  \leq \sharp \mathbf{WP} \leq \sharp \mathbf{cWP}. \end{equation} 
      \end{prop}

  \begin{proof}  It is clear from the definitions. Indeed, $\mathbf{WP} \subset \mathbf{cWP}$ implies the second inequality. The first one follows from the fact that in $\sharp \mathbf{WP}$ each Weyl point is counted with positive sign, while in $\sharp_{\text{alg}} \mathbf{WP}$ some of them may be counted with negative sign. \end{proof}
      
  The number of real Weyl points $\sharp \mathbf{WP}$ depends on the choice of the perturbation $H_t$ of $H$, see Example~\ref{ex:spin1}. However,  $\sharp \mathbf{cWP} $ does not depend on the perturbation, it can be computed from the unperturbed map $H$ by algebraic methods. Indeed, $\sharp \mathbf{cWP} $ is exactly the multiplicity of $f=H_{\C}$ with respect to $(\Sigma, A_0; \lambda_0)$, hence our results can be used to compute it, see Section~\ref{ss:compweyl}.
  
  Also $ \sharp_{\text{alg}} \mathbf{WP} $ does not depend on the perturbation. In Section~\ref{ss:chern} we show that each $\sharp_{\text{alg}} \mathbf{WP}_{j, j+1}$ can be expressed in terms of the Chern numbers of the bands involved in the $k$-fold degeneracy of $H$, see also \cite{Bruno2006}. 

 \begin{prop}\label{pr:parity} The parities of $\sharp_{\text{alg}} \mathbf{WP}  $, $ \sharp \mathbf{WP}$ and  $ \sharp \mathbf{cWP}$ are the same.
 \end{prop}

 \begin{proof}
     For $\sharp_{\text{alg}} \mathbf{WP}  $ and $ \sharp \mathbf{WP}$ this follows from the definition, indeed, since $-1 \equiv 1$ modulo 2, the signed number of points has the same parity as their number. 
     
     The complex Weyl points are the complex solutions of the system of equations 
     \begin{equation}
         m_{ij} := M_{ij}\left(H_t(x,y,z) - \lambda \mathds{1}\right)=0,
     \end{equation}
     where $M_{ij}$ are the minors. Since the matrix $H_t(x,y,z) - \lambda \mathds{1}$ is hermitian, its diagonal minors $m_{ii}$ are real, that is, power series of the variables $x, y, z, \lambda$ with real coefficients. The other minors are in conjugate pairs $\overline{m}_{ij}=m_{ji}$, hence, by replacing them with $m_{ij}+m_{ji}$ and $i(m_{ij}-m_{ji})$, all the equations have real coefficients. Hence, the non-real solutions are in conjugate pairs, proving that $ \sharp \mathbf{WP}$ and  $ \sharp \mathbf{cWP}$ have the same parity. 
 \end{proof}

\subsection{Effective Hamiltonian and two-fold degeneracy}\label{ss:efftwofold}
For parameter values $(x,y,z)$ in a sufficiently small neighborhood of 0 consider the SW-decomposition of $A=H(x,y,z)$ with respect to $A_0$ and $\lambda_0$, see Equation~\ref{eq:swdec} or \cite[Theorem 3.1.2]{PinterSW}. Let $\widetilde{H}_{\text{eff}}^{\text{tr}=0}(x,y,z)$ denote the traceless part of the $k \times k$ effective Hamiltonian $\widetilde{H}_{\text{eff}}(x,y,z)$. (This is the upper-left $k \times k$ block of the $n \times n$ effective Hamiltonian $H_{\text{eff}}(x,y,z)$, whose traceless part is denoted by $H_{\text{eff}}^{\text{tr}=0}(x,y,z)$.)

$\widetilde{H}_{\text{eff}}^{\text{tr}=0}(x,y,z)$ can be expressed in a suitable basis of $\text{Herm}_0(k) \cong \R^{k^2-1}$ of traceless $k \times k$ hermitian matrices, see \cite[Section 2]{PinterSW}\footnote{Note that in \cite{PinterSW} the notation is different, the traceless effective Hamiltonian is denoted by $H_{\text{eff}}$.} In this way we obtain an analytic map germ
\begin{equation}
h: (\R^3, 0) \to (\R^{k^2-1}, 0),
\end{equation}
say, \emph{effective map germ}, see \cite[Section 3.4]{PinterSW}. Its complexification $h_{\C}: (\C^3, 0) \to (\C^{k^2-1}, 0)$ is defined by the same power series as $h$ considered with complex variables.

\begin{rem}
    It is possible to take first the complexification $f=H_{\C}: \C^3 \to \C^{n^2}$ of the Hamiltonian. Note that $\lambda_0$ is a strictrly $k$-fold eigenvalue of $A_0=f(0)$, since $A_0$ is hermitian. Then, by the complex SW decomposition we obtain the complex traceless effective Hamiltonian $\widetilde{f}^{\text{tr}=0}_{\text{eff}}(x,y,z)$, i.e. the traceless part of $\widetilde{f}_{\text{eff}}(x,y,z)$. Using the chosen basis of $\text{Herm}_0(k)$ over $\R$ as a complex basis of the traceless $k \times k$ matrices $\C^{k \times k}_0$, the expression of $\widetilde{f}^{\text{tr}=0}_{\text{eff}}(x,y,z)$ in this basis is equal to $h_{\C}(x,y,z)$. In other words, the effective germ of the complexification is the same as the complexification of the effective germ. 
\end{rem}

In the following we focus on two-fold degeneracy points ($k=2$). In this case $A_0 \in \Sigma_{\text{herm}} \subset \Sigma$ is a smooth point of the branches $(\Sigma_{\text{herm}}, A_0; \lambda_0)$ and $(\Sigma, A_0; \lambda_0)$, see Corollary~\ref{co:twofoldsmooth}. In this case the effective germ is $h=(h_1, h_2, h_3):(\R^3, 0) \to (\R^3, 0)$, noteworthy it is equi-dimensional.

A perturbation $H_t$ of $H$ induces a perturbation $h_t$ of the effective germ, $f_t$ of the complexification of $H$ and $h_{\C, t}$ of the complexification of the effective germ. Observe that 
\begin{equation}\label{eq:hameff}
    H_t^{-1}(\Sigma_{\text{herm}}, A_0; \lambda_0)=h_t^{-1}(0) \subset f_t^{-1}(\Sigma, A_0; \lambda_0)=h_{\C,t}^{-1}(0)
\end{equation}
holds by point (b) of the Complex SW chart theorem~\ref{th:compsw}, since the (complex) traceless effective Hamiltonian is degenerate at $(x,y,z)$ if and only if it is 0, that is, $h(x,y,z)=0$ ($h_{\C, t}=0$, respectively). 

$H_t$ is transverse to $(\Sigma_{\text{herm}}, A_0; \lambda_0)$ at the pre-images $(x,y,z)$ if and only if the Jacobian of $h_t$ at these points has  maximal rank 3, see \cite[Theorem 3.4.5]{PinterSW}. In this case these are Weyl points, i.e., $h_t^{-1}(0)=\mathbf{WP}$. The sign (topological charge) $+1$ or $-1$ of each Weyl point can be defined as the sign of the Jacobian determinant of $h_t$ at the point. 

Similarly, $f_t$ is transverse to $(\Sigma, A_0; \lambda_0)$ if and only if the Jacobian of $h_{\C,t}$ at the pre-images has  maximal rank 3. In this case they are complex Weyl points, i.e., $h_{\C,t}^{-1}(0)=\mathbf{cWP}$. Complex Weyl points do not have a charge (or it can be always considered to be $+1$), in fact, holomorphic maps always preserve the orientation induced by the complex structure.

 Recall the following notions, see details in See \cite[Appendix]{BirthQ} and \cite{MondBook}:
 \begin{itemize}
     \item If the degeneracy point is isolated in real sense, that is, $h^{-1}(0)=\{0\}$, then the \emph{local degree} (or \emph{index})  $\deg_0(h)$ of $h$ is defined as follows. We take a sphere $S^2_{\epsilon} \subset \R^3$ centered at 0 with a sufficiently small radius such that the only root of $h$ inside $S^2_{\epsilon}$ is 0. Then $\deg_0(h)$ is defined as the (global) degree of the normalized map
     \begin{equation}
         \widetilde{h}:= \left. \frac{h}{\|h\|} \right|_{S^2_{\epsilon}}: S^2_{\epsilon} \to S^2,
     \end{equation}
     that is, the signed number of preimages of a regular value. In physicist literature, $\deg_0(h)$ is called the topological charge of the two-fold degeneracy point. This generalizes the sign of the Weyl points, indeed, if $0$ is a Weyl point (i.e. the rank of the Jacobian of $h$ at 0 is 3), then $\deg_0(h)= \pm 1$ is the sign of it. The sum of the local degrees is preserved by perturbation, see Proposition 2 in \cite[Appendix]{BirthQ} or \cite[Cor. E.2]{MondBook}. In particular, for any generic perturbation,  $\sharp_{\text{alg}} \mathbf{WP}=\deg_0(h)$ holds for the signed number of Weyl points, see the discussion after \cite[Appendix]{BirthQ}, Proposition 2.
     \item The \emph{local algebra} of $h$ is the real algebra defined as $Q_0(h)=\mathcal{O}_3/I_h$, where $I_h=I(h_1,h_2,h_3) \subset \mathcal{O}_3$ is the ideal generated by the components of $h$. Similarly, the local algebra of the complexification is $Q_0(h_{\C})=\mathcal{O}_3/I_{h_{\C}}$. $Q_0(h_{\C})$ is a complex algebra, and it is the complexification of $Q_0(h)$, i.e. $Q_0(h_{\C}) \cong Q_0(h) \otimes \C$.
     \item The \emph{local multiplicity} of $h$ (or $h_{\C}$) is $\text{mult}_0 (h)= \dim_{\R} Q_0(h)=\dim_{\C} Q_0(h_{\C})=\text{mult}_0 (h_{\C})$. The local multiplicity is finite if and only if the degeneracy point is isolated in the complex sense, i.e. $h_{\C}^{-1}(0)=\{0\}$, which implies $h^{-1}(0)=\{0\}$, see \cite[Thm. D.5.]{MondBook}. Since the sign of the complex Weyl points of a perturbation are $+1$, we have $\text{mult}_0 (h_{\C})=\deg_0(h_{\C})=\sharp \mathbf{cWP}$ for the number of complex Weyl points, see the discussion after \cite[Appendix]{BirthQ}, Proposition 2. or \cite[Cor. E.3]{MondBook}.
 \end{itemize}

 Hence, for two-fold degeneracy points, the bounds for the number of Weyl points given in Equation~\eqref{eq:bound1} can be written as
 \begin{equation}\label{eq:bound2}
|\deg_0(h)| \leq \sharp \mathbf{WP} \leq \text{mult}_0(h).
 \end{equation}

 In Section~\ref{ss:twofold} we determine the local algebra $Q_0(h)$ directly from $H$ (up to isomorphism), avoiding SW and effective germ. The local degree $\deg_0(h)$ is determined by $Q_0(h)$ via the Eisenbud--Levine theorem \cite{Eisenbud1978,EisenbudLevineTeissier1977}.

\subsection{Chern numbers}\label{ss:chern}

Consider a setup similar to Section~\ref{ss:degpoints}, that is, a Hamiltonian $H: \R^3 \to \text{Herm}(n)$ and a degeneracy point $H(0)=A_0$ with a $k$-fold degenerate eigenvalue $\lambda_0$. Also we fix a generic perturbation $H_t$ of $H$.
By Equation~\eqref{eq:bound1}, $\sharp_{\text{alg}} \mathbf{WP}$ (defined by Equation~\eqref{eq:algnumb}) is a lower bound for the number of Weyl points $\mathbf{WP}$. Here we express $\sharp_{\text{alg}} \mathbf{WP}$ in terms of the Chern numbers of the bands involved in the degeneracy. In particular, this shows that $\sharp_{\text{alg}} \mathbf{WP}$ does not depend on the choice of the perturbation, it is determined by the unperturbed map $H$.

Let $\xi_1, \dots, \xi_n$ denote the eigenvector bundles (`bands') over $\text{Herm}(n) \setminus \Sigma$. Each one is a complex line bundle (with fibers isomorphic to the complex line $\C$), and $\xi_i$ is defined over the set of matrices with non-degenerate $i$-th eigenvalue, i.e. with $\lambda_{i-1} < \lambda_i < \lambda_{i+1}$. Denoting by $\Sigma_{i, i+1} \subset \Sigma_{\text{herm}}$ the set of matrices with $\lambda_i=\lambda_{i+1}$, $\xi_i$ is defined over $\text{Herm}(n) \setminus (\Sigma_{i-1, i} \cup \Sigma_{i, i+1})$.

Given a smooth oriented closed submanifold $M^2 \subset \R^3$ such that $\lambda_i$ is not degenerate for $H|_{M^2}$, consider the induced bundle $\eta_i=(H|_{M^2})^*(\xi_i)$. Let $c_1(\eta_i)[M^2]$ or shortly $c_1(\eta_i)$ denote the \emph{first Chern number} of $\eta_i$, that is, the first Chern class evaluated on the fundamental class of $M^2$. If $H|_{M^2}$ avoids $\Sigma$, then every Chern number is defined on $M^2$, and $\sum_{i=1}^n c_1(\eta_i)=0$, since it is the first Chern number of the trivial bundle $\bigoplus_{i=1}^n \eta_i$.

Note that in physicist literature $c_1(\eta_i)$ is mostly defined via the integral of the Berry curvature, which is the curvature form of an appropriate connection on $\eta_i$ called Berry connection (which actually agrees with the Chern connection) \cite{Asboth, ArnoldSelMath1995}.
Although this approach is effective for computations, we use here the following equivalent description of the Chern numbers.

At a point $p \in M^2$, let $L_i \subset \C^n$ be the eigenspace of $H(p)$ corresponding to $\lambda_i$. Consider $L_i$ as a point of the complex projective space $\C P^{n-1}$. Let $\widetilde{\eta}_i: M^2 \to \C P^{n-1}$ be the smooth map defined by $\widetilde{\eta}_i(p):=L_i $. Fix a complex projective hypersurface $\C P^{n-2} \subset \C P^{n-1}$ such that $\widetilde{\eta}_i$ is transverse to $\C P^{n-2}$ (to reach this, smooth homotopy of $\widetilde{\eta}_i$ or $H|_{M^2}$ is also allowed). The orientations of $M^2$ and $\C P^{n-2}$ define a sign $\pm 1$ for each preimage $\widetilde{\eta}_i^{-1}(\C P^{n-2})$, let $ \sharp_{\text{alg}}\widetilde{\eta}_i^{-1}(\C P^{n-2})$ denote the signed number of preimages.
\begin{prop}\label{pr:cherndef}
    $ \sharp_{\text{alg}}\widetilde{\eta}_i^{-1}(\C P^{n-2})=c_1(\eta_i)$.
\end{prop}

\begin{proof} For the notions corresponding to homology and characteristic classes see e.g. \cite{MilnorStasheff}.
%Milnor, Stasheff: Characteristic classes
Let $\tau_{n-1}=\tau$ be  the tautological complex line bundle over $\C P^{n-1}$, and let $c_1(\tau) \in H^2(\C P^{n-1}, \Z) $ be the first Chern class of $\tau$. Since $\widetilde{\eta}_i$ is the inducing map of $\eta_i$, that is, $\eta_i=\widetilde{\eta}_i^*(\tau)$, we have $c_1(\eta_i)=\widetilde{\eta}_i^*(c_1(\tau))[M^2]$, where $[M^2] \in H_2(M^2, \Z)$ is the oriented fundamental class. 

The Poincaré dual of  $c_1(\tau)$  is the homology class $[\C P^{n-2}] \in H_{2n-4}(\C P^{n-1}, \Z)$ of any complex projective hypersurface $\C P^{n-2} \subset \C P^{n-1}$. By standard homological arguments,
\begin{equation}
c_1(\eta_i)=\widetilde{\eta}_i^*(c_1(\tau))[M^2]=c_1(\tau)(\widetilde{\eta}_{i*}[M^2])=[\C P^{n-2}] \cdot (\widetilde{\eta}_{i*}[M^2])=\sharp_{\text{alg}}\widetilde{\eta}_i^{-1}(\C P^{n-2}),
\end{equation}
where $[A] \cdot [B]$ denote the intersection number of two homology classes of complementary dimensions. This proves the proposition.
\end{proof}

\begin{cor}\label{co:nondeg}
    Assume that $W^3 \subset \R^3$ is a compact submanifold such that $\lambda_i$ is non-degenerate for the restriction $H|_{W^3}$ of a Hamiltonian $H: \R^3 \to \text{Herm}(n)$. Then $c_i(\eta_i)[\partial W^3]=0$.
\end{cor}

\begin{proof}
    In this case $\widetilde{\eta}_i: \partial W^3 \to \C P^{n-1}$ extends to $W^3$, hence $\widetilde{\eta}_{i*}[\partial W^3]=0 \in H_2(\C P^{n-1}, \Z)$. This implies that $ \sharp_{\text{alg}}\widetilde{\eta}_i^{-1}(\C P^{n-2})=0$.  (Note that this corollary follows from the Stokes theorem if $c_1$ is defined via the integral of the Berry curvature $\Omega$, since  $d \Omega$ (the divergence) is 0.)
\end{proof}

As this corollary suggests, the Chern numbers can be expressed in terms of the degeneracy points inside $M^2$. In case a two-fold degeneracy ($H(0)=A_0$ with $\lambda_l < \lambda_{l+1}=\lambda_{l+2} <\lambda_{l+3}$; $\lambda_0:=\lambda_{l+1}$), let $M^2:=S^2_{\epsilon} \subset \R^3$ be a sphere centered at 0 with sufficiently small radius such that 0 is the only enclosed degeneracy point. Consider  $c_1(\eta_{l+1})$ and $c_1(\eta_{l+2})$ evaluated on this sphere. Recall from Section~\ref{ss:efftwofold} the effective germ $h: (\R^3, 0) \to (\R^3, 0)$ and its local degree $\deg_0(h)$, which is the topological charge of the degeneracy point at 0. 

\begin{prop}\label{pr:ketszeres}
    In the above setup, $-c_1(\eta_{l+1})=c_1(\eta_{l+2})=\deg_0(h)$.
\end{prop}

\begin{proof}
We show that the Chern numbers of the two bands are equal to the Chern numbers of the $2 \times 2$ traceless effective Hamiltonian, and they are equal to $\pm \deg_0(h)$. 

  We perform a (hermitian) SW decomposition around the degeneracy point $A_0 \in \Sigma_{l+1,l+2}$, see the definitions around the SW chart Equation~\eqref{eq:swdec}. Let  $A_0'=G_0^{-1} \cdot   A_0 \cdot G_0$ be the partial diagonalization of $A_0$, however, in the hermitian case $A_0'$ can be assumed to be diagonal. For $(x,y,z)$ in a neighborhood of 0, the decomposition is
  \begin{equation}
  H(x,y,z)=G_0 \cdot e^{S(x,y,z)} \cdot (A_0' +C(x,y,z)+H_{\text{eff}}(x,y,z)) \cdot e^{-S(x,y,z)} \cdot G_0^{-1}.
  \end{equation}
  
  The map $H|_{S^2_{\epsilon}}: S^2_{\epsilon} \to \text{Herm}(n) \setminus (\Sigma_{l, l+1} \cup \Sigma_{l+1, l+2} \cup \Sigma_{l+2, l+3})$ is homotopic to $(A_0'+H^{\text{tr=0}}_{\text{eff}})|_{S^2_{\epsilon}}$.  The homotopy is realized by $(tS, tC, t \cdot \text{tr}(H_{\text{eff}}))$, joining these terms of the SW decomposition with 0, and simultaneously joining $G_0$ with the identity in the unitary group $U(n)$. It is really a homotopy between $H|_{S^2_{\epsilon}}$ and $(A_0'+H^{\text{tr=0}}_{\text{eff}})|_{S^2_{\epsilon}}$ in $\text{Herm}(n) \setminus (\Sigma_{l, l+1} \cup \Sigma_{l+1, l+2} \cup \Sigma_{l+2, l+3})$, i.e. it leaves the eigenvalues $\lambda_{l+1}$ and $\lambda_{l+2}$ non-degenerate on $S^2_{\epsilon}$. 
  Therefore, 
  \begin{equation}
      c_1(H^*(\xi_{l+i}))[S^2_{\epsilon}]=c_1((A_0'+H^{\text{tr=0}}_{\text{eff}})^*(\xi_{l+i}))[S^2_{\epsilon}]=c_1 ((\widetilde{H}^{\text{tr=0}}_{\text{eff}})^*(\xi_{i}))[S^2_{\epsilon}],
  \end{equation}
 for $i=1,2$. The second equation comes from the observation that the two dimensional subspace of $\C^n$ spanned by the eigenspaces $(A_0'+H^{\text{tr=0}}_{\text{eff}})^*(\xi_{l+i})$  of $A_0'+H^{\text{tr=0}}_{\text{eff}}(x,y,z)$ is constant, that is, it does not depend on $(x,y,z) \in S^2_{\epsilon}$. The $2 \times 2$ traceless effective Hamiltonian $\widetilde{H}^{\text{tr=0}}_{\text{eff}}$ acts on this two dimensional space. Denoted by  $\eta_{\text{eff},i}:=(\widetilde{H}^{\text{tr=0}}_{\text{eff}})^*(\xi_{i})$ the eigenvector bundles of $\widetilde{H}^{\text{tr=0}}_{\text{eff}}$, the above argument shows that their Chern numbers are equal to the original Chern numbers, that is, $c_1(\eta_{\text{eff},i})=c_1(\eta_{l+i})$. 
 
 To show that they are equal to $\pm \deg_0 (h)$, recall that 
  \begin{equation}
      \widetilde{H}^{\text{tr=0}}_{\text{eff}}(x,y,z)=h_1(x,y,z) \cdot \sigma_x  + h_2(x,y,z) \cdot \sigma_y + h_3(x,y,z) \cdot \sigma_z \in \text{Herm}_0(2).
  \end{equation} 
% Hence, it is enough to show that the Chern numbers of the two eigenvector bundles $\eta_{\text{eff},1}, \eta_{\text{eff},2} $ of $\widetilde{H}^{\text{tr=0}}_{\text{eff}}$ are equal to $\pm \deg_0(h)$, where $\eta_{\text{eff},i}=(\widetilde{H}^{\text{tr=0}}_{\text{eff}})^*(\xi_{i})$
 The inducing maps $\widetilde{\eta}_{\text{eff}, i}: S^2_{\epsilon} \to \C P^1 \cong S^2$ of the bundles $\eta_{\text{eff},i}$ are the compositions 
 \begin{equation}\label{eq:compozit}
     \widetilde{\eta}_{\text{eff}, i}=\widetilde{\xi}_i \circ \widetilde{H}^{\text{tr=0}}_{\text{eff}}|_{S^2_{\epsilon}},
 \end{equation}
 where $\widetilde{\xi}_i: \text{Herm}_0(2) \setminus \{0\} \to \C P^1$ are the inducing maps of the eigenvector bundles, that is, $\widetilde{\xi}_i^*(\tau_1)=\xi_i$ for $i=1,2$. Fixing a (generic) point $\{L_0\}=\C P^0 \subset \C P^1$, a direct calculation of the eigenvectors of $x\sigma_x+y \sigma_y+z \sigma_z$ shows that they depend only on the direction $(x,y,z)/\sqrt{x^2+y^2+z^2} \in S^2$. Moreover, there is only one point on the unit sphere $S^2 \subset \text{Herm}_0(2)$ whose $i$-th eigenspace is $L_0$, that is,  $(\widetilde{\xi}_i|_{S^2})^{-1}(L_0)$ is one point, and the map preserves the orientation for $i=2$ and it reverses the orientation for $i=1$. Hence,
 \begin{equation}
     c_1(\xi_1)[S^2]=\deg (\widetilde{\xi}_1|_{S^2})=-1 \text{ and }
     c_1(\xi_2)[S^2]= \deg (\widetilde{\xi}_2|_{S^2})=1.
     \end{equation}

     Then, by Proposition~\ref{pr:cherndef} and Equation~\eqref{eq:compozit},
     \begin{equation}
         c_1(\eta_{\text{eff},i})[S^2_{\epsilon}]=
         \sharp_{\text{alg}}\widetilde{\eta}_{\text{eff},i}^{-1}(L_0)
         =
         \deg (\widetilde{\eta}_{\text{eff}, i})=   \deg (\widetilde{\xi}_i|_{S^2}) \cdot \deg  \left( \left. \frac{h}{\|h\|} \right|_{S^2_{\epsilon}}\right)=(-1)^i \deg_0 (h),
     \end{equation}
 proving the proposition.
\end{proof}

Assume that $M^2 \subset \R^3$ is the boundary of a compact submanifold $Q^3 \subset \R^3$, and a Hamiltonian $H: \R^3 \to \text{Herm}(n)$ has only generic degeneracy (Weyl points) corresponding to $\lambda_i$ in $Q^3$. In other words, $H|_{Q^3}$ is transverse to the closure of $\Sigma_{i-1,i} \cup \Sigma_{i, i+1}$. Hence, $\sharp_{\text{alg}} (H|_{Q^3})^{-1} (\Sigma_{i, i+1})$ and $ \sharp_{\text{alg}} (H|_{Q^3})^{-1} (\Sigma_{i-1, i})$ is the signed number of the Weyl points in $Q^3$ corresponding to the degeneracy of $\lambda_i$ with $\lambda_{i+1}$, and the degeneracy of $\lambda_{i-1}$ with $\lambda_{i}$, respectively.

\begin{prop}\label{pr:chernminus}
In the above setup, 
\begin{equation}
    c_1(\eta_i)[M^2]= \sharp_{\text{alg}} (H|_{Q^3})^{-1} (\Sigma_{i-1, i})- \sharp_{\text{alg}} (H|_{Q^3})^{-1} (\Sigma_{i, i+1}).
\end{equation}
\end{prop}

\begin{proof}
    Note that there are only a finite number of Weyl points inside $Q^3$, otherwise their accumulation point would be a non-generic degeneracy point. Let they be denoted by $p_j \in \Sigma_{i-1,i}$ and $q_j \in \Sigma_{i,i+1}$. 
    
    Let $B^3_{q_j} \subset Q^3$ and $B^3_{q_j} \subset Q^3$ be closed balls centered at $p_j$ and $q_j$, respectively, with so small radius such that $p_j$ ($q_j$) is the only degeneracy point inside (corresponding to the degeneracy of $\lambda_i$). 
    Define $W^3:=Q^3 \setminus (\bigcup \text{int}(B^3_{q_j}) \cup \bigcup \text{int}(B^3_{q_j}))$. Its boundary is
    \begin{equation}
        \partial W^3=M^2 \cup \bigcup_j (-S^2_{p_j}) \cup \bigcup_j (-S^2_{q_j}),
    \end{equation}
    where $-S^3_{p_j}$ is the boundary sphere $S^2_{p_j}=\partial B^3_{p_j}$ with reversed orientation, similarly, $-S^2_{q_j}=-\partial B^3_{q_j}$. 
    
    Then $c_1(\eta_i)[W^3]=0$ holds by Corollary~\ref{co:nondeg}, since $\lambda_i$ is non-degenerate in $W^3$. Hence, 
    \begin{equation}
        c_1(\eta_i)[M^2]= \sum_j c_1(\eta_i)[S^2_{p_j}]+\sum_j c_1(\eta_i)[S^2_{q_j}].
    \end{equation}
By Proposition~\ref{pr:ketszeres}, $c_1(\eta_i)[S^2_{p_j}]$ is equal to the sign (topological charge) $\pm 1$ of the Weyl point $p_j$, while $c_1(\eta_i)[S^2_{q_j}]$ is the opposite of the sign of the Weyl point $q_j$. Therefore, 
\begin{equation}
    \sum_j c_1(\eta_i)[S^2_{p_j}]=\sharp_{\text{alg}} (H|_{Q^3})^{-1} (\Sigma_{i-1, i}) \text{ and } \sum_j c_1(\eta_i)[S^2_{q_j}]=-\sharp_{\text{alg}} (H|_{Q^3})^{-1} (\Sigma_{i, i+1}),
\end{equation}
proving the proposition.
\end{proof}

Let us return to the main focus of this section, that is, a Hamiltonian $H: \R^3 \to \text{Herm}(n)$ and a degeneracy point $H(0)=A_0$ with a $k$-fold degenerate eigenvalue $\lambda_l < \lambda_{l+1}=\dots =\lambda_{l+k}=:\lambda_0 < \lambda_{l+k+1}$. We also fix a generic perturbation $H_t$ of $H$. Fix a small ball $B^3_{\epsilon} \subset \R^3$ with boundary $S^2_{\epsilon}$ centered at 0 which does not contain other degeneracy points, and the Weyl points of $H_t$ born from the degeneracy of $\lambda_0$ at 0 are contained in $B^3_{\epsilon}$. Write $c_1(\eta_j):=c_1(\eta_j)[S^2_{\epsilon}]$.

Recall that $\mathbf{WP}_{j,j+1}=H_t^{-1}(\Sigma_{j, j+1})$ denote the set of Weyl points corresponding to the degeneracy of $\lambda_j$ and $\lambda_{j+1}$, where $j=l+1, \dots, l+k-1$. Let $\sharp_{\text{alg}}\mathbf{WP}_{j,j+1}$ be their signed number.  Proposition~\ref{pr:chernminus} implies the following:

\begin{cor}
In the above setup, 
\begin{equation}
    c_1(\eta_j)=\sharp_{\text{alg}}\mathbf{WP}_{j-1,j} - \sharp_{\text{alg}}\mathbf{WP}_{j,j+1},
\end{equation}
hence,
\begin{equation}\label{eq:lowerchern}
  \sharp_{\text{alg}}  \mathbf{WP}_{j,j+1}= - \sum_{i=l+1}^j c_1(\eta_i).
\end{equation}
\end{cor}

In particular, $\sharp_{\text{alg}}  \mathbf{WP}_{j,j+1}$ does not depend on the choice of the perturbation $H_t$ of $H$. 

\begin{cor}The lower bound $\sharp_{\text{alg}} \mathbf{WP}$ for the number of Weyl points $\sharp \mathbf{WP}$ of $H_t$ in Equation~\eqref{eq:bound1} can be expressed in terms of the Chern numbers as
\begin{equation}\label{eq:lowerbound}
    \sharp_{\text{alg}} \mathbf{WP}= \sum_{j=l+1}^{l+k-1} |\sharp_{\text{alg}} \mathbf{WP}_{j, j+1}|=\sum_{j=l+1}^{l+k-1} \left| \sum_{i=l+1}^j c_1(\eta_i) \right|.
\end{equation}
\end{cor}

\begin{rem}
    For 2-parameter symmetric families $H: \R^2 \to \text{Symm}_{\R}(n)$ one can also consider the eigenvector bundles, which are real line bundles in this case. Instead of Chern numbers only the $\Z_2$  Stiefel--Whitney numbers are defined for them. By analogy of the above argument, one can patch up a lower bound for the number of Weyl points from the Stiefel--Whitney numbers.
\end{rem}

\subsection{Complex Weyl points born from a multi-fold degeneracy point}\label{ss:compweyl}

As in Section~\ref{ss:degpoints}, we consider a Hamiltonian $H: \R^3 \to \text{Herm}(n)$ and a degeneracy point $H(0)=A_0$ with a $k$-fold degenerate eigenvalue $\lambda_l < \lambda_{l+1}=\dots =\lambda_{l+k}=:\lambda_0 < \lambda_{l+k+1}$. We also fix a generic perturbation $H_t$ of $H$. Let $f=H_{\C}$ be the complexification of $H$. 

By Equation~\eqref{eq:bound1}, the number of complex Weyl points $\sharp \mathbf{cWP} $ is an upper bound for the number of Weyl points $\mathbf{WP}$ of $H_t$. We apply the arguments in Section~\ref{s:proofs} to determine $\sharp \mathbf{cWP} $. By Theorem~\ref{th:pullback} we conclude the following.

\begin{cor}\label{co:pullback}

   In the above setup,
    \begin{equation}
        \sharp \mathbf{cWP} = \dim_{\R} \frac{\mathcal{O}^{\R}_4}{J^{\R}},
    \end{equation}
    where $J^{\R} \subset \mathcal{O}^{\R}_4$ is the ideal generated by the $(n-1) \times (n-1) $ minors \begin{equation}
        M_{ij}(H(x,y,z)-(\lambda +\lambda_0) \mathds{1})
    \end{equation}
    of $H(x,y,z)-(\lambda +\lambda_0)\mathds{1}$.
    \end{cor}

    \begin{proof}
        By Theorem~\ref{th:pullback}, 
         \begin{equation}
        \sharp \mathbf{cWP} = \dim_{\C} \frac{\mathcal{O}^{\C}_4}{J^{\C}},
    \end{equation}
    where $J^{\C} \subset \mathcal{O}^{\C}_4$ is the ideal generated by the $(n-1) \times (n-1) $ minors \begin{equation}
        M_{ij}(f(x,y,z)-(\lambda +\lambda_0) \mathds{1})
    \end{equation}
    of $f(x,y,z)-(\lambda +\lambda_0)\mathds{1}$. Since $\mathcal{O}^{\C}_4/J^{\C}=(\mathcal{O}^{\R}_4/J^{\R}) \otimes \C$, the two dimensions are equal.
    \end{proof}

  \begin{rem}\label{re:compiso}
      The finiteness of $\dim_{\C}(\mathcal{O}^{\C}_4/J^{\C})=\dim_{\R}(\mathcal{O}^{\R}_4/J^{\R})$ is equivalent with the fact that $f$ is isolated with respect to $(\Sigma, A_0; \lambda_0)$, i.e. the degeneracy point is isolated in the complex sense. Indeed, both are equivalent with the fact the vanishing locus of $J$ in $(\C^4, 0)$ is only one point origin. See Example~\ref{ex:onepoint}, cf. \cite[Theorem D.5.]{MondBook}.
  \end{rem}

For linear Hamiltonian we have the following consquence of Corollary~\ref{co:kanegyzet_symm}:

\begin{cor}\label{co:pyram}
    If in addition $A_0=0$ and $H: (\R^3, 0) \to (\text{Herm}(n), 0)$ as above is a linear map such that its complexification $f=H_{\C}$ is isolated with respect to $(\Sigma, 0)$, then
       \begin{equation}
           \sharp \mathbf{cWP}= \frac{k^2(k^2-1)}{12}.
     \end{equation}
\end{cor}

More generally, by Corollary~\ref{co:generic}, $\sharp \mathbf{cWP}= k^2(k^2-1)/12$ holds not only for linear Hamiltonian, but for the most generic Hamiltonians with a $k$-fold degeneracy point. We suggest the name \emph{$k$-fold Weyl point} for these degeneracy points, on the analogy of the generic 2-fold degeneracy points called Weyl points. 

\begin{rem}\label{re:lowerweyl} There is an important difference between the 2-fold and $k$-fold case (for $k>2$). Fora a 2-fold Weyl point the lower and upper bounds are equal, namely, $\sharp_{\text{alg}} \mathbf{WP}=\sharp \mathbf{cWP}=1$, hence, $\sharp \mathbf{WP}=1$ holds for any perturbation, according to the fact that the Weyl point is stable. In contrast, for $k$-fold Weyl points the upper bound is always $\sharp \mathbf{cWP}= k^2(k^2-1)/12$ by definition, but there are $k$-fold Weyl points with different lower bound $\sharp_{\text{alg}} \mathbf{WP}$. In fact, this lower bound is expressed in terms of the Chern numbers, see Equation~\eqref{eq:lowerbound}, and the Chern numbers can be different. See the example in Section~\ref{sec:bandstructureexample}.

It would be interesting to determine the possible patterns of Chern numbers $\{c_1(\eta_j)  \}$ for $k$-fold Weyl points, at least for linear Hamiltonians. They may be determined from the real algebra $\mathcal{O}^{\R}_4/J^{\R}$, on analogy of the Eisenbud--Levine theorem, relevant for non-generic 2-fold degeneracy points, which determines $\deg_0(h)$ from the local algebra $Q_0(h)=\mathcal{O}_3^{\R}/I_h$, see Section~\ref{ss:efftwofold} and \ref{ss:twofold}.

\end{rem}

\subsection{Example: a spinful particle in an external magnetic field}\label{ss:exspin}

Recall that the irreducible representations of $SU(2)$ are labeled by the half integers $s \in \frac{1}{2} \N$, where the corresponding representation space is $\C^{2s+1}$. 
The Hamiltonian of a spin-$s$ particle in a magnetic field is a linear map~\cite{Sakurai, FultonHarris}
\begin{equation}
    H: \R^3 \to \text{Herm}(2s+1),
\end{equation}
which is in fact the Lie algebra representation $H: \mathfrak{su}(3) \to \mathfrak{su}(2s+1)$, up to the identification of $\mathfrak{su}(n) =i \cdot \text{Herm}_0(n)$ with $\text{Herm}_0(n)$ (traceless hermitian matrices). In the $s=1/2$ case the spin operators $S_x$, $S_y$, $S_z$ are represented by the Pauli matrices $\frac{1}{2} \sigma_x$, $\frac{1}{2}\sigma_y$, $\frac{1}{2} \sigma_z$, both sets satisfying the same commutation relations (we measure angular momentum in $\hbar$ units, so the spin operators are dimensionless).

Therefore, $H(x,y,z)=xS_x+yS_y+zS_z$ is interpreted as the Hamiltonian corresponding to the magnetic field vector $\mathbf{B}=(x,y,z)$. Its eigenvalues, that is, the possible energy levels of a spin-$s$ particle in the magnetic field $\mathbf{B}$, are  $\lambda \in \{-rs, -r(s-1), \dots, r(s-1), rs\}$, where $r= \| \mathbf{B} \|=\sqrt{x^2+y^2+z^2}$. This follows from the argument below.

Considering the adjoint action (conjugation) of $SU(2)$ on $\text{Herm}_0(2)$ and $\text{Herm}_0(2s+1)$, the spin Hamiltonian $H$ is an $SU(2)$-equivariant map. Intuitively, rotating simultaneously the magnetic field and the spin, the Hamiltonian remains unchanged. Since the orbits of the adjoint action of $SU(2)$ on $\mathfrak{su}(2)\cong \R^3$ are the spheres ($r=\text{constant}$), this property is inherits for the image of $H$. Hence, $H(x,y,z)$ is unitary equivalent to $H(0,0,r)=rS_z$, in particular, the eigenvalues of $H(x,y,z)$ are $\lambda=-rs, -r(s-1), \dots, r(s-1), rs$.

$H(0)=0$ is an $n=k=2s+1$-fold degeneracy point of $H$. Intuitively this corresponds to the fact that for zero magnetic field the system has a rotational symmetry, which implies that the energy levels are not separated. 
Since for $r > 0$ all the eigenvalues of $H(x,y,z)$ are different, the degeneracy point at 0 is isolated, more precisely, $H$ is isolated with respect to $\Sigma_{\text{herm}}^{(2s+1)}$, that is, $H^{-1}(\Sigma_{\text{herm}}^{(2s+1)})=\{0\}$.
Furthermore, we show that the degeneracy point at 0 is isolated in the complex sense too. Let $f:=H_{\C}: \C^3 \to \C^{2s+1}$ be the complexification of the spin Hamiltonian $H$.

\begin{prop}\label{pr:spinhamisol}
   $f$ is isolated with respect to $\Sigma^{(2s+1)}$, that is, $f^{-1}(\Sigma^{(2s+1)})=\{0\}$. 
\end{prop}

\begin{proof}
    We use the explicit form of the spin operators in a basis where $S_z$ is diagonal. With appropriate choices of the relative phases of the basis vectors, the spin-$s$ operators can be brought to the following form~\cite{Sakurai,Edmonds1996}:
    \begin{eqnarray*}
        (S_\pm)_{ab} &=& \sqrt{s(s + 1) - b(b \pm 1)}\delta_{a, b\pm 1},\\
        S_x &=& \frac{1}{2}\left(S_+ + S_-\right),\\
        S_y &=& \frac{1}{2i}\left(S_+ - S_-\right),\\
        (S_z)_{ab} &=& a \;\delta_{ab},
    \end{eqnarray*}
    where the indices $a,b$ take values $s \geq a,b \geq -s$ in steps of 1, and the Kronecker delta $\delta_{ab}$ is 1 if $a=b$ and 0 otherwise.
    The key properties of the spin matrices in this basis are that $S_z$ is diagonal, $S_x$ and $S_y$ are only nonzero on the first diagonals above and below the main digonal. Furthermore, $S_x$ is purely real, and $S_y$ is purely imaginary, and their entries satisfy $(S_y)_{ab} = -i (S_x)_{ab}$ if $a>b$ and $(S_y)_{ab} = i (S_x)_{ab}$ if $a<b$.

    Now we consider the lifted complexified spin Hamiltonian $\widetilde{f}(x,y,z, \lambda)=xS_x+yS_y+zS_z-\lambda \mathds{1}_k$ in this basis. 
    Evaluating its minor corresponding to the bottom left entry, we find that it is the determinant of a lower triangular matrix, and only the first diagonal above the main diagonal of the full matrix contributes: 
    \begin{equation}
        M_{-s, s}(\widetilde{f}) = \frac{1}{2^{2s}} \left(\prod_{b = -s}^{s-1} \sqrt{s(s + 1) - b(b + 1)}\right) \left(x - i y\right)^{2s}.
    \end{equation}
    This shows that $M_{-s, s}(\widetilde{f}) = 0$ iff $x = i y$.
    The same consideration for the minor of the upper right entry results in $x = -i y$, hence $x = y = 0$.
    Because $f(0, 0, z) = z S_z$, and $S_z$ is diagonal with all different diagonal entries, it is degenerate iff $z=0$.
\end{proof}

\begin{remark} 
    We note that spin Hamiltonians with complex parameters do have points $(x,y,z) \in \C^3$, where the algebraic multiplicity of some eigenvalue of $H_{\C}(x,y,z)$ is greater than 1, but the corresponding geometric multiplicity (the dimension of the corresponding eigensubspace) remains 1.
    For example, at parameters $y = \pm i x$ and $z = 0$, the spin Hamiltonian corresponds to a single Jordan block of size $2s+1$ corresponding to eigenvalue $0$ with algebraic degeneracy $2s+1$.
    This fact is not surprising, because the algebraic degenerate matrices form a codimension 1 algebraic subset in $\C^{n \times n}$, cf. Section~\ref{ss:degcomp}.
\end{remark}

By Corollary~\ref{co:kanegyzet_symm} we conclude that

\begin{cor}\label{co:upperspin}
    Any generic perturbation $H_t$ of the spin-$s$ Hamiltonian has $\sharp \mathbf{cWP}=k^2(k^2-1)/12$ complex Weyl points (where $k=2s+1$), which is an upper bound for the number of real Weyl points $\sharp \mathbf{WP}$.
\end{cor}

For a lower bound, consider the eigenvector bundles $\eta_a $ of $H$ corresponding to $a \in \{-s, -s+1, \dots, s-1, s\} $ over a small sphere $S^2 \subset \R^3$ centered at $0$. That is, the fiber $\eta_a(x,y,z) \subset \C^k$ over $(x,y,z) \in S^2$ satisfies $H(x,y,z)\cdot v=r a \cdot v$ for $v \in \eta_a(x,y,z)$, where $r=\sqrt{x^2+y^2+z^2}$. The Chern numbers are computed, for example, in Ref.~\cite{Nakahara}:

\begin{lem}\label{le:chernspin} For spin Hamiltonian, the first Chern number of $\eta_a $ is $c_1(\eta_a)=2a$. 
\end{lem}

\begin{cor}\label{co:spinbounds}
    \begin{equation}
  \frac{k(k^2-1)}{6}     \leq \sharp \mathbf{WP} \leq \frac{k^2(k^2-1)}{12}
    \end{equation}
    holds for the number of Weyl points $\mathbf{WP}$ of any generic perturbation $H_t$ of $H$ (where $k=2s+1$).
\end{cor}

\begin{proof}
    The upper bound is already established in Corollary~\ref{co:upperspin}. By Equation~\eqref{eq:bound1}, $\sharp_{\text{alg}} \mathbf{WP}$ is a lower bound for $\sharp \mathbf{WP}$, and it can be computed by Equation~\eqref{eq:lowerbound} as
    \begin{equation}
    \sharp_{\text{alg}} \mathbf{WP}=\sum_{a=-s}^{s} \left| \sum_{b=-s}^{a} c_1(\eta_b)  \right| =\sum_{a=-s}^{s} \left| \sum_{b=-s}^{a} 2b  \right|
    =\sum_{a=-s}^{s} \left| (a-s)(a+s+1)  \right|=\frac{k(k^2-1)}{6},
    \end{equation}
 the last equation can be proved e.g. by induction.
\end{proof}

Note that the lower bound agrees with the upper bound for symmetric families, see Section~\ref{ss:proof-multi}, this is the sequence of \emph{tetrahedral numbers} $\binom{ k+1}{3}$. 

In case of spin-1 Hamiltonian, the lower bound is $\sharp_{\text{alg}} \mathbf{WP}=4$, and the upper bound is $\sharp \mathbf{cWP}=6$. In Section~\ref{ex:spin1} two perturbations are given, one with $\sharp \mathbf{WP}=4$ and the other one with $\sharp \mathbf{WP}=6$. The first one is simpler than the second one, indeed, it is a \emph{constant perturbation (translation)}, that is,  it has the special form $H_t=H+tH_1$ with $H_1 \in \text{Herm}(3)$.

\begin{rem} We list here some related open questions.
\begin{itemize}
    \item We do not know whether the bounds are sharp or not for a spin Hamiltonian, $s>1$.
    \item It can be shown that 6 real Weyl points cannot be obtained from a constant perturbation (translation)  of the spin-1 Hamiltonian $H$. In contrast, constant perturbation of other (random) $3 \times 3$ Hamiltonian can have 6 real Weyl points. Is there any deeper reason of this phenomenon?
    \item What are the possible Chern number configurations $\{c_1(\eta_1), c_1(\eta_2), c_1(\eta_3) \}$ of $3 \times 3$ linear Hamiltonians?
\end{itemize}
    
\end{rem}

\subsection{Example: Electronic band structures of crystalline materials}
\label{sec:bandstructureexample}

In solid-state physics, electronic structure calculations aim to understand materials starting from information about the crystal structure and the chemical composition.
The crystal's symmetry group always contains the discrete translation subgroup isomorphic to $\mathbb{Z}^3$ (in 3 dimensions).
This allows Fourier-transformation to reciprocal space, the so-called Brillouin-zone isomorphic to a torus, whose elements are the wavenumber vectors ${\bf k}=(k_x, k_y, k_z) \in T^3$.
Knowledge of the chemical composition, more precisely, the number and type of atoms in a unit cell, allows to restrict the Hilbert-space to a finite number of dimensions, corresponding to the atomic orbitals of a unit cell most relevant to the electronic properties.

With these considerations, one can construct the Bloch-Hamiltonian $H: T^3 \to \text{Herm}(n)$, which is a matrix-valued function from the Brillouin-zone torus  to the space of $n\times n$ Hermitian matrices (operators acting on the restricted Hilbert-space of a single unit cell), ~\cite{AshcroftMermin}.
In most physical cases we can make the further assumption that the Hamiltonian is local in real-space (i.e. its matrix elements decay exponentially with distance), which guarantees that the Bloch-Hamiltonian $H$ is an analytical map in reciprocal space (with respect to the canonical chart on $T^3$).
Often in practice, the Bloch-Hamiltonian is only of interest in the vicinity of some special $\bf k$  point (for example $\bf k=0$). 
In this case the Bloch-Hamiltonian can be Taylor-expanded, and treated as an analytic map germ around this point, this is the so-called $\bf k \cdot p$ Hamiltonian. 

So far we have only used the translational symmetry of the crystal.
Other symmetries, such as rotations and mirrors, have further profound effects.
These symmetries can enforce multifold degeneracies at certain high-symmetry $\bf k$ points through the higher-dimensional representations of the symmetry group ~\cite{BradlynScience}.
These degeneracy points have been extensively studied and classified in the topological condensed matter literature~\cite{Alpin, encyclopedia}.
Introducing symmetry-breaking perturbations, these degeneracies will generally split into Weyl points, hence our results are perfectly applicable in this context.
The authors have already demonstrated similar methods applied to twofold degeneracies in crystalline band structures in previous work~\cite{BirthQ,naselli2024stability}.

Depending on the symmetry group of the crystal and the $\bf k$ point in question, various orders of degeneracy are possible.
Ref.~\cite{Alpin} extensively studied 4-fold degeneracy points protected by crystal symmetries.
While it was known that a $\bf k\cdot p$ Hamiltonian of the same form as the spin-$3/2$ Hamiltonian discussed in the previous section was allowed by symmetry, they found that a more generic class of linear Hamiltonians (parametrized by $\alpha =(\alpha_0, \alpha_1, \alpha_2) \in \R^3$) is also compatible with the symmetry:
\begin{equation}
\begin{aligned}
H_{(\alpha)}(\bf k) = & \alpha_0 \left[ 2k_x \sigma_x \tau_z + k_y \left( -\sqrt{3} \sigma_x \tau_0 - \sigma_y \tau_0 \right)  + k_z \left( \sigma_x \tau_x + \sqrt{3} \sigma_y \tau_x \right) \right] \\
& + \alpha_1 \left[ -2k_x \sigma_y \tau_z + k_y \left( -\sigma_x \tau_0 + \sqrt{3} \sigma_y \tau_0 \right) + k_z \left( \sqrt{3} \sigma_x \tau_x - \sigma_y \tau_x \right) \right] \\
& + 2\alpha_2 \left[ k_x \sigma_z \tau_x + k_y \sigma_0 \tau_y + k_z \sigma_z \tau_z \right]
\end{aligned}
\end{equation}
where $\sigma_i$ and $\tau_i$ are two sets of $2 \times 2$ Pauli matrices, and their products are taken in the tensor product sense, as customary in physics literature.
\begin{comment}
\begin{equation}
    H({\bf{k}}) = k_x H_x + k_y H_y + k_z H_z 
\end{equation}

\begin{eqnarray}
    H_x &=& \begin{pmatrix}0 & 2 \alpha_{2} & 2 \alpha_{0} - 2 i \alpha_{1} & 0\\2 \alpha_{2} & 0 & 0 & - 2 \alpha_{0} + 2 i \alpha_{1}\\2 \alpha_{0} + 2 i \alpha_{1} & 0 & 0 & - 2 \alpha_{2}\\0 & - 2 \alpha_{0} - 2 i \alpha_{1} & - 2 \alpha_{2} & 0\end{pmatrix} \\
    H_y &=& \begin{pmatrix}0 & - 2 i \alpha_{2} & \alpha_{0}( i - \sqrt{3}) - \alpha_{1} (1+ \sqrt{3} i) & 0\\
    2 i \alpha_{2} & 0 & 0 & - \sqrt{3} \alpha_{0} + i \alpha_{0} - \alpha_{1} - \sqrt{3} i \alpha_{1}\\
    - \sqrt{3} \alpha_{0} - i \alpha_{0} - \alpha_{1} + \sqrt{3} i \alpha_{1} & 0 & 0 & - 2 i \alpha_{2}\\
    0 & - \sqrt{3} \alpha_{0} - i \alpha_{0} - \alpha_{1} + \sqrt{3} i \alpha_{1} & 2 i \alpha_{2} & 0\end{pmatrix}\\
    H_z &=& \begin{pmatrix}2 \alpha_{2} & 0 & 0 & \alpha_{0} - \sqrt{3} i \alpha_{0} + \sqrt{3} \alpha_{1} + i \alpha_{1}\\0 & - 2 \alpha_{2} & \alpha_{0} - \sqrt{3} i \alpha_{0} + \sqrt{3} \alpha_{1} + i \alpha_{1} & 0\\0 & \alpha_{0} + \sqrt{3} i \alpha_{0} + \sqrt{3} \alpha_{1} - i \alpha_{1} & - 2 \alpha_{2} & 0\\\alpha_{0} + \sqrt{3} i \alpha_{0} + \sqrt{3} \alpha_{1} - i \alpha_{1} & 0 & 0 & 2 \alpha_{2}\end{pmatrix}
\end{eqnarray}
\end{comment}

This Hamiltonian has an isolated degeneracy point for real $\bf k$ unless $\alpha_2 = \pm \sqrt{\alpha_0^2 + \alpha_1^2}$, or $\alpha_2 = 0$, or $\sqrt{\alpha_0^2 + \alpha_1^2} = 0$.
These surfaces in the three-dimensional parameter space separate regions with different patterns of Chern numbers $\{c_1(\eta_j) \ | \ j=1, \dots, 4 \}$ for the four bands, calculated on a sphere surrounding $\bf k=0$ where they are non-degenerate.
Ref.~\cite{Alpin} finds Chern number patterns $\{-3, 5, -5, 3\}$, $\{-3, -1, 1, 3\}$, or the reverse order.
While the second one is identical to the pattern of a spin-$3/2$ Hamiltonian (see Lemma~\ref{le:chernspin}), the first one is different.
This also results in a different lower bound $\sharp_{\text{alg}} \mathbf{WP}$ on the number of Weyl points  $\sharp \mathbf{WP}$ born upon a perturbation, see Equation~\eqref{eq:lowerbound}: $\sharp_{\text{alg}} \mathbf{WP}$ is equal to $8$ in the first case, $10$ in the second.

To apply the upper bound $\sharp \mathbf{cWP}$ one should check that the degeneracy at 0 is also isolated in the complex sense (for non-exceptional parameter values $\alpha$). By Remark~\ref{re:compiso} this is the case if and only if the dimension $\dim_{\C}(\mathcal{O}^{\C}_4/J^{\C})=\dim_{\R}(\mathcal{O}^{\R}_4/J^{\R})$ appears in Corollary~\ref{co:pullback} is finite. We computed this dimension  for random values of the parameter $\alpha$ by computer, and we obtained a finite value, namely, $\sharp \mathbf{cWP}=20$ in these cases. We conjecture that the degeneracy is isolated in the complex sense if it is isolated in the real sense, that is, for every $\alpha_i$ triple in the interior of the four regions defined above. Under this hypothesis, $\sharp \mathbf{cWP}=20$ is an upper bound for $\sharp \mathbf{WP}$ in all regions of the parameter $\alpha$, according to Corollary~\ref{co:pyram} (indeed, $4^2(4^2-1)/12=20$).

While numerical calculations demonstrate that the lower bound is sharp, it is unclear whether the upper bound is sharp for real $\bf k$.
We leave detailed study of the maximum number of Weyl points created by perturbing this, and other relevant degeneracy points protected by crystallographic symmetries in solid-state band structures, to future work.

\subsection{Two-fold degeneracy points can be described without SW}\label{ss:twofold}

%\notegyuri{there was a 3x3 example in the original version\
%also i think the effective hamiltonian (and not just the algebra) of the 2-fold quasi-degenerate subspace of a 3x3 matrix around diag(1,0,0) can be given up to contact equivalence(?) with the determinant method and elimination of lambda}

Here we deduce the local algebra $Q_0(h)$ (up to isomorphism) of the effective germ $h$ of a 2-fold degeneracy point directcy from the Hamiltonian $H$, without SW. In a sense explained below, the isomorphism type of $Q_0(h)$ characterises the type of the two-fold degeneracy.

In case of a degeneracy point $p_0 \in M^3$ of a Hamiltonian $H: M^3 \to \text{Herm}(n)$ with a 2-fold degenerate eigenvalue $\lambda_0$ of $A_0=H(p_0)$, the traditional method of investigation is to derive the effective germ $h:(\R^3, 0) \to (\R^3, 0)$ from SW (see Section~\ref{ss:efftwofold}), which characterizes the degeneracy type of $H$ at $p_0$. Actually the `type of degeneracy' can be precisely defined as the equivalence class of the germ $h$ up to a suitable equivalence relation. The \emph{contact equivalence} \cite[Ch. 4]{MondBook} is a reasonable choice from several reasons, for example (1) the most important invariants, for example $\deg_0(h)$ (up to sign) and $\text{mult}_0(h)$, are preserved by contact equivalence, (2) the contact equivalence class of $h$ characterizes the \emph{contact type} of $H$ and $\Sigma$ at $A_0$ in sense of \cite{MontaldiContact}. Contact equivalence was used in~\cite{naselli2024stability} to study 2-fold degeneracy points. Although other type of equivalences were also introduced for the classification of 2-fold degeneracy points \cite{Teramoto2017,teramoto2020application}, here we aim for the contact eqiuvalence class of $h$. 

Consider two effective germs $h_1, h_2$ come from Hamiltonians with isolated degeneracy in complex sence, that is, $h_{1, \C}^{-1}(0)=\{0\}$ and $h_{2, \C}^{-1}(0)=\{0\}$ by Equation~\eqref{eq:hameff}. Then, by \cite[Theorem 4.4]{MondBook}, $h_1$ and $h_2$    are contact equivalent if and only if their local algebras $Q_0(h_1)$ and $Q_0(h_2)$ are isomorphic, as real algebras. This means that in this case the invariants of map germs preserved by contact equivalence are determined by the local algebra. For example, $\text{mult}_0(h)=\dim Q_0(h)$, and $\deg_0(h)$ (up to sign) can be computed by the Eisenbud--Levine algorithm \cite{Eisenbud1978,EisenbudLevineTeissier1977}
as the index of a suitable quadratic form on $Q_0(h)$.

Consider the Hamiltonian germ $H: (\R^3, 0) \to (\text{Herm}(n), A_0)$ (in local coordinates of $M^3$ around $p_0=0$), and its lift $\widetilde{H}(x, y,z, \lambda)=(H(x,y,z), \lambda+\lambda_0 )$ and ${H}'(x, y,z, \lambda)=H(x,y,z)- (\lambda +\lambda_0) \mathds{1}$. Consider the (real) algebra $\mathcal{O}_4/J$, where $J=I_{n-1}({H}')$ is the ideal in $\mathcal{O}_4$ generated by the $(n-1) \times (n-1)$ minors $M_{ij}(H-(\lambda_0 + \lambda \mathds{1}))$. (Here every algebra is real, but we omit the index $\R$ used earlier. Recall that  $\text{Herm}(n) \cong \R^{n^2}$.)

\begin{prop}
  $\mathcal{O}_4/J$ is isomorphic to $Q_0(h)$, as real algebras.
\end{prop}

\begin{proof}
    By Proposition~\ref{pr:hermcomp},  $(\widetilde{\Sigma}_{\text{herm}}, (A_0, \lambda_0))$ and $(\Sigma_{\text{herm}}, A_0;\lambda_0)$ are non-singular germs and the projection $p: \R^{n^2+1} \to \R^{n^2}$, $p(A, \lambda)=A$ is an analytic diffeomorphism between them. By SW, the vanishing ideal $I(\Sigma_{\text{herm}}, A_0;\lambda_0)$ of $(\Sigma_{\text{herm}}, A_0;\lambda_0) \subset (\R^{n^2}, A_0)$ is generated by the components of the traceless effective Hamiltonian $A_{\text{eff}}^{\text{tr}=0}$, hence its pull-back via $H$ is the ideal $I_{h}=I(h_1, h_2, h_3) \subset \mathcal{O}_3$. The pull-back of $I(\widetilde{\Sigma}_{\text{herm}}, (A_0, \lambda_0))$ via $\widetilde{H}$ is $J \subset \mathcal{O}_4$. Hence, we have to show that \begin{equation}    \frac{\mathcal{O}_3}{H^*(I(\Sigma_{\text{herm}}, A_0;\lambda_0)) \cdot \mathcal{O}_3} \cong \frac{\mathcal{O}_4}{\widetilde{H}^*(I(\widetilde{\Sigma}_{\text{herm}}, (A_0, \lambda_0))) \cdot \mathcal{O}_4}.
    \end{equation}
    
    Consider $\Sigma$ also in $\R^{n^2+1}$ via the inclusion $\R^{n^2} \subset \R^{n^2 +1}$ of the $\{ \lambda=0\} $ hyperplane. The vanishing ideal of $(\Sigma_{\text{herm}}, A_0;\lambda_0) \subset (\R^{n^2 +1}, (A_0, 0))$ is equal to $I(\Sigma_{\text{herm}}, A_0;\lambda_0)+(\lambda)$.
    We define the map germ $G: (\R^4, 0) \to (\R^{n^2+1}, (A_0, 0))$, $G(x,y,z, \lambda)=(H(x,y,z), \lambda)$. Obviously,
        \begin{equation}    \frac{\mathcal{O}_3}{H^*(I(\Sigma_{\text{herm}}, A_0;\lambda_0)) \cdot \mathcal{O}_3} \cong \frac{\mathcal{O}_4}{G^*(I(\Sigma_{\text{herm}}, A_0;\lambda_0)+(\lambda)) \cdot \mathcal{O}_4}.
    \end{equation}

        The local diffeomorphism $p$ between $(\widetilde{\Sigma}_{\text{herm}}, (A_0, \lambda_0))$ and $(\Sigma_{\text{herm}}, A_0;\lambda_0)$ can be extended to a local diffeomorphism $\Phi: (\R^{n^2 +1}, (A_0, \lambda_0)) \to (\R^{n^2 +1}, (A_0, 0))$ with the properties
        \begin{enumerate}
            \item $\Phi(\widetilde{\Sigma}_{\text{herm}}, (A_0, \lambda_0))=(\Sigma_{\text{herm}}, A_0;\lambda_0)$ (that is, $\Phi$ is an extension of $p$),
            \item $\Phi \circ \widetilde{H}=G$.
        \end{enumerate}
        Therefore,

\begin{equation}
        \frac{\mathcal{O}_4}{G^*(I(\Sigma_{\text{herm}}, A_0;\lambda_0)+(\lambda)) \cdot \mathcal{O}_4} \cong 
        \frac{\mathcal{O}_4}{\widetilde{H}^*(I(\widetilde{\Sigma}_{\text{herm}}, (A_0, \lambda_0))) \cdot \mathcal{O}_4},
        \end{equation}
        which proves the proposition. \end{proof}

In practice, this also means that $\lambda$ always can be eliminated, that is, $[\lambda]$ is expressed by the other variables in the quotient algebra, as the following example shows. Note that it shows the analogous argument for 2-parameter families of real symmetric matrices.

\begin{example} We illustrate on a real symmetric example that the algebra of the degeneracy can be derived without calculating the effective model. The method is the same as for the 3-fold degeneracy: after substracting $\lambda$
\bean
H(x,y)-\lambda\mathds{1}=\begin{pmatrix}
2-\lambda &x&y\\
x&-\lambda &0\\
y&0&-\lambda
\end{pmatrix},
\eean
we consider the ideal $J$ generated by the $2\times 2$ minors. The four residue classes 
$[1]$, $[x]$, $[y]$, $[x^2] = [y]^2 = -2[\lambda ]$ form a basis of the algebra $\mathcal{O}_3/J$. It can be seen that this is isomorphic to the local algebra $Q_0(h)$ of $h(x,y)=(x^2-y^2, xy)$ (called `quadratic Weyl point') with local multiplicity 4, local degree 2, see e.g. \cite{BirthQ, naselli2024stability}. 

The SW transformation can be performed exactly on this $3 \times 3$ Hamiltonian, see \cite[Example 3.5.7]{PinterSW} (where a slightly modified version is computed), which verifies that the corresponding effective germ is contact equivalent to $h(x,y)=(x^2-y^2, xy)$.
\end{example}

\section*{Author contributions}
Gy. Frank provided the initial project idea, computed the results for specific examples, and conjectured the general results.
The mathematical formalization was done by G. Pint\'{e}r and Gy.
Frank, with contributions from D. Varjas and A. P\'{a}lyi.
Gy. Frank, A. P\'{a}lyi, and D. Varjas developed the physical examples.
A. P\'{a}lyi and D. Varjas produced the figures.
The initial draft of the manuscript was written by G. Pint\'{e}r with the contribution of Gy. Frank.
All authors reviewed and edited the manuscript. The authors are listed in alphabetic order.

% You can rewrite it, I've just written my version with my own words.

%\noteandras{There is a standard taxonomy for author contributions, \url{https://credit.niso.org/}. Let's use that one.}

%Notes
% Gergo Pinter: conceptualisation, formal analysis, methodology, writing - original draft
% Gyorgy Frank: conceptualisation, formal analysis, methodology, writing - original draft
% Daniel Varjas: funding acquisition, visualisation, writing - ?
% Andras Palyi: funding acquisition, writing - review and editing.
%credit categories that seem to be relevant for this work: 
% conceptualisation
% formal analysis
% funding acquisition - Varjas, Palyi
% investigation
% methodology
% supervision
% visualisation
% writing - original draft - Pinter
% writing - review and editing - Palyi

\section*{Acknowledgments}
We thank M\'{a}ty\'{a}s Domokos, L\'{a}szló Feh\'{e}r, \'{A}kos Matszangosz, and David Mond for their useful feedback on the project, and especially Alexander Hof for many insightful technical discussions.
This work was supported by the HUN-REN Hungarian Research Network through the Supported Research Groups Programme
(HUN-REN-BME-BCE Quantum Technology Research Group, TKCS-2024/34).
D.V. received funding from the Deutsche Forschungsgemeinschaft (DFG, German Research Foundation) under Germany’s Excellence Strategy through the W\"{u}rzburg-Dresden Cluster of Excellence on Complexity and
Topology in Quantum Matter – ct.qmat (EXC 2147, project-ids 390858490 and 392019).
G.P. and D.V. acknowledge funding from the National Research, Development and Innovation Office of Hungary under OTKA grant no. FK 146499.

\section*{Data availability statement}

Our manuscript has no associated data.

\appendix
\section{Appendix: degree and multiplicity}\label{a:app}

For convenience we collect here the basic concepts of (affine) algebraic and local analytic geometry used throughout the paper, see Section~\ref{app:basic}. The main focus is on the degree and multiplicity of complex analytic set germs $(X, x_0)$, which is discussed in Section~\ref{app:multi} based on \cite[Appendix D.3.]{MondBook}. 

The main goal of this appendix if the discussion of the multiplicity of holomorphic map germs with respect to a complex analytic set germ in Section~\ref{ss:intsect}. Its main theorem is Proposition~\ref{pr:everypert}, which is the key concept of this paper, and for which we didn't find any direct reference in the literature, hence, we wrote a detailed proof for it.

Throughout the paper, the dimension of an algebra always means its dimension as a vector space over the base field, which is usually $\C$, sometimes $\R$. That is, we do not use the dimensions of rings (like Krull dimension).

\subsection{Basic concepts in algebraic and local analytic geometry}\label{app:basic}

We provide a very sketchy introduction to algebraic geometry here, without claiming to be complete. For affine and projective algebraic geometry we refer to the classical books e.g. \cite{HartshoneBook, EisenbudCommutativeAlgebra, Shafarevich}, for local analytic geometry we refer to \cite{Jong, MondBook}. For a very short introduction of algebras and ideals in physical context we refer to the appandix of \cite{BirthQ}.

In high school coordinate geometry, geometrical subsets (e.g. curves, surfaces) of $\mathbb{R}^n$ are defined by equations, i.e. as the common zero locus of some functions of $n$ variables. Let $V \subset \mathbb{R}^n$ be the common zero locus of the functions $f_1, \dots, f_k$, that is,
\bean
V=\bigcap_{i=1}^k f_i^{-1}(0)=f^{-1}(0),
\eean
where $f=(f_1, \dots, f_k): \mathbb{R}^n \to \mathbb{R}^k$. Which functions vanish on $V$? Beyond the linear combinations of $f_i$ (with real coefficients), the functions in form
\bean\label{eq:ideal}
g=\sum_{i=1}^k g_i f_i,
\eean
also vanish at the points $V$, where the coefficients $g_i: \mathbb{R}^n \to \mathbb{R}$ are functions. The functions $g$ in form \eqref{eq:ideal} are the elements of the ideal generated by $f_i$. Moreover, if, for example, $f_1(x, y)=(x+3y)^2$, then $\sqrt{f_1}=x+3y$ also vanishes on $V$. These observations motivate the setup of algebraic geometry.

Algebraic geometry is basicly discussed on 3 different levels, the space of functions under examination is chosen to fit to the level. Another variable input is the base field, which could be $\mathbb{R}$ or $\mathbb{C}$ in our cases, however the fundamental theorem (''Nullstellensatz'') requires to work over an algebraic closed field, which is $\mathbb{C}$ for us. The usual levels are the following:
\begin{enumerate}
    \item \emph{Affine algebraic geometry} in $\mathbb{C}^n$, the considered functions are the polynomials, $\mathbb{C}[x_1, x_2, \dots, x_n]$. 
    \item \emph{Projective algebraic geometry} in $\mathbb{P}^n:=\mathbb{C} P^n$, the functions are the homogeneous polynomials of $n+1$ variables.
    \item \emph{Local analytic geometry} in $\mathbb{C}^n$ around a point $p_0 \in \mathbb{C}^n$ for example, the origin $p_0=0$, the functions are the germs of analytic functions, that is, locally convergent power series, whose algebra is denoted by $\mathcal{O}_n:=\mathcal{O}_{\mathbb{C}^n, 0}=\mathbb{C}\{x_1, \dots, x_n\}$.
\end{enumerate}

Our investigation of degeneracy points is local, hence we need local analytic geometry, but we also use affine and projective algebraic geometry. 
The fundamental ingredients and theorems are the same in every level, but here we summarize them in the local analytic case, with notes on the affine case. 

The basic algebraic objects are the ideals $I \subset \mathcal{O}_n$, which are in a kind of duality with the local subsets $V \subset \mathbb{C}^n$. More precisely, instead of sets we have to consider \emph{set germs}: two subsets $V$ and $W$ define the same set germ at 0 if their intersections with a suitably small neighborhood of the origin are the same. The set germ of $V$ at 0 is denoted by $(V, 0) \subset (\C^n, 0)$. The set germ $(V, p)$ of a subset $V$ at any $p \in V$ is defined in the same way. In the definitions below, $V=(V,0)$ or $(V,P)$, we omit base point from the notation.

Every ideal $I \subset \mathcal{O}_n$ defines an \emph{analytic set germ} $(V(I), 0)$, the vanishing locus\footnote{To be precise, in the local analytic case one should take representatives for these definitions, see \cite{Jong}.} of $I$:
\bean
V(I)=\{p \in \mathbb{C}^n \ | \ f(p)=0 \mbox{ for all } f \in I  \}.
\eean

The vanishing ideal $I(V)$ of a set germ $V$ is
\bean
I(V)=\{ f \in \mathcal{O}_n \ | \ f(p)=0 \mbox{ for all } p \in V \}.
\eean
It is clear that $I(V)$ is always an ideal of $\mathcal{O}_n$, and it has the following properties:
\begin{enumerate}
    \item If $I_1 \subset I_2$, then $V(I_1) \supset V(I_2)$.
    \item If $V_1 \subset V_2$, then $I(V_1) \supset I(V_2)$.
    \item $I(V(I)) \supset I$.
    \item $V(I(V)) \supset V$.\footnote{The notation $\subset$  allows equality as well.}
\end{enumerate}

The fundamental theorem of algebraic geometry is called \emph{Nullstellensatz}.\footnote{The affine version is called Hilbert's Nullstellensatz, its local analytic analogue is the Rückert Nullstellensatz.} It states that over an algebraic closed field (like $\C$ in our case)
\bean
I(V(I)) =\sqrt{I}
\eean
holds, where $\sqrt{I}$ is the \emph{radical} of $I$ defined as
\bean
\sqrt{I}=\{f \in \mathcal{O}_n \ | \ \exists b \in \mathbb{N} \ f^b \in I \}.
\eean
It is clear that $\sqrt{I} \supset I$, and $I(V(I)) \supset \sqrt{I}$. At first sight we cannot see that $\sqrt{I}$ is really an ideal of $\mathcal{O}_n$, and of course the hard direction of the theorem, $I(V(I)) \subset \sqrt{I}$ is not trivial.

How does this whole setup relate to the high-school coordinate geometry? Via the Noetherian property of the ring $\mathcal{O}_n$! That is, every ideal $I $ of $\mathcal{O}_n$ is generated by finitely many elements, $f_1, \dots, f_k$, hence, $V(I)$ is the common zero locus of the germs $f_i$.\footnote{Note that the Noetherian property holds in (real or complex) analytic category, but it does not hold in real $\mathcal{C}^{\infty}$ category: the algebra of $\mathcal{C}^{\infty}$ real germs is not Noetherian.}  Therefore, the two concepcts `common zero locus of some functions' and `vanishing locus of an ideal' are the same.

An ideal $I$ with $\sqrt{I}=I$ is called \emph{radical ideal}. The radical of an ideal is a radical ideal, that is, $\sqrt{\sqrt{I}}=\sqrt{I}$. For a radical ideal $I(V(I))=I$ holds by Nullsetellensatz. Possibly many different ideals define the same analytic set germ $(V, 0)$, but the vanishing ideal $I(V)$ is always a radical ideal and it is unique. $I(V)$ is called the \emph{reduced analytic structure} of $(V, 0)$. An ideal $I$ with $V(I)=V$ (or equivalently, $\sqrt{I}=I(V)$) is  called an \emph{analytic structure} on $(V,0)$. An analytic set germ has many different analytic structures, but only one reduced analytic structure, in other words, there is a bijection between analytic set germs and radical ideals of $\mathcal{O}_n$ via Nullstellensatz.

In contrast of analytic set germs, the notion of space germ takes the analytic structure into consideration as well. A \emph{(complex analytic) space germ} is an analytic set germ $V=V(I)$ endowed with the quotient algebra
    \bean
    \mathcal{O}_{V, 0} := \frac{\mathcal{O}_n}{I}.
    \eean

%According to this argument we distinguish two notions:

%A set germ $V$ which is $V=V(I)$ with an ideal $I$ is called \emph{local analytic set germ}. 
%Note that in affine and projective algebraic geometries the notion \emph{variety} is used for the zero locus of an ideal, however I do not know, whether is local (analytic) variety is used as a synonym of analytic set. 
%\textcolor{blue}{CLARIFY THE TERMINOLOGY -- suggestion: (complex) space germ everywhere}

%\begin{enumerate}
%    \item \emph{Local analytic space germ:} the zero locus $V=V(I)$ of an ideal of $\mathcal{O}_n$. It can be endowed with many different analytic structures. Nullstellensatz says that there is a bijection between local analytic set germs and radical ideals of $\mathcal{O}_n$.
%    \item \emph{Complex space germ:} an analytic set germ $V=V(I)$ endowed with the quotient algebra
%    \bean
 %   \mathcal{O}_{V, 0} := \frac{\mathcal{O}_n}{I}.
%    \eean
%    Nullstellensatz says that there is a bijection between complex analytic space germs and ideals of $\mathcal{O}_n$.
%\end{enumerate}

For an ideal $I$, the quotient $\mathcal{O}_n/I$ is called the \emph{local algebra} of the space germ $(V(I), \mathcal{O}_n/I)$. Intuitively, $\mathcal{O}_n/I$ 
is interpreted as \emph{the algebra of germs of functions on $V(I)$} in sense the if $f_1-f_2 \in I$ holds for two function germs $f_1, f_2 \in \mathcal{O}_n$, then they agree on $V(I)$. However, the other direction is not true, that is, $f_1|_{V(I)} =f_2|_{V(I)}$ does not imply $f_1-f_2 \in I$ in general. This is because the function algebra $\mathcal{O}_n/I$ is associated to the analytic structure defined by $I$, and not to the analytic set germ $V(I)$. 

\begin{ex}[$x^2$ from \cite{ArnoldSing}] For a fixed $t$, consider the function $f_t(x)=x^2-t \in \C[x]$. In the affine space $\C$, the vanishing locus of $f_t(x)$ consists of two points $x=\pm \sqrt{t}$, if $t \neq 0$. However, for $t=0$, the zero locus of $x^2$ is only one point, the origin. Denoting the ideal of $\mathcal{O}_1$ generated by $x^2$ by $I=(x^2) $, we conclude that the corresponding analytic set germ is $V(I)=\{0\}=V(x)$, and its vanishing ideal $I(V(I))=(x)$, where $(x)$ is the ideal generated by $x$. Indeed, $(x)=\sqrt{(x^2)}$. Hence, as an analytic set germ, $V(I)$ is only one point. We loose the information that this point is created by merging the two roots $x=\pm \sqrt{t}$ together, as $t$ tends to 0.

However, this information is preserved by the local algebra $\mathcal{O}_1/I$. Namely, the dimension of $\mathcal{O}_1/I$ is 2, as a complex vector space. A basis is formed by the residue classes $[1]=1+I$ and $[x]=x+I$. 

The space of functions defined on the two points $x=\pm \sqrt{t}$ is also two dimensional, since a function is defined by two independent complex values. As a change of basis in the function space, these functions can be identified with the linear functions $ax+b$. This property is inherited to the $t=0$ case, showing that $x^2=0$ determines a `fat point' at 0 with multiplicity $2=\dim(\mathcal{O}_1/I)$. As a space germ, we think on $V(I)$ as this fat point.
    
\end{ex}

\begin{ex}[One point spaces]\label{ex:onepoint}
 Let $V=\{0\} \subset \mathbb{C}^n$, that is, only one point, the origin. It seems very trivial, in fact, 
\bean
V=\{x_1=x_2= \dots =x_n=0 \}.
\eean
The ideal generated by the functions $x_i$ is the unique maximal ideal $\mathfrak{m}_n$ of $\mathcal{O}_n$, therefore $V=V(\mathfrak{m}_n)$. It is the reduced structure of $V$, indeed, $\sqrt{\mathfrak{m}_n}=\mathfrak{m}_n$, since $\mathfrak{m}_n$ is maximal. 

What are the other analytic structures of $V$? By Nullstelelnsatz, an ideal $I$ defines $V(I)=V$ if and only  if $\sqrt{I}=\mathfrak{m}_n$. 

For a  map germ $f=(f_1, \dots, f_k): (\mathbb{C}^n, 0) \to (\mathbb{C}^k, 0)$ and for the ideal $I_f=I(f_1, \dots, f_k)$  generated by the components $f_i$ of $f$, the analytic set germ
\bean
f^{-1}(0)=V(I_f)
\eean
consists of only one point (the origin) if and only if the (complex) dimension of the local algebra
\bean
Q(f)= \frac{\mathcal{O}_n}{I_f}
\eean
is finite \cite[Theorem D.5]{MondBook}. Compering it with the Nullstellensatz, we conclude that for an ideal $I$ of $ \mathcal{O}_n$ the following are equivalent:
\bean
V(I)=\{0\} \ \Leftrightarrow \ \sqrt{I}=\mathfrak{m}_n \ \Leftrightarrow \ \dim_{\mathbb{C}} \frac{\mathcal{O}_n}{I} < \infty.
\eean

For the investigation of isolated degeneracy points we actually study different analytic structures of the origin!

Cf. Theorem D.5. in [Mond--BAllesteros].
\end{ex}

For a positive dimensional example see Example~\ref{ex:cusp}. In many important cases, the generators $f_1, \dots, f_k$ of the ideal $I$ are polynomials, therefore, they also generate an ideal $I_{\text{aff}}$ in the ring of polynomials $ \C[x_1, \dots, x_n]$, which defines an affine algebraic subset (affine variety) $X:=V(I_{\text{aff}}) \subset \C^n$.  Then, it defines an analytic space germ $(X, p)$ at each point $p \in X$ (the \emph{analytification} of $X$ at $p$) with ideal generated by $I$ in the local ring $\mathcal{O}_{\C^n, p}$. It can be translated to $\mathcal{O}_n=\mathcal{O}_{\C^n, 0}$ by taking the polynomials $f(x+p)$ .

\subsection{Multiplicities}\label{app:multi} There are several similar notions of degrees and multiplicities, their relations are clarified in \cite[Appendix D.3.]{MondBook}. We present here a short summary of this.

The most important notion is \cite[Definition D.3.]{MondBook}, the multiplicity of an analytic set germ $(X, x_0) \subset (\mathbb{C}^n, x_0)$ of dimension $d$. Shortly, it is the intersection multiplicity of $(X, x_0)$ with a generic $(n-d)$-dimensional affine linear subspace $(L, x_0) \subset \mathbb{C}^n$ at $x_0$. It is equal to the number of generic intersection points $X \cap L'$ of $X$ and a slightly shifted parallel copy $L'$ of $L$.

\cite[Definition D.3.]{MondBook}   uses a slightly different approach to defined $\text{mult}(X, x_0)$ based on the previously defined notions. 
Consider a \emph{finite, surjective} map germ 
\bean
f: (X, x_0) \to (Y, y_0),
\eean
where $(X, x_0)$ and $(Y, y_0)$ are irreducible analytic set germs of the same dimension $d$. The \emph{degree} of $f$ is defined as
\bean
\deg (f):=\sharp f^{-1}(y),
\eean
where $y \in Y$ is any generic value of a sufficiently small representative of the germ $f$. By \cite[Lemma D.2]{MondBook}, the degree does not depent on the choice of the generic value $y$, it depends only on the map germ $f$.

The degree $\deg (f)$ has several equivalent algebraic characterizations based of the induced homomorphism 
\bean
f^*: \mathcal{O}_{Y, y_0} \to \mathcal{O}_{X, x_0}
\eean
between the local algebras, see \cite[Lemma D.2.--Example D.3.]{MondBook}.

Consider finite, surjective map germs in the special case $Y=\C^d$, that is,
\bean
f: (X, x_0) \to (\mathbb{C}^d, 0),
\eean
where $(X, x_0)$ is irreducible. In this case the \emph{multiplicity} of $f$ is defined as the dimension of the quotient algebra:
\bean\label{eq:multideg}
\mbox{mult} (f):= \dim
\frac{\mathcal{O}_{X, x_0}}{f^* \mathfrak{m}_d \cdot \mathcal{O}_{X, x_0}}= \dim
\frac{\mathcal{O}_{X, x_0}}{I(f_1, f_2, \dots, f_d)}.
\eean

In general, for map germs $f: (X, x_0) \to (\mathbb{C}^d, 0)$ the multiplicity is not equal to the degree. By \cite[Corollary D.6]{MondBook}, we have inequality  \bean\label{eq:degmult}
    \deg(f) \leq \mbox{mult}(f),
    \eean
     with equality if and only if $(X, x_0)$ is \emph{Cohen–Macaulay}. The algebraic definition of Cohen–Macaulay property can be found in \cite[Appendix C.3.6.]{MondBook}, but in this paper we use the above equivalent characterization. Intuitively, $(X, x_0)$ is Cohen--Macaulay if it behaves nicely under perturbation of map germs in sense of Proposition~\ref{pr:everypert}, see also the onservation of ultiplicity \cite[Corollary E.6]{MondBook}. Furthermoore, we use the following facts about Cohen--Macaulay property:  
\begin{enumerate}
\item If $(X, x_0)=(\mathbb{C}^d, 0)$, then it is Cohen--Macaulay \cite[Exercise C.3.2]{MondBook}. Therefore, the degree of a map germ $f:(\mathbb{C}^d, 0) \to (\mathbb{C}^d, 0)$ is equal to its multiplicity, which is the main point in \cite{BirthQ}.
   \item  Recall that $(X, x_0)$ is a \emph{complete intersection} if its vanishing ideal $I(X, x_0) \subset \mathcal{O}_{\C^n, x_0}$ can be generated by $ \text{codim}(X \subset \C^n)=n-d$ elements. Cohen--Macaulay property is more general than complete intersection, that is, every complete intersection is Cohen--Macaulay \cite[Proposition C.11]{MondBook}.
   \item Cohen--Macaulay is \emph{pure-dimensional} (equi-dimensional), that is, there is no embedded component and each component has the same (co)dimension \cite[Proposition C.8]{MondBook}.
  \item Determinantal varieties are Cohen--Macaulay, see \cite{BrunsVetterBook} and Section~\ref{ss:matrvar}.
%  \item An affine algebraic variety $X \subset \C^n$ is Cohen--Macaulay if and only if every analytic set germ $(X, x_0)$ is Cohen--Macaulay. See Section~\ref{ss:proof-notcohmac} for application.
  \item   Being Cohen--Macaulay is an `open property', that is, a Cohen--Macaulay analytic set germ has a representative whose germ at every point is Cohen--Macaulay \cite[Prop. 18.8.]{EisenbudCommutativeAlgebra}. See Section~\ref{ss:proof-notcohmac} for application.
\end{enumerate}

The \emph{multiplicity} of a dimension $d$ analytic set germ $(X, x_0) \subset (\mathbb{C}^n, x_0)$ is defined as follows. Consider a \emph{generic} affine linear (i.e. shifted linear) projection 
\bean
\widetilde{\pi}: (\mathbb{C}^n, x_0) \to (\mathbb{C}^d, y_0), 
\eean
its restriction 
\bean
\pi:=\widetilde{\pi}|_{(X, x_0)}: (X, x_0) \to (\mathbb{C}^d, y_0)
\eean
is finite and surjective. Then we define
\bean\label{eq:multidef}
\mbox{mult}(X, x_0):=\deg(\pi).
\eean
The affine subspeces $L$ and $L'$ mentioned above are $L=\pi^{-1}(y_0)$ and $L'=\pi^{-1}(y)$ with a generic value $y \in Y$ (close to $y_0$). 

Put the ingredients together in the simplest case! Assume that $(X, 0) \subset (\mathbb{C}^n, 0)$ is an irreducible Cohen--Macaulay analytic set germ of dimension $d$, $\widetilde{\pi}: (\mathbb{C}^n, x_0) \to (\mathbb{C}^d, y_0)$ is a generic linear projection whose restriction $\pi: (X, 0) \to (\mathbb{C}^d, 0)$ is a finite surjective map. Then we have
\bean
\mbox{mult}(X, 0)= \deg( \pi) = \mbox{mult}(\pi).
\eean

Furthermore, let $I(X) \subset \mathcal{O}_n$ be the (reduced) ideal of $(X, 0)$, then $\mathcal{O}_{X, x_0}=\mathcal{O}_n/I(X)$. Let  $I_{\widetilde{\pi}} \subset \mathcal{O}_n$ be the ideal generated by the components $\widetilde{\pi}_j$ of $\widetilde{\pi}$ (which are germs of linear functions $\C^n \to \C$). 
Let $I_{\pi} \subset \mathcal{O}_{X, x_0}$ be the ideal generated by the components $\pi_j=\widetilde{\pi}_j|_X$ of $\pi$. Then we have
\bean\label{eq:multi}
\mbox{mult}(\pi)=\frac{\mathcal{O}_{X, 0}}{I_{\pi}}= \dim
\frac{\frac{\mathcal{O}_n}{I(X)}}{I_{\pi}}= \dim
\frac{\mathcal{O}_n}{I(X)+I_{\widetilde{\pi}}}, 
\eean
 where $I(X)+I_{\widetilde{\pi}}$ is the ideal generated by $I(X)$ and $I_{\widetilde{\pi}}$.
 
 Assume that $I(X)$ is generated by the elements $g_i \in \mathcal{O}_n$. Note that their number is not necessarily $n-d$, since $(X, 0)$ is not necessarily a complete intersection. %Let $\pi_j$ ($j=1, \dots, d$) be the components of $\widetilde{\pi}$. 
 Then 
 \bean\label{eq:multikulti}
 \mbox{mult}(X, 0)=\mbox{mult}(\pi)=\dim \frac{\mathcal{O}_n}{I(X)+I_{\pi}}=
 \dim \frac{\mathcal{O}_n}{I(g_i, \pi_j)}
 \eean 
 With a special choice of coordinates in $\C^n$ it can be reached that
 \bean
 \pi(x_1, \dots, x_n)=(x_1, \dots, x_d).
 \eean

\begin{ex}[Cusp]\label{ex:cusp}
    Consider the affine algebraic subset $X=\{x^2-y^3=0\} \subset \C^2$, and its germ $(X,0)$ at 0. Consider the projections $\widetilde{\pi}: (\C^2, 0) \to (\C, 0)$, $\widetilde{\pi}(x,y)=ax+by$ (with $(0,0) \neq (a,b) \in \C^2$) and the restrictions $\pi=\widetilde{\pi}|_X$. Since $\dim(X,0)=1$ equals to the number of generators of the vanishing ideal $I(x^2-y^3)\subset \mathcal{O}_2$, $(X,0)$ is a complete intersection, hence, Cohen--Macaulay. 

    First we compute the multiplicity according to Equation~\eqref{eq:multikulti}:
    \begin{equation}
        \text{mult}(X,0)=\dim \frac{\mathcal{O}_2}{I(x^2-y^3, ax+by)}.
    \end{equation}
    
    If $b\neq 0$, then $[y]=-\frac{a}{b}[x]$ holds in the quotient algebra, hence, $[x^2-y^3]=[x^2+\frac{a^3}{b^3}x^3]$. The ideal generated by $x^2+\frac{a^3}{b^3}x^3$ in $\mathcal{O}_1$ consists of all power series of order at least 2. Indeed, $x^2+\frac{a^3}{b^3}x^3=x^2(1+\frac{a^3}{b^3}x)$, and $1+\frac{a^3}{b^3}x$ is invertible in $\mathcal{O}_1$ (see e.g. the appendix of \cite{BirthQ}).
    
    Hence, a basis of the algebra is formed by the residue classes $[1], [x]$, the dimension is 2, that is, $\text{mult}({\pi})=2$. Since this is the generic case, $\text{mult}(X,0)=\text{mult}({\pi})=2$.

    In the exceptional case $b=0$, $[x]=0$ in the quotient, hence, $[x^2-y^3]=-[y^3]$. A basis is formed by $[1], [y], [y^2]$, hence, by equation~\eqref{eq:multideg}, $\text{mult}({\pi})=3$ in this case. 

    Then we compute the degree $\deg(\pi)$ by definition, which agrees with $\text{mult}(\pi)$, since $(X,0)$ is Cohen--Macaulay. 
    First of all, consider the line $L=\{ax+by=0\} \subset \C^2$, and determine $(\pi^{-1}(0),0)=L \cap (X,0)$. This is the set of solutions of the system of equations
    \begin{equation}
       \left. \begin{array}{ccc}
       ax+by & = & 0 \\
       x^2-y^3 & =& 0 \\
       \end{array}
       \right\}
    \end{equation}
    
    For $a=0$ or $b=0$ there is only one solution $(x,y)=(0,0)$. If none of them is not 0, then we obtain one more solution $(x,y)=(-b^3/a^3, b^2/a^2)$, hence, $L \cap X$ consists of these two points (as sets). However, with a choice of a sufficiently small representative of $(X,0)$ this extra solution can be avoided. Hence, the intersection is $L \cap (X,0)=\{0\}$ in sense of set germs.  In any case, we fix the small representative of $(X,0)$ before perturbing $L$.

    Then consider a point $0 \neq t \in \C$, $|t| \ll 1$. Its preimages $\pi^{-1}(t)$ in $X$ are the intersections of $X$ with the line $L'=\{ax+by=t\} \subset \C^2$ for small $t \neq 0$. Keep in mind that  we want to count only the intersection points born from 0, i.e. those which go to 0 as $t$ tends to 0. These are the intersection points of $L'$ with the fixed representative $(X,0)$.
    
    These preimages (intersection points) can be obtained as the solutions of
    \begin{equation}
       \left. \begin{array}{ccc}
       ax+by & = & t \\
       x^2-y^3 & =& 0 \\
       \end{array}
       \right\}
    \end{equation}
   
    In the expected case $b=0$, 
    \begin{equation}
        x=\frac{1}{a} t, \ y=\rho \sqrt[3]{\frac{1}{a^2}t^2},
    \end{equation}
    with the three third roots of unity $\rho$ we get 3 solutions born from 0.

    In the general case, first consider $a=0$. Then,
    \begin{equation}
        x=\pm \sqrt{\frac{1}{b^3}t^2}, \ y=\frac{1}{b} t
    \end{equation}
    gives 2 solutions born from 0.

    If $a \neq 0$ and $b \neq 0$, substitution leads to the equation
\begin{equation}
    x^2-\frac{1}{b^2}(t-ax)^3=0.
\end{equation}
This has 3 solutions, but one of them  converges to $(x,y)=(-b^3/a^3, b^2/a^2)$ as $t$ tends to 0. Hence, we have 2 solutions born from 0. In other words, $L' \cap X$ consists of 3 points in this case, but $L' \cap (X,0)$ has 2 points.

Note that this third solution $(x,y)=(-b^3/a^3, b^2/a^2)$ of $L \cap X$ goes to the infinity (in the projective plane $\C P^2$) as $a$ tends to 0, and if $b$ tends to 0 it merges with the origin. This explains the increased multiplicity for $b=0$.
\end{ex}

For non Cohen--Macaulay examples see e.g. \cite[Exercise C.3.2]{MondBook}.

\subsection{Intersection of map germs and analytic set germs}\label{ss:intsect}

%\notegergo{Alex, Gyuri, compare with [MondGor]. Ss. 3.5. Also, [Mond--Ballesteros C.4.1]}

We introduce a terminology \emph{multiplicity of holomorphic map germs
with respect to a complex analytic set germ} serving as a common generalization of the degree of finite map germs $f:(X, x_0) \to (Y, y_0)$ between analytic set germs of the same dimension $d$, and the multiplicity of analytic set germs $(X,x_0) \subset (\C^n, x_0)$, and we clarify its basic properties in Propostion~\ref{pr:everypert}. This is the key point of the discussion in \ref{s:proofs}. We didn't find explicit reference in the literature, although \cite[C.5]{MondBook}, in particular, Proposition C.12, and Corollary E.6. are similar. The proof of Propostion~\ref{pr:everypert} is due to Alex Hof.

Let $(X, x_0) \subset (\C^n, x_0) $ be an irreducible analytic germ of dimension $d$. Let $f: (\mathbb{C}^{n-d}, 0) \to (\mathbb{C}^n, x_0)$ be a map germ. We say that $f$ is \emph{isolated with respect to $(X, x_0)$} if $f^{-1}(X, x_0)=\{0\}$ (as a set, not a space germ) holds for a sufficiently small representative of $f$ and $(X,x_0)$. Given $f$ and $(X,x_0)$, we fix such representatives, which we do not distinguish from the germ in notation.  We call a map germ $\widetilde{\pi}: (\C^n, x_0) \to (\C^{d}, 0)$  \emph{finite with respect to $(X,x_0)$} if $\pi:=\widetilde{\pi}|_{(X,x_0)}X: (X, x_0) \to (\C^d, 0)$ is finite, that is, $\pi^{-1}(0)=\{x_0\}$. See equivalent characterizations of finite map germs in \cite[Theorem D.5]{MondBook}.  

If $f$ is an immersion, then there is a (not unique) corresponding germ of submersion $\widetilde{\pi}: (\C^n, x_0) \to (\C^{d}, 0)$ such that $(\widetilde{\pi}^{-1}(0), x_0) \subset (\C^n, x_0)$ is the image of $f$. Obviously, in this case, $f$ is isolated with respect to $(X, x_0)$ if and only if $\widetilde{\pi}$ is finite with respect to $(X,x_0)$.

We define a one parameter perturbation $f_t: \C^{n-d} \to \C^n$ of $f$ as follows. Starting from an unfolding of $f$ in form 
\begin{equation}
F:(\C^{n-d} \times \C, (0,0)) \to (\C^n \times \C, (x_0,0)), \ F(x, t)=(f_t(x), t), \ f_{t=0}=f
\end{equation}
and we take a sufficiently small representative of $F$ defined on the domain $\mathcal{U}_x \times \mathcal{U}_t \subset \C^{n-d} \times \C$, cf. \cite[Sec. 5.5]{MondBook} and \cite[Appendix]{BirthQ}. For a fixed parameter value $t \in \mathcal{U}_t$, we call the map $f_t(x): \mathcal{U}_x \to \C^n$ the perturbation of $f$ corresponding to $t$. Taking the fixed (or smaller) representative $(X, x_0) \subset \mathcal{U}_x$, the elements of the preimage $f_t^{-1}(X, x_0) \subset \mathcal{U}_x$ are thought as the intersection points of the perturbation $f_t$ and $X$ born from $x_0$. 

The perturbation is generic \emph{generic (transverse) with respect to $(X, x_0)$}, if, with a sufficiently small choice of the representative of $F$, for every $0 \neq t \in \mathcal{U}_t$ and every point $p \in f_t^{-1}(X, x_0)$, $(X,x_0)$ is non-singular at $f_t(p)$, and $f_t$ is transverse to (the non-singular part of) $X$. 

\begin{rem} By the (stratified) transversality theorem and the fact that discriminants are complex analytic subsets of positive codimension, it seems that generic perturbations (with respect to a suitable topology) are transverse in the above sense. However, it is not obvious for the authors, see e.g. point 1. of Remark 4.4.4. in \cite{PinterSW}. In this paper, generic perturbation means transverse in the above sense, and we do not use that these perturbations are really generic.
\end{rem}

\begin{prop}\label{pr:everypert} Let $(X, 0) \subset (\C^n, 0) $ be an irreducible analytic set germ of dimension $d$. Assume that $(X,x_0)$ is Cohen--Macaulay. Let $f: (\mathbb{C}^{n-d}, 0) \to (\mathbb{C}^n, x_0)$ be a holomorphic map germ, and assume that $f$ is isolated with respect to $(X, x_0)$.
    \begin{enumerate}[label=(\alph*)]
        \item  For every  one parameter perturbation $f_t$ of $f$ which is generic with respect to $(X, x_0)$ and for every $t \neq 0$, $f_t$ has the same number of intersection points, that is, $\sharp f_t^{-1}(X, x_0)$ does not depend on the choice of the generic perturbation. Namely,
        \begin{equation}\label{eq:preimages}
            \sharp f_t^{-1}(X, x_0) = \dim \frac{\mathcal{O}_{n-d}}{f^*(I(X, x_0)) \cdot \mathcal{O}_{n-d}}.
        \end{equation}
        \item If, furthermore, $f$ is a germ of immersion, then  
        \begin{equation}
        \sharp f_t^{-1}(X, x_0) = \dim \frac{\mathcal{O}_{n-d}}{f^*(I(X, x_0)) \cdot \mathcal{O}_{n-d}}=
        \dim \frac{\mathcal{O}_{\C^n, x_0}}{I(X, x_0)+I_{\widetilde{\pi}}}=\text{mult} (\pi), \end{equation}
        where $\widetilde{\pi}: (\C^n, x_0) \to (\C^{d}, 0)$ is a germ of submersion corresponding to $f$, that is, $(\widetilde{\pi}^{-1}(0), x_0) \subset (\C^n, x_0)$ is the image of $f$, and $I_{\widetilde{\pi}}=\widetilde{\pi}^* (\mathfrak{m}_d) \cdot \mathcal{O}_{\C^n, x_0}$ is the ideal generated by the components of $\widetilde{\pi}$, and $\pi=\widetilde{\pi}|_X$.
    \end{enumerate}
    (Recall that $\dim=\dim_{\C}$ refers to the complex vector space dimension over $\C$ everywhere.)
\end{prop}

\begin{proof}
For simplicity we assume that $x_0=0$, the general case is the similar. Moreover, sometimes we omit the base point from the notation. First we prove part (a) and (b) under the additional condition that $f$ is an immersion. Then $F:(\C^{n-d} \times \C, 0) \to (\C^n \times \C, 0)$, $F(x, t)=(f_t(x), t)$ is an immersion as well. For a germ of submersion $\widetilde{\pi}: (\C^{n}, 0) \to (\C^d, 0) $ corresponding to $f$ there is an induced unfolding (and perturbation) 
\begin{eqnarray}
    \widetilde{\Pi}: (\C^n \times \C, 0) &\to& (\C^d \times \C, 0)  \\
    \widetilde{\Pi}(x, t) &=&(\widetilde{\pi}_t(x), t).
\end{eqnarray}
%$\widetilde{\Pi}: (\C^n \times \C, 0) \to (\C^d \times \C, 0) $, $\widetilde{\Pi}(x, t)=(\widetilde{\pi}_t(x), t)$
with the following properties: 
\begin{enumerate}
    \item $\Pi:=\widetilde{\Pi}|_{X \times \C}$ is finite, that is, $\Pi^{-1}(0,0)=\pi^{-1}(0)=0$, where $\pi=\pi_{t=0}=\widetilde{\pi}|_X$.
    \item $\widetilde{\pi}_t^{-1}(0)$ is the image of $f_t$, hence, $\widetilde{\Pi}^{-1}(\{0\} \times \C) \subset \C^n \times \C$ is the image of $F$. 
    \item The space germ $Z:=\Pi^{-1}(\{0\} \times \C) \subset X \times \C
    $ has dimension 1. $Z$ is the intersection of $X \times \C$ with the image of $F$.
    \item The space germ $W:=F^{-1}(X \times \C)=F^{-1}(Z) \subset \C^{n-d} \times \C$ has dimension 1. 
    \item $F|_W: W \to Z$ is an isomorphism of space germs (at $(0,0)$), because $F$ is an immersion, that is, an isomorphism between the non-singular ambient space germs $\C^{n-d} \times \C $ and $F(\C^{n-d}  \times \C)$ (at $(0,0)$).
    \item $Z \subset X \times \C$ is a complete intersection. Indeed, $Z$ is the vanishing locus of the ideal in $\mathcal{O}_{X \times \C, 0}$ generated by the first $d$ components of $\Pi$ (that is, the components of $\pi_t$), and their number $d$ equals to the  codimension of $Z$ in $ X \times \C$.
    \item $Z $ is Cohen--Macaulay. Indeed, $X$ is Cohen--Macaulay, hence $X \times \C$ is also Cohen--Macaulay, and  $Z$ is a coomplete intersection, therefore, Cohen--Macaulay space germ in $X \times \C$. 
\end{enumerate}

Consider the projection to the second coordinate $\sigma: (Z, (0,0)) \to (\C, 0)$, $\sigma(x,t)=t$. $\sigma$ is finite, that is, $\sigma^{-1}(0)=\{0\}$. Observe that 
\begin{equation}
     \sigma^{-1}(t)= \pi_t^{-1}(0),
\end{equation}
hence, 
\begin{equation}
     \sharp \sigma^{-1}(t)= \sharp \pi_t^{-1}(0)=\sharp f_t^{-1}(X).
\end{equation}

Since $Z$ is Cohen--Macaulay, 
%$\sigma$ is \emph{flat}, which implies in this particular situation that the cardinality $\sharp \sigma^{-1}(t)$ of $\sigma^{-1}(t)$ is the same for all $t \neq 0$, and it is equal to $\text{mult}(\sigma)$ (at $(0,0)$).
$\deg(\sigma)=\text{mult}(\sigma)$. Since every value $t \neq 0$ is a regular value of $\sigma$, $\deg (\sigma)=\sharp \sigma^{-1}(t)$ for $t \neq 0$. The multiplicity is
\begin{equation}
    \text{mult}(\sigma) =\dim \frac{\mathcal{O}_{Z, (0,0)}}{I_{\sigma}}=\dim \frac{\mathcal{O}_{n+1}}{I(X)+I(\widetilde{\Pi}_1, \dots, \widetilde{\Pi}_d)+ I(t)}=\dim \frac {\mathcal{O}_{n}}{I(X)+I_{\widetilde{\pi}}},
\end{equation}
On the other hand, $\text{mult}(\sigma)=
    \text{mult}(\sigma \circ F|_W)$, since $F|_{W}$ is an isomorphism. Therefore,
\begin{equation}
    \text{mult}(\sigma)=
    \text{mult}(\sigma \circ F|_W)=
    \dim \frac{\mathcal{O}_{W, (0,0)}}{I_{\sigma \circ F|_W}}
    = \dim \frac{\mathcal{O}_{n-d+1}}{F^{*}(I(X \times \C)) \cdot \mathcal{O}_{n-d+1} +I(t)}=
    \dim \frac{\mathcal{O}_{n-d}}{f^{*}(I(X ))  \cdot \mathcal{O}_{n-d}},
\end{equation}
proving (a) and (b) if $f$ is an immersion.

If $f$ is not an immersion, we define a new map germ  
\begin{eqnarray}
    \widetilde{f}: (\C^{n-d} , 0) &\to& (\C^n \times \C^{n-d}, 0) \\
    \widetilde{f}(x) &=&(f(x), x)
\end{eqnarray}
 and a new space germ $\widetilde{X}=X \times \C^{n-d} \subset \C^n \times \C^{n-d}$. $\widetilde{X}$ is an irreducible Cohen--Macaulay analytic set germ,  $\widetilde{f}$ is an immersion and $\widetilde{f}$ is isolated with respect to $\widetilde{X}$. A perturbation $f_t$ of $f$ induces a perturbation $\widetilde{f}_t(x)=(f_t(x), x)$, moreover, if $f_t$ is generic with respect to $X$, then $\widetilde{f}_t$ is generic with respect to $\widetilde{X}$. Therefore, by the above argument,
 \begin{equation}
            \sharp \widetilde{f}_t^{-1}(\widetilde{X}) = \dim \frac{\mathcal{O}_{n-d}}{\widetilde{f}^*(I(\widetilde{X})) \cdot \mathcal{O}_{n-d}}
        \end{equation}
        holds. But $\sharp \widetilde{f}_t^{-1}(\widetilde{X})=\sharp f_t^{-1}(X)$ and $\widetilde{f}^*(I(\widetilde{X}))= f^*(I(X))$, hence, equation~\eqref{eq:preimages} holds, proving (a) if $f$ is not an immersion.
    
\end{proof}

\begin{remark}
This proof does not work if $X$ is not Cohen--Macaulay. For a fixed perturbation $f_t$, the number $\sharp f_t^{-1}(X)$ is the same for every nonzero value $t \neq 0$, but we cannot compare different perturbations in this case. The authors guess that the number $\sharp f_t (X)$ depends on the perturbation if $X$ is not Cohen--Macaulay. 

On the other hand, Theorem~\ref{th:pullback} suggests that the independence of $\sharp f_t (X)$ on the perturbation also holds for generically one-to-one images of a Cohen--Macaulay space germs, like $X=\Sigma$ is our case. 

Also note that in the immersion case $ \sharp f_t^{-1}(X, x_0)$ is equal to the intersection multiplicity of the analytic space germs $(X,x_0)$ and the image of $f$ (which is the zero locus of $\widetilde{\pi}$). However, in the general case we wanted to avoid investigation of the image $f$, which may be not reduced.
\end{remark}

Also if $(X, x_0)$ is not necessarily Cohen--Macaulay, we call the number of intersection points $\sharp f_t^{-1}(X,x_0)$ the \emph{multiplicity of $f$ with respect to $(X,x_0)$} at $0 \in (X,x_0)$. We call a germ of immersion $f$ (and also the corresponding germ of submersion $\pi$) \emph{generic with respect to $(X,x_0)$} if the multiplicity of $f$ with respect to $(X,x_0)$ is equal to $\text{mult}(X, 0)$, defined in \eqref{eq:multidef}.\footnote{Not to be confused with a perturbation $f_t$ of $f$ is generic (i.e., transverse) with respect to $(X,x_0)$.} In particular, if a  surjective affine linear linear map $\widetilde{\pi}: (\C^n, x_0) \to (\C^d,0)$ is generic with respect to $(X,x_0)$, assumed to be Cohen--Macaulay, and
\bean
f: (\mathbb{C}^{n-d}, 0) \to (L, x_0) \subset (\C^n, x_0)
\eean
is a linear parametrization of $L=\widetilde{\pi}^{-1}(0) \subset (\mathbb{C}^n, x_0)$, then
    \bean\label{eq:quest}
\dim \frac{\mathcal{O}_{n-d}}{f^*(I(X,x_0)) \cdot \mathcal{O}_{n-d}}= \dim \frac{\mathcal{O}_{\C^n, x_0}}{I(X,x_0)+I_{\widetilde{\pi}}}=\text{mult}(X, x_0).
\eean

\subsection{Homogeneous case}\label{ss:homo}

We show that with respect to a homogeneous variety a linear map germ is generic if and only if it is isolated.

An ideal $I \subset \C[x_1, \dots, x_n] $ is \emph{homogeneous} if it is generated by homogeneous polynomials $f_1, \dots, f_k$. Equivalently, $I$ is homogeneous if and only if $f \in I$ implies that $f^{(d)} \in I$ for every $d=0, \dots, \deg(f)$, where $f=\sum_{d=0}^{\deg(f)} f^{(d)}$ is the (unique) decomposition of $f$ as the sum of homogeneous polynomials $f^{(d)}$ of degree $d$. For details see \cite{FultonAlgebraicCurves}.

The vanishing locus $X=V(I) \subset \C^n$ of a homogeneous ideal $I$ is a cone, that is, $p \in X$ implies that $\lambda p \in X$ for all $\lambda \in \C$. Indeed, $f(\lambda p)=\lambda^d f(p)$ for a homogeneous polynomial $f$ of degree $d$. In other words, $X$ is the union of linear subspaces of $\C^n$ of dimension 1. 

Therefore, a homogeneous ideal $I$ determines a subset $\P X \subset \P^{n-1}$, where $\P^{n-1}:=\C P^{n-1}$ is the complex projective space of dimension $n-1$ consists of the linear subspaces of $\C^n$ of dimension 1. Such a subset, that is, the vanishing locus of a homogeneous ideal in $\P^{n-1}$ is called \emph{projective variety}. The dimension of $\P X$ is $d-1$.

Let $I$ be a homogeneous radical ideal, that is, $\sqrt{I}=I$, then $X$ is reduced. In this case, the \emph{degree} of $\P X$ is defined as follows. Consider projective subspaces of complementary dimension, $\P^{n-d} \subset \P^{n-1}$. A generic subspace intersects $\P X $ at non-singular points, and the intersection is transverse. For such subspaces $\P^{n-d}$, the number of intersection points is the same, it is the degree of $\P X$, that is,
\begin{equation}
    \deg (\P X)=\sharp (\P X \cap \P^{n-d})
\end{equation}
for generic projective subspaces $\P^{n-d} \subset \P^{n-1}$. By a classical result,
    $\deg (\P X)=\text{mult} (X, 0)$.

\begin{prop}\label{pr:finitegeneric}
    Let $V(I)=X^d \subset \C^n$ be  a reduced affine algebraic subspace defined by a homogeneous radical ideal $I \subset \C[x_1, \dots, x_n]$. Consider a linear projection $\widetilde{\pi}: \C^n \to \C^d$. Assume that $\pi=\widetilde{\pi}|_X: X \to \C^d$ is finite and surjective. Then $\pi$ is generic in sense that it fulfills the definition of the multiplicity, that is, $\deg(\pi)=\text{mult}(X,0)$.
\end{prop}

\begin{proof}
    Let $L:=\widetilde{\pi}^{-1}(0)$ be the linear subspace of dimension $n-d$ defined by $\widetilde{\pi}$, and $L':=\widetilde{\pi}^{-1}(y)$, where $y \in \C^d$ is a generic value close to $y_0=0$,  be one of its affine translates which has transverse intersection with $X$ at non-singular points.    
    We have to show that the number of points of $\pi^{-1}(y)=X \cap L'$ is equal to the multiplicity of $(X,0)$.

  By homogeneity, the preimages $L'_t=\widetilde{\pi}^{-1}(ty)$ also intersect $X$ transversely at smooth points (for $t \neq 0$); the intersection points for different $t$ values form complex linear subspaces of dimension one contained in $X$, spanned by the points of $X \cap L'$.  Let $\mathcal{L} \subset \C^n$ be the linear subspace of dimension $n-d+1$ spanned by $L$ and $L'$. In other words, $\mathcal{L}=\bigsqcup_{t \in \C} L'_t$. $\mathcal{L}$ intersects the non-singular part of $X$ transversely, and $X \cap \mathcal{L}$ is the union of the one dimensional  linear subspaces. 

    The projectivization
    $\P \mathcal{L} \subset \P^{n-1}$ is a projective subspace of dimension $n-d$. The intersection of $\P \mathcal{L}$ and $\P X$ is transverse, the intersection points $\P X \cap \P \mathcal{L}$ correspond to the one dimensional linear subspaces in $X \cap \mathcal{L}$. Hence, their number is equal to the degree of $\P X$, which equals to the multiplicity of $(X, 0)$, that is,
    \begin{equation}
        \sharp (X \cap L')=\sharp (\P X \cap \P \mathcal{L})= \deg (\P X)=\text{mult}(X, 0).
    \end{equation}
    This proves the proposition.
    
\end{proof}

    Then, by Propositions~\ref{pr:everypert} and \ref{pr:finitegeneric} we conclude that

    \begin{cor}\label{co:finitegeneric2}
        Let $V(I)=X^d \subset \C^n$ be  a reduced affine algebraic subspace defined by a homogeneous radical ideal $I \subset \C[x_1, \dots, x_n]$. Consider an injective linear map $l: \C^{n-d} \to \C^n$. Then,
        \begin{enumerate}[label=(\alph*)]
        \item $l$ is generic with respect to $(X,0)$ if and only if it is isolated, that is, $l^{-1}(X,0)=\{0\}$.
        \item If $l$ is isolated with respect to $(X,0)$ and $X$ is Cohen--Macaulay, then
         \bean
         \text{mult}(X, 0)=
\dim \frac{\mathcal{O}_{n-d}}{l^*(I(X,0)) \cdot \mathcal{O}_{n-d}}.
\eean
\end{enumerate}
    \end{cor}

    \begin{rem}
        Note that Proposition~\ref{pr:finitegeneric} and Corollary~\ref{co:finitegeneric2} do not hold if $X$ is not homogeneous, see Example~\ref{ex:cusp}.
    \end{rem}

\begin{rem}\label{re:alex} For such a space $X$ defined by a homogeneous ideal in $\C^n$, we can see some extra facts about the condition that $X$ be Cohen-Macaulay. Specifically, since the Cohen-Macaulay locus of a variety is open, we can see that the germ $(X, 0)$ of $X$ at the origin is Cohen-Macaulay if and only if some neighborhood of the origin in $X$ is Cohen-Macaulay; using the scaling action on $X$, we can then see that this will be true if and only if $X$ is everywhere Cohen-Macaulay, since for any point $p$ outside our neighborhood we can choose small enough $t \in \C^*$ that $t \cdot p$ is in our neighborhood.

On the other hand, we can see that the projectivized space $\P X$ is Cohen-Macaulay if and only if $X \setminus \{0\}$ is Cohen-Macaulay, since $X \setminus \{0\}$ is a $\C^*$-bundle over $\P X$ and, for a ring $R$, $R$ is Cohen-Macaulay if and only if $R[t^{\pm 1}]$ is. In particular, we see that the Cohen-Macaulayness of $X$ implies that of $\P X$, but not the reverse.
\end{rem}

\printbibliography

\end{document}